\documentclass[letterpaper,11pt]{article}
\usepackage{lmodern}
\usepackage{amsmath, amsthm, amssymb, thm-restate}
\usepackage[table,xcdraw]{xcolor}
\usepackage[ruled]{algorithm2e}
\SetArgSty{textnormal}
\usepackage{setspace}
\usepackage{mathtools}
\usepackage{algpseudocode}
\usepackage[numbers]{natbib}
\usepackage{comment} 
\usepackage{tcolorbox} 
\usepackage{xfrac}
\usepackage{multirow}
\usepackage{caption}
\usepackage{bm}
\usepackage{newfloat}
\usepackage{enumitem}
\usepackage{dblfloatfix} 
\usepackage{wrapfig}
\usepackage{dsfont}
\usepackage{mdframed}
\usepackage{multirow}
\usepackage{tikz}
\usetikzlibrary{positioning}
\usepackage{pgffor}
\usepackage{threeparttable}
\usetikzlibrary{positioning,fit}
\usepackage{fancyhdr}
\usepackage{hyperref}
\usepackage[noabbrev,nameinlink]{cleveref}
\usepackage{refcount}
\usepackage{tikz}
\usepackage{pgfplots}
\pgfplotsset{compat=1.18}

\usepackage[margin=1in]{geometry}

\allowdisplaybreaks

\definecolor{mygreen}{RGB}{10,110,230}
\definecolor{myred}{RGB}{10,110,230}

\hypersetup{
     colorlinks=true,
     citecolor= mygreen,
     linkcolor= myred
}

\newcommand{\ev}{\mathop{\mathbb{E}}} 
\newcommand{\vecx}{\Vec{x}} 
\newcommand{\y}{y} 
\newcommand{\ranking}{{\textsc{Ranking}}} 
\DeclareMathSymbol{\shortminus}{\mathbin}{AMSa}{"39}

\newcommand{\Franking}{{\textsc{FRanking}}} 
\newcommand{\Fullyranking}{{\textsc{Fully-Ranking}}} 

\newcommand{\M}{\textsc{R}}
\newcommand{\A}{\textsc{A}}

\SetKwRepeat{Do}{do}{while}

\renewcommand{\epsilon}{\ensuremath{\varepsilon}}

\let\originalleft\left
\let\originalright\right
\renewcommand{\left}{\mathopen{}\mathclose\bgroup\originalleft}
\renewcommand{\right}{\aftergroup\egroup\originalright}

\crefname{lemma}{Lemma}{Lemmas}
\crefname{theorem}{Theorem}{Theorems}
\crefname{property}{Property}{Properties}
\crefname{claim}{Claim}{Claims}
\crefname{fact}{Fact}{Facts}
\crefname{result}{Result}{Results}
\crefname{corollary}{Corollary}{Corollaries}
\crefname{definition}{Definition}{Definitions}
\crefname{observation}{Observation}{Observations}
\crefname{proposition}{Proposition}{Propositions}
\crefname{assumption}{Assumption}{Assumptions}
\crefname{line}{Line}{Lines}
\crefname{figure}{Figure}{Figures}
\creflabelformat{property}{(#1)#2#3}
\crefname{equation}{}{}
\crefname{section}{Section}{Sections}
\crefname{appendix}{Appendix}{Appendices}
\crefname{algCounter}{Algorithm}{Algorithms}

\newtheorem{lemma}{Lemma}[section]
\newtheorem{theorem}[lemma]{Theorem}

\newtheorem{corollary}[lemma]{Corollary}

\newtheorem{definition}[lemma]{Definition}
\newtheorem{claim}[lemma]{Claim}
\newtheorem{fact}[lemma]{Fact}

\newtheorem*{remark*}{Remark}

\newtheorem*{definition*}{Definition}

\makeatletter
\def\thm@space@setup{\thm@preskip= 0.2cm \thm@postskip=\thm@preskip}
\makeatother

\definecolor{mygreen}{RGB}{20,155,20}
\definecolor{myred}{RGB}{195,20,20}
\definecolor{linkcolor}{RGB}{0,0,230}
\definecolor{mylightgray}{RGB}{230,230,230}
\definecolor{verylightgray}{RGB}{240,240,240}
\definecolor{commentcolor}{RGB}{120,120,120}

\SetCommentSty{mycommfont}

\hypersetup{colorlinks=true, citecolor= mygreen, linkcolor= myred, urlcolor= mygreen}

\newcounter{myalgctr}

\newenvironment{tbox}{\par\addvspace{0.2cm}\begin{tcolorbox}[width=\textwidth, boxsep=2pt, left=1pt, right=1pt, top=4pt, boxrule=1pt, arc=0pt, colback=white, colframe=black]}{\end{tcolorbox}}

\newenvironment{tboxh}{\par\addvspace{0.2cm}\begin{tcolorbox}[width=\textwidth, boxsep=2pt, left=1pt, right=1pt, top=4pt, boxrule=1pt, arc=0pt, colback=white, colframe=black, float=t]}{\end{tcolorbox}}

\newenvironment{graytbox}{
\par\addvspace{0.1cm}
\begin{tcolorbox}[width=\textwidth,
                  boxsep=5pt,
                  left=1pt,
                  right=1pt,
                  top=2pt,
                  bottom=2pt,
                  boxrule=0pt,
                  arc=0pt,
                  colback=mylightgray,
                  colframe=black,
                  ]
}{
\end{tcolorbox}
}

\newcommand{\tboxhrule}{\vspace{0.1cm} \hrule \vspace{0.2cm}}

\newenvironment{titledtbox}[1]{\begin{tbox}#1 \tboxhrule}{\end{tbox}}
\newenvironment{titledtboxh}[1]{\begin{tboxh}#1 \tboxhrule}{\end{tboxh}}

\newcounter{protocolcounter}

\crefname{protocolcounter}{Algorithm}{Algorithms}

\makeatletter
\renewcommand{\paragraph}{%
  \@startsection{paragraph}{4}%
  {\z@}{10pt}{-1em}%
  {\normalfont\normalsize\bfseries}%
}
\makeatother

\makeatletter
\patchcmd{\@algocf@start}
  {-1.5em}
  {0pt}
  {}{}
\makeatother

\title{A Unified Framework for Analysis of \\ Randomized Greedy Matching Algorithms}

\author{
Mahsa Derakhshan \\ {\em Northeastern University} \and 
Tao Yu  \\ {\em Northeastern University}
}
\date{}

\newif\ifdoublecol
\doublecolfalse

\begin{document}
\maketitle

\setcounter{page}{1}

Randomized greedy algorithms form one of the simplest yet most effective approaches for computing approximate matchings in graphs. In this paper, we focus on the class of \emph{vertex-iterative} (VI) randomized greedy matching algorithms, which process the vertices of a graph $G=(V,E)$ in some order $\pi$ and, for each vertex $v$, greedily match it to the first available neighbor (if any) according to a preference order $\sigma(v)$. Various VI algorithms have been studied, each corresponding to a different distribution over $\pi$ and $\sigma(v)$.

We develop a unified framework for analyzing this family of algorithms and use it to obtain improved approximation ratios for \textsc{Ranking} and \textsc{FRanking}, the state-of-the-art VI randomized greedy algorithms for the random-order and adversarial-order settings, respectively. In \textsc{Ranking}, the decision order $\pi$ is drawn uniformly at random and used as the common preference order for all vertices, whereas \textsc{FRanking} uses an adversarially chosen decision order $\pi$ and a uniformly random preference order $\sigma$ shared by all vertices. We obtain an approximation ratio of $0.560$ for \textsc{Ranking}, improving on the previous best ratio of $0.5469$ by Derakhshan, Roghani, Saneian, and Yu~[SODA 2026]. For \textsc{FRanking}, we obtain a ratio of $0.539$, improving on the $0.521$ bound of Huang, Kang, Tang, Wu, Zhao, and Zhu~[JACM 2020]. These results also imply state-of-the-art approximation ratios for \emph{oblivious matching} and \emph{fully online matching} problems on general graphs.

Our analysis framework also enables us to prove improved approximation ratios for graphs with no short odd cycles. Such graphs form an intermediate class between general graphs and bipartite graphs. In particular, we show that \textsc{Ranking} is at least $0.570$-competitive for graphs that are both triangle-free and pentagon-free. For graphs whose shortest odd cycle has length at least $129$, we prove that \textsc{Ranking} is at least $0.615$-competitive.

\newpage
\tableofcontents
\newpage



\section{Introduction}
Randomized greedy matching algorithms form a broad and fundamental class of algorithms for approximating matching under uncertainty. The line of research was initiated by Dyer and Frieze~\cite{Randomized_greedy_matching}, who showed that the natural edge-iterative algorithm, which processes all edges greedily in a uniformly random order, achieves a worst-case approximation ratio of $0.5$, identical to the deterministic greedy algorithm. On the bright side, a large class of randomized greedy algorithms that are vertex-iterative have been studied and shown to achieve approximation ratios strictly greater than $0.5$~\cite{0.505_simplified_Ranking,0.521Franking,0.523ranking,0.526ranking,0.531RDO,0.546ranking,2/3MRG,randomized_greedy_matching_II}.

The vertex-iterative (VI) randomized greedy matching  algorithms follow the recipe first formalized by Aronson, Dyer, Frieze, and Suen~\cite{randomized_greedy_matching_II}: given a graph $G=(V,E)$, the algorithm processes vertices in $V$ in some order $\pi$, and for each $v$ processed, it greedily matches its first available neighbor, if any, according to a preference order $\sigma(v)$. Different VI randomized greedy algorithms correspond to different distributions of $\pi$ and $\sigma(v)$. A summary of the state-of-the-art results is given in \Cref{tab:sota_matching_ratios}.

The VI randomized greedy framework naturally applies to the classical online matching problem under the vertex arrival model. In this setting, vertices arrive sequentially; upon arrival, all their neighbors are revealed\footnote{\label{footnote:fullyonline}See a detailed discussion of the fully online model and the neighbor-revealing assumption in \cref{Fully-Online_Section}.}. Each vertex must be matched to one of its neighbors immediately, without knowledge of future arrivals. The VI randomized greedy framework captures this process by interpreting the arrival order as the processing order $\pi$ and the preference order $\sigma(v)$ as the algorithm’s decision rule upon the arrival of each vertex. For example, the celebrated \ranking{} algorithm~\cite{karp1990} can be viewed as a VI randomized greedy algorithm operating in the online bipartite matching model with one-sided vertex arrivals.

By this connection, we can roughly classify VI randomized greedy algorithms into two major categories, depending on whether the process (arrival) order $\pi$ is uniformly at random or adversarial. We call these two models the uniform random vertex arrival model and the adversarial (fully online) vertex arrival model, respectively. In this paper, we study the \ranking{} and the \Franking{} algorithms, corresponding respectively to the best state-of-the-art VI randomized greedy algorithms for these two models.

\vspace{ -2 mm}
\paragraph{Uniform Random Vertex Arrival Model.}
In the uniform vertex arrival model~\cite{randomized_greedy_matching_II,2/3MRG,0.531RDO,flawedranking,0.505_simplified_Ranking,0.523ranking,0.526ranking,0.546ranking}, $\pi$ is a uniform random permutation over $V$ and the algorithm can choose $\sigma(v)$ for each $v$. Choosing each $\sigma(v)$ as a fixed preference order (as in RDO) or an independent random permutation (as in MRG) both yield a state-of-the-art $0.531$ approximation ratio, by Tang, Wu, and Zhang~\cite{0.531RDO}. Setting $\sigma(v)=\pi$ results in the \ranking{} algorithm for general graphs, which achieves an approximation ratio of $0.5469$, proved recently by Derakhshan, Roghani, Saneian, and Yu~\cite{0.546ranking}, with an upper bound of $0.727$~\cite{0.727Ranking}. 

\vspace{ -2 mm}
\paragraph{Adversarial Vertex Arrival Model.}
In the adversarial arrival (fully online) model, $\pi$ is an arbitrary fixed ordering. If each $\sigma(v)$ is chosen independently and uniformly at random, the resulting algorithm IRP achieves a worst-case ratio of $0.5 + o(1)$ \cite{Randomized_greedy_matching}. When all $\sigma(v)$ share the same uniform random permutation, we obtain \Franking{}, whose current approximation ratio lies between $0.521$ and $0.5671$ by Huang, Kang, Tang, Wu, Zhao, and Zhu~\cite{0.521Franking}.

\begin{table*}[h]
    \centering
\captionsetup{font=footnotesize}
    \caption{Approximation Ratios of VI Randomized Greedy Matching Algorithms on General Graphs\protect\footnotemark }
    
    \label{tab:sota_matching_ratios}
    \begin{tabular}{|l|c|c|c|c|}
        \hline
        \textbf{Algorithm} & \textbf{Decision $(\pi)$} & \textbf{Preference $(\sigma)$} & \textbf{Lower Bound} & \textbf{Upper Bound} \\
        \hline
        Greedy & Adversarial & Adversarial & $0.5$ & $0.5$ \\
        \hline
        \Franking{} & Adversarial & Common Random & $0.521$~\cite{0.521Franking} $\to$ $\mathbf{0.539}$ & $0.5671$~\cite{0.521Franking} \\
        \hline
        IRP & Adversarial & Independent Random & $0.5$ & $0.5$~\cite{Randomized_greedy_matching}\\
        \hline
        RDO & Random & Adversarial & $0.531$~\cite{0.531RDO} & $0.625$~\cite{0.531RDO}\\
        \hline
        \ranking{} & \multicolumn{2}{c|}{Common Random} & $0.5469$~\cite{0.546ranking} $\to$ $\mathbf{0.560}$ & $0.727$~\cite{0.727Ranking} \\
        \hline
        UUR & Random & Common Random & $0.531$~\cite{0.531RDO} $\to$ $\mathbf{0.539}$ & $0.75$~\cite{flawedranking} \\
        \hline
        MRG & \multicolumn{2}{c|}{Independent Random} & $0.531$~\cite{0.531RDO} & $0.666$~\cite{2/3MRG} \\
        \hline

    \end{tabular}
    \begin{tablenotes}
    \centering
    \footnotesize
    \item \textbf{Note:} Our results are shown in bold numbers.
  \end{tablenotes}
\end{table*}

\footnotetext{See a short discussion of various VI randomized greedy matching algorithms in \cref{Sec:various_VI_algorithms}.}

\paragraph{Query-Commit Problem.}\label{query-commit-view-intro} All of the above algorithms can be viewed as randomized vertex-iterative algorithms for the query-commit matching problem. In the query-commit matching (oblivious matching) problem, introduced by~\cite{Randomized_greedy_matching}, and subsequently studied in~\cite{flawedranking,0.523ranking,0.526ranking,0.531RDO,0.546ranking,0.521Franking}, the algorithm does not have prior information about the edge set $E$ and must select an ordered list $L$ (possibly randomized) of edges to query. Whenever a queried edge $e$ exists with both ends available, the algorithm is forced to include that edge in the matching. The objective remains to maximize the size of the produced matching relative to the maximum matching in $G$. Although $L$ is not restricted to be vertex-iterative (i.e., the algorithm is not required to query all edges incident to a vertex consecutively), the current best approximation ratio for this problem is $0.5469$, achieved by the vertex-iterative algorithm \ranking{}~\cite{0.546ranking}, with an upper bound of $0.7583$~\cite{edgeweighted0.659}.

\paragraph{Randomized Greedy Matching on Bipartite Graphs.} Randomized greedy matching is much better understood on bipartite graphs. In the online bipartite matching with one-sided adversarial
vertex arrival model, Karp, Vazirani, and Vazirani~\cite{karp1990} show that \ranking{} achieves an optimal $1-\frac{1}{e}$ approximation ratio. Under uniform random vertex arrival, \ranking{} achieves a $0.696$ approximation ratio on bipartite graphs~\cite{0.696Ranking_bipartite}, which is close to its upper bound $0.727$~\cite{0.727Ranking}. Compared to bipartite graphs, the best approximation ratio achieved by any randomized greedy matching algorithm on general graphs, despite extensive efforts~\cite{0.521Franking,0.526ranking,0.531RDO,0.546ranking,randomized_greedy_matching_II,flawedranking}, remains as small as $0.5469$.

\paragraph{Large Odd Girth Graphs.}
The sharp contrast between bipartite and general graphs suggests studying intermediate graph classes. A natural candidate is to classify graphs by their odd girth, namely the length of the shortest odd cycle. As the odd girth of a graph grows, the graph becomes increasingly bipartite-like at a local scale. For example, if a graph has odd girth $k$, then any subgraph with fewer than $k$ vertices has to be bipartite. To the best of our knowledge, this perspective has not been systematically explored in the analysis of randomized greedy matching on general graphs.

\subsection{Our Contribution} In this paper, we show that \ranking{} and \Franking{} achieve approximation ratios of $0.560$ and $0.539$, improving from $0.5469$ \cite{0.546ranking} and $0.521$ \cite{0.521Franking}, respectively. 
\begin{graytbox}
\begin{restatable}{theorem}{RankingGeneral}
The approximation ratio for \ranking{} is at least $0.560$ for general graphs.
\end{restatable}
\end{graytbox}

\begin{graytbox}
\begin{restatable}{theorem}{FRankingGeneral}
    The approximation ratio for \Franking{} is at least $0.539$ for general graphs.
\end{restatable}
\end{graytbox}
Our result for \ranking{} applies to both the randomized greedy matching problem and the query-commit matching problem on general graphs, improving the approximation ratio from $0.5469$~\cite{0.546ranking} to $0.560$.

Our result for \Franking{} improves the state-of-the-art approximation ratio for the fully online matching problem\footref{footnote:fullyonline} on general graphs. Moreover, our bound also extends to the UUR algorithm, improving upon the $0.531$ ratio established through the RDO analysis in~\cite{0.531RDO}.

Our analysis builds on two notions introduced separately in prior work: \emph{victims}~\cite{0.521Franking} and \emph{backups}~\cite{0.546ranking}. A central technical contribution of this paper is to show that these notions are intrinsically connected: both arise from the alternating path structure underlying all randomized greedy matching algorithms. The resulting framework (\cref{sec:Alt_path_backup_victim}) hence applies broadly to randomized greedy algorithms beyond vertex-iterative ones, and may be useful in other settings, including weighted variants.

This alternating-path viewpoint also leads to a new odd-girth-sensitive analysis: by relating alternating paths to odd cycles, we initiate the study of randomized greedy matching algorithms parameterized by odd girth, and show that \ranking{} achieves stronger approximation ratios as the odd girth increases\footnote{See \Cref{tab:fn_rectangular4} for results for other odd girths.}.
\begin{graytbox}
\begin{theorem}
The \textsc{Ranking} algorithm achieves an approximation ratio of at least $0.562$ on triangle-free graphs, at least $0.570$ on graphs that are both triangle-free and pentagon-free, and at least $0.615$ on graphs of odd girth at least $129$.
\end{theorem}
\end{graytbox}

\subsection{Further Related Works}
There has also been a long line of research on randomized greedy matching algorithms on vertex/edge weighted variants. We briefly discuss these works and their results.

\noindent\textbf{Weighted General Graphs.}
Weighted variants of randomized greedy matching algorithms in general graphs remain difficult. \cite{0.526ranking} showed that there exists an algorithm for the vertex-weighted randomized greedy matching problem that achieves a $0.5015$ approximation ratio. For the edge-weighted variant, \cite{0.531RDO} proved that there exists an algorithm achieving a $0.5014$ approximation ratio. The upper bound for both cases, on the other hand, remains $0.7583$ \cite{edgeweighted0.659}, the same as that for the unweighted case.

\noindent\textbf{Weighted Bipartite Graphs.} 
The vertex-weighted variant of \ranking{} has been proven to achieve a state-of-the-art $0.686$ approximation ratio through a line of work \cite{VertexWeightedBipartite1,VertexWeightedBipartite2,VertexWeightedBipartite3}, narrowing the gap toward the unweighted $0.696$ bound. For the edge-weighted case, \cite{edgeweighted0.659} showed an algorithm achieving a $0.659$ approximation ratio. Notably, the algorithm by~\cite{edgeweighted0.659} is equivalent to \ranking{} when all edges assume unit weights.

\section{Our Techniques}
To provide a general framework for analyzing randomized greedy matching algorithms, we view each algorithm as operating on an ordered list $L$ of edges (or, more precisely, on a total ordering of vertex pairs) that it processes greedily. Different randomized greedy algorithms correspond to different distributions over such orderings $L$. For example, in the case of \ranking{}, we have $L = \pi \times \pi$, where $\pi$ is a uniformly random permutation of the vertices. We use $\M(L)$ to denote the matching produced by running the greedy algorithm on $L$ (see \cref{def:greedyalg} for a formal definition).

Our analysis of both \ranking{} and \Franking{} builds on the standard primal--dual framework widely used to study randomized greedy matching algorithms (see \cite{RandomizedPrimalDual,Economicbasedanalysis,0.521Franking,0.5671optimalbipartite,0.531RDO,0.546ranking} and references therein). For each edge added to the matching by the algorithm, one unit of total \emph{gain} is distributed among its endpoints and a potential third vertex. Let $\text{gain}(v)$ denote the gain assigned to vertex $v$. To establish the approximation ratio, it suffices to show that for every edge $(u, u^\ast)$ in a maximum matching of $G$, the expected total gain satisfies
\[\ev\left[\text{gain}(u) + \text{gain}(u^\ast)\right] \ge \alpha,\]
for some constant $\alpha$, which directly implies the same lower bound on the algorithm’s approximation ratio.

The expectation is typically taken in two stages. First, we fix an arbitrary ordering of all edges except those incident to $u^\ast$, denoted by $L_{\shortminus u^\ast}$, and compute the expectation over the randomization associated with $u^\ast$ while keeping $L_{\shortminus u^\ast}$ fixed. Then we take the outer expectation over $L_{\shortminus u^\ast}$. Formally, the randomized primal--dual method evaluates
 $$\ev_{L_{\shortminus u^\ast}}\left[\; \ev_{\text{$u^\ast$}}\;\left[\text{gain$(u)$ $+$ gain$(u^\ast)$}\right]\right]$$

In the bipartite analysis of \ranking{}, an important monotonicity property holds: if vertex $u$ is matched under $L_{\shortminus u^\ast}$, it remains matched under the full ordering $L$, regardless of the randomization involving $u^\ast$. Leveraging this fact, along with a carefully designed gain-sharing rule, \cite{RandomizedPrimalDual,Economicbasedanalysis} established the optimal $1 - \tfrac{1}{e}$ approximation ratio for \ranking{} in bipartite graphs. However, this property fails in general graphs, which is the central obstacle in extending such analyses. In particular, a \emph{bad event} may occur when the arrival of $u^\ast$ causes a previously matched neighbor $u$ to become unmatched. If this phenomenon were to happen systematically, then no randomized greedy algorithm could surpass the trivial $1/2$-approximation guarantee.

\paragraph{Backups, Victims and Blockers.} To address this challenge in general graphs, two key concepts---\emph{backup}~\cite{0.546ranking} and \emph{victim}~\cite{0.521Franking}---were introduced in the literature for the analyses of \ranking{} and \Franking{}, respectively. In this work, we build on these notions, highlight their complementary roles, uncover new structural properties about them, and generalize both concepts to apply to randomized greedy matching algorithms beyond their original, model-specific settings.

\medskip\noindent
The concept of a \emph{backup} was introduced by~\cite{0.546ranking} in the analysis of \ranking{}. Informally, if a vertex $u$ has a backup (with respect to a given ordering), it is the vertex to which $u$ would be matched if its current match were removed from the ordering. The analysis in~\cite{0.546ranking} shows that when $u$ has a backup $b$ with respect to $L_{\shortminus u^\ast}$, the arrival of $u^\ast$ cannot cause $u$ to become unmatched. Moreover, they argue that if the probability of $u$ not having a backup is high, then the probability that $u$ and $u^\ast$ are simultaneously matched is also high. Hence, the notion of a backup serves to rule out the systemic occurrence of the bad event described earlier.

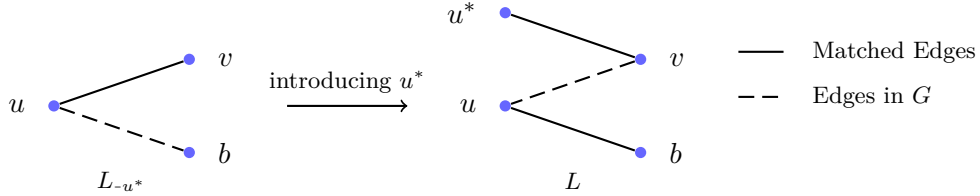
\ifdoublecol
\begin{tikzpicture}[
  even/.style={circle, fill=blue!60!white, inner sep=1.5pt},
  odd/.style={circle, fill=blue!60!white, inner sep=1.5pt},
  legendtext/.style={font=\footnotesize, align=justify},
]

\def\xLR{0.72}
\def\yStep{0.62}
\def\sep{4.8}

\def\textW{9.0}
\def\legTextW{2.2}

\pgfmathsetmacro{\xA}{-\sep/2}
\pgfmathsetmacro{\xB}{\sep/2}

\pgfmathsetmacro{\sysLeft}{\xA - \xLR - 1.2}
\pgfmathsetmacro{\sysRight}{\xB + \xLR + 1.2}
\pgfmathsetmacro{\sysCenter}{(\sysLeft+\sysRight)/2}

\begin{scope}

\foreach \i in {0,...,2} {
  \pgfmathsetmacro{\y}{-(\yStep)*\i}
  \ifodd\i
    \node[odd] (gA\i) at ({\xA-\xLR},\y) {};
    \node[anchor=east] at ({\xA-\xLR-0.25},\y) {$u$};
  \else
    \node[even] (gA\i) at ({\xA+\xLR},\y) {};
    \ifnum\i=0
      \node[anchor=west] at ({\xA+\xLR+0.25},\y) {$v$};
    \else
      \node[anchor=west] at ({\xA+\xLR+0.25},\y) {$b$};
    \fi
  \fi
}

\draw[thick] (gA0) -- (gA1);
\draw[dash pattern=on 6pt off 3pt, thick] (gA1) -- (gA2);

\node[font=\footnotesize] at (\xA, {-(\yStep)*2 - 0.38}) {$L_{\shortminus u^*}$};

\def\Uzero{$u^*$}
\def\Uone{$v$}
\def\Utwo{$u$}
\def\Uthree{$b$}

\foreach \i in {0,...,3} {
  \pgfmathsetmacro{\y}{-(\yStep)*\i + (\yStep)}

  \ifodd\i
    \node[even] (gB\i) at ({\xB+\xLR},\y) {};
    \ifnum\i=1
      \node[anchor=west] at ({\xB+\xLR+0.25},\y) {\Uone};
    \fi
    \ifnum\i=3
      \node[anchor=west] at ({\xB+\xLR+0.25},\y) {\Uthree};
    \fi
  \else
    \node[odd] (gB\i) at ({\xB-\xLR},\y) {};
    \ifnum\i=0
      \node[anchor=east] at ({\xB-\xLR-0.25},\y) {\Uzero};
    \fi
    \ifnum\i=2
      \node[anchor=east] at ({\xB-\xLR-0.25},\y) {\Utwo};
    \fi
  \fi
}

\draw[thick] (gB0) -- (gB1);
\draw[dash pattern=on 6pt off 3pt, thick] (gB1) -- (gB2);
\draw[thick] (gB2) -- (gB3);

\node[font=\footnotesize] at (\xB, {-(\yStep)*3 + (\yStep) - 0.38}) {$L$};

\pgfmathsetmacro{\yArrow}{-1*\yStep}

\draw[->, thick]
  ({\xA+\xLR+1.3}, \yArrow) -- ({\xB-\xLR-1.3}, \yArrow);

\node[font=\footnotesize, align=center]
  at ({(\xA+\xB)/2}, {\yArrow+0.32})
  {introducing $u^*$};

\pgfmathsetmacro{\xLegText}{\xA+\xLR+1.3}   
\pgfmathsetmacro{\xLegLineL}{\xLegText-0.6}
\pgfmathsetmacro{\xLegLineR}{\xLegText-0.1}
\pgfmathsetmacro{\yLegTop}{-1.75}
\pgfmathsetmacro{\legSep}{0.38}             

\draw[thick] (\xLegLineL,\yLegTop) -- (\xLegLineR,\yLegTop);
\node[legendtext, anchor=west, font=\footnotesize] at (\xLegText,\yLegTop) {Matched Edges};

\draw[dash pattern=on 6pt off 3pt, thick, font=\footnotesize]
  (\xLegLineL,{\yLegTop-\legSep}) -- (\xLegLineR,{\yLegTop-\legSep});
\node[legendtext, anchor=west] at (\xLegText,{\yLegTop-\legSep}) {Edges in $G$};

\pgfmathsetmacro{\xTextLeft}{\xA-\xLR-0.2}
\pgfmathsetmacro{\yText}{-2.3}

\node[anchor=north west, text width=7cm, align=left, font=\footnotesize]
  at (\xTextLeft,\yText)
{An example of how backup works. In $L_{\shortminus u^*}$, $u$ is matched to $v$ and has a backup $b$. By introducing $u^*$, the match of $u$ could be taken away. But since $u$ has a backup $b$, $u$ is still matched.};

\end{scope}
\end{tikzpicture}
\else
\begin{tikzpicture}[
  even/.style={circle, fill=blue!60!white, inner sep=1.5pt},
  odd/.style={circle, fill=blue!60!white, inner sep=1.5pt},
  legendtext/.style={font=\footnotesize, align=left},
]

\def\xLR{0.9}
\def\yStep{0.62}
\def\sep{6}
\def\legGap{2.2}
\def\legTextW{3.4}

\pgfmathsetmacro{\xA}{-\sep/2}
\pgfmathsetmacro{\xB}{\sep/2}
\pgfmathsetmacro{\xL}{\xB + \legGap}

\pgfmathsetmacro{\sysLeft}{\xA - \xLR - 2}
\pgfmathsetmacro{\sysRight}{\xL + \legTextW + 2}
\pgfmathsetmacro{\sysCenter}{(\sysLeft+\sysRight)/2}
\pgfmathsetmacro{\textWidth}{\sysRight - \sysLeft}

\begin{scope}[shift={({-\sysCenter},-2)}]

\foreach \i in {0,...,2} {
  \pgfmathsetmacro{\y}{-(\yStep)*\i}
  \ifodd\i
    \node[odd] (gA\i) at ({\xA-\xLR},\y) {};
    \node[anchor=east] at ({\xA-\xLR-0.25},\y) {$u$};
  \else
    \node[even] (gA\i) at ({\xA+\xLR},\y) {};
    \ifnum\i=0
      \node[anchor=west] at ({\xA+\xLR+0.25},\y) {$v$};
    \else
      \node[anchor=west] at ({\xA+\xLR+0.25},\y) {$b$};
    \fi
  \fi
}

\draw[ thick] (gA0) -- (gA1);
\draw[dash pattern=on 6pt off 3pt, thick] (gA1) -- (gA2);

\node[font=\footnotesize] at (\xA, {-(\yStep)*2 - 0.38}) {$L_{\shortminus u^*}$};


\def\Uzero{$u^*$}
\def\Uone{$v$}
\def\Utwo{$u$}
\def\Uthree{$b$}

\foreach \i in {0,...,3} {
  \pgfmathsetmacro{\y}{-(\yStep)*\i + (\yStep)}

  \ifodd\i
    \node[even] (gB\i) at ({\xB+\xLR},\y) {};
    \ifnum\i=1
      \node[anchor=west] at ({\xB+\xLR+0.25},\y) {\Uone};
    \fi
    \ifnum\i=3
      \node[anchor=west] at ({\xB+\xLR+0.25},\y) {\Uthree};
    \fi
  \else
    \node[odd] (gB\i) at ({\xB-\xLR},\y) {};
    \ifnum\i=0
      \node[anchor=east] at ({\xB-\xLR-0.25},\y) {\Uzero};
    \fi
    \ifnum\i=2
      \node[anchor=east] at ({\xB-\xLR-0.25},\y) {\Utwo};
    \fi
  \fi
}

\draw[ thick] (gB0) -- (gB1);
\draw[dash pattern=on 6pt off 3pt, thick] (gB1) -- (gB2);
\draw[ thick] (gB2) -- (gB3);

\node[font=\footnotesize] at (\xB, {-(\yStep)*3 + (\yStep) - 0.38}) {$L$};

\pgfmathsetmacro{\yArrow}{-1*\yStep}

\draw[->, thick]
  ({\xA+\xLR+1.3}, \yArrow) -- ({\xB-\xLR-1.3}, \yArrow);

\node[font=\footnotesize, align=center]
  at ({(\xA+\xB)/2}, {\yArrow+0.32})
  {introducing $u^*$};

\pgfmathsetmacro{\yL}{0.10}

\draw[thick] (\xL,\yL) -- ({\xL+0.6},\yL);
\node[legendtext, anchor=west, text width=\legTextW cm]
  at ({\xL+0.85},\yL) {Matched Edges};

\draw[dash pattern=on 6pt off 3pt, thick] (\xL,{\yL-0.6}) -- ({\xL+0.6},{\yL-0.6});
\node[legendtext, anchor=west, text width=\legTextW cm]
  at ({\xL+0.85},{\yL-0.6}) {Edges in $G$};

\pgfmathsetmacro{\yText}{-2*\yStep - 1} 
\node[anchor=north west, text width=\textWidth cm, font=\footnotesize] at (\sysLeft,\yText)
{An example of how backup works. In $L_{\shortminus u^*}$, $u$ is matched to $v$ and has a backup $b$. By introducing $u^*$, \quad\quad\quad\quad\quad the match of $u$ could be taken away. But since $u$ has a backup $b$, $u$ is still matched.};
\end{scope}
\end{tikzpicture}
\fi

\noindent The notion of a \emph{victim} was first introduced by~\cite{0.521Franking} in the analysis of \Franking{} and later adopted for RDO~\cite{0.531RDO}. Intuitively, when a vertex $w$ has a victim $u$, vertex $w$ in some sense takes away $u$’s opportunity to be matched---that is, if $w$ were removed, $u$ would become matched. Accordingly, in~\cite{0.521Franking,0.531RDO}, $w$ may transfer part of its gain to its victim $u$; this transferred gain is referred to as \emph{compensation}. By showing that vertices are more likely to receive compensation than to pay it, both works obtain approximation ratios strictly greater than $0.5$. In this work, we refer to the counterpart of a victim as a \emph{blocker}; that is, in the above example, $w$ is a blocker for $u$. We note that a victim may have multiple blockers, which plays a crucial role in our analysis: most of the time, a victim has at least two blockers and therefore receives more compensation than it pays. In graphs of large odd girth, this phenomenon is amplified: a victim will have at least $k/2$ blockers in graphs of odd girth $\geq 2k+1$. This allows us to collect compensations proportional to the length of the shortest odd cycle, resulting in improved approximation ratios for graphs of larger odd girth.

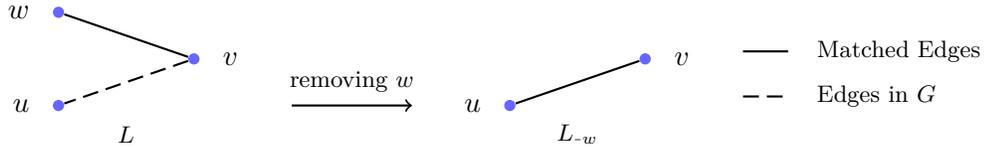
\ifdoublecol
\begin{tikzpicture}[
  even/.style={circle, fill=blue!60!white, inner sep=1.5pt},
  odd/.style={circle, fill=blue!60!white, inner sep=1.5pt},
  legendtext/.style={font=\footnotesize, align=justify},
]

\def\xLR{0.72}
\def\yStep{0.62}
\def\sep{4.8}

\def\textW{7cm}

\pgfmathsetmacro{\xA}{-\sep/2}
\pgfmathsetmacro{\xB}{\sep/2}

\begin{scope}

\def\Uzero{$w$}
\def\Uone{$v$}
\def\Utwo{$u$}

\foreach \i in {0,...,2} {
  \pgfmathsetmacro{\y}{-(\yStep)*\i + (\yStep)}

  \ifodd\i
    \node[even] (gL\i) at ({\xA+\xLR},\y) {};
    \ifnum\i=1
      \node[anchor=west] at ({\xA+\xLR+0.25},\y) {\Uone};
    \fi
  \else
    \node[odd] (gL\i) at ({\xA-\xLR},\y) {};
    \ifnum\i=0
      \node[anchor=east] at ({\xA-\xLR-0.25},\y) {\Uzero};
    \fi
    \ifnum\i=2
      \node[anchor=east] at ({\xA-\xLR-0.25},\y) {\Utwo};
    \fi
  \fi
}

\draw[thick] (gL0) -- (gL1);
\draw[dash pattern=on 6pt off 3pt, thick] (gL1) -- (gL2);

\node[font=\footnotesize] at (\xA, {-(\yStep)*2 + (\yStep) - 0.38}) {$L$};

\foreach \i in {0,...,1} {
  \pgfmathsetmacro{\y}{-(\yStep)*\i}
  \ifodd\i
    \node[odd] (gR\i) at ({\xB-\xLR},\y) {};
    \node[anchor=east] at ({\xB-\xLR-0.25},\y) {$u$};
  \else
    \node[even] (gR\i) at ({\xB+\xLR},\y) {};
    \node[anchor=west] at ({\xB+\xLR+0.25},\y) {$v$};
  \fi
}

\draw[thick] (gR0) -- (gR1);

\node[font=\footnotesize] at (\xB, {-(\yStep)*1 - 0.38}) {$L_{\shortminus w}$};

\pgfmathsetmacro{\yArrow}{-1*\yStep}

\draw[->, thick]
  ({\xA+\xLR+1.3}, \yArrow) -- ({\xB-\xLR-1.3}, \yArrow);

\node[font=\footnotesize, align=center]
  at ({(\xA+\xB)/2}, {\yArrow+0.32})
  {removing $w$};

\pgfmathsetmacro{\xLegText}{\xA+\xLR+1.3}   
\pgfmathsetmacro{\xLegLineL}{\xLegText-0.6}
\pgfmathsetmacro{\xLegLineR}{\xLegText-0.1}
\pgfmathsetmacro{\yLegTop}{-1.2}
\pgfmathsetmacro{\legSep}{0.38}

\draw[thick] (\xLegLineL,\yLegTop) -- (\xLegLineR,\yLegTop);
\node[legendtext, anchor=west, font=\footnotesize]
  at (\xLegText,\yLegTop) {Matched Edges};

\draw[dash pattern=on 6pt off 3pt, thick]
  (\xLegLineL,{\yLegTop-\legSep}) -- (\xLegLineR,{\yLegTop-\legSep});
\node[legendtext, anchor=west, font=\footnotesize]
  at (\xLegText,{\yLegTop-\legSep}) {Edges in $G$};

\pgfmathsetmacro{\xTextLeft}{\xA-\xLR-0.2}
\pgfmathsetmacro{\yText}{-1.75}

\node[anchor=north west, text width=\textW, align=left, font=\footnotesize]
  at (\xTextLeft,\yText)
{An example of a blocker and victim pair $w,u$. The victim $u$ is unmatched in $L$. When the blocker vertex $w$ is removed, $u$ becomes matched.};

\end{scope}
\end{tikzpicture}
\else
\begin{tikzpicture}[
  even/.style={circle, fill=blue!60!white, inner sep=1.5pt},
  odd/.style={circle, fill=blue!60!white, inner sep=1.5pt},
  legendtext/.style={font=\footnotesize, align=left},
]

\def\xLR{0.9}
\def\yStep{0.62}
\def\sep{6}
\def\legGap{2.2}
\def\legTextW{3.4}
\def\graphShift{0.5} 

\pgfmathsetmacro{\xA}{-\sep/2}     
\pgfmathsetmacro{\xB}{\sep/2}      
\pgfmathsetmacro{\xL}{\xB + \legGap}

\pgfmathsetmacro{\sysLeft}{\xA - \xLR - 2}
\pgfmathsetmacro{\sysRight}{\xL + \legTextW + 2}
\pgfmathsetmacro{\sysCenter}{(\sysLeft+\sysRight)/2}
\pgfmathsetmacro{\textWidth}{\sysRight - \sysLeft}

\begin{scope}[shift={({-\sysCenter},0)}]

\begin{scope}[xshift=\graphShift cm]

\def\Uzero{$w$}
\def\Uone{$v$}
\def\Utwo{$u$}

\foreach \i in {0,...,2} {
  \pgfmathsetmacro{\y}{-(\yStep)*\i + (\yStep)}

  \ifodd\i
    \node[even] (gL\i) at ({\xA+\xLR},\y) {};
    \ifnum\i=1
      \node[anchor=west] at ({\xA+\xLR+0.25},\y) {\Uone};
    \fi
  \else
    \node[odd] (gL\i) at ({\xA-\xLR},\y) {};
    \ifnum\i=0
      \node[anchor=east] at ({\xA-\xLR-0.25},\y) {\Uzero};
    \fi
    \ifnum\i=2
      \node[anchor=east] at ({\xA-\xLR-0.25},\y) {\Utwo};
    \fi
  \fi
}

\draw[thick] (gL0) -- (gL1);
\draw[dash pattern=on 6pt off 3pt, thick] (gL1) -- (gL2);

\node[font=\footnotesize] at (\xA, {-(\yStep)*2 + (\yStep) - 0.38}) {$L$};

\foreach \i in {0,...,1} {
  \pgfmathsetmacro{\y}{-(\yStep)*\i}
  \ifodd\i
    \node[odd] (gR\i) at ({\xB-\xLR},\y) {};
    \node[anchor=east] at ({\xB-\xLR-0.25},\y) {$u$};
  \else
    \node[even] (gR\i) at ({\xB+\xLR},\y) {};
    \node[anchor=west] at ({\xB+\xLR+0.25},\y) {$v$};
  \fi
}

\draw[thick] (gR0) -- (gR1);

\node[font=\footnotesize] at (\xB, {-(\yStep)*1 - 0.38}) {$L_{\shortminus w}$};

\pgfmathsetmacro{\yArrow}{-1*\yStep}

\draw[->, thick]
  ({\xA+\xLR+1.3}, \yArrow) -- ({\xB-\xLR-1.3}, \yArrow);

\node[font=\footnotesize, align=center]
  at ({(\xA+\xB)/2}, {\yArrow+0.32})
  {removing $w$};

\pgfmathsetmacro{\yL}{0.10}

\draw[thick] (\xL,\yL) -- ({\xL+0.6},\yL);
\node[legendtext, anchor=west, text width=\legTextW cm]
  at ({\xL+0.85},\yL) {Matched Edges};

\draw[dash pattern=on 6pt off 3pt, thick] (\xL,{\yL-0.6}) -- ({\xL+0.6},{\yL-0.6});
\node[legendtext, anchor=west, text width=\legTextW cm]
  at ({\xL+0.85},{\yL-0.6}) {Edges in $G$};

\end{scope} 

\pgfmathsetmacro{\yText}{-2*\yStep }

\node[anchor=north west, text width=\textWidth cm, font=\footnotesize]
  at (\sysLeft,\yText)
{An example of a blocker and victim pair $w,u$. The victim $u$ is unmatched in $L$. When the blocker vertex \newline $w$ is removed, $u$ becomes matched.};

\end{scope}

\end{tikzpicture}
\fi

Our most crucial observations regarding the relationship between the notions of backup and victim emerge when examining the symmetric difference of the matchings $\M(L)$ and $\M(L_{\shortminus u^\ast})$, namely, the alternating path. As suggested by its name, the alternating path is a path that starts with vertex $u^\ast$, and alternates between edges of $\M(L)$ and $\M(L_{\shortminus u^\ast})$.  We formally state this and more useful facts about the alternating path in \cref{lem:alt-path}. Let $P=(u_0, \dots, u_k)$ represent the alternating path with $u_0= u^*$. We next highlight some crucial properties of the vertices on this path.

\begin{enumerate}

\item Any even-indexed vertex $u_{i}$ in the path other than $u^*$ becomes worse off after the introduction of $u^*$, that is, $u_{i}$ is matched via a later edge or is unmatched in $\M(L)$ compared to $\M(L_{\shortminus u^\ast})$. Furthermore, its backup with respect to $L_{\shortminus u^\ast}$ uniquely determines this outcome. I.e., $u_{i}$'s
match (if it exists) in $\M(L)$ is exactly its backup with respect to $L_{\shortminus u^\ast}$.

 \item  Suppose $u$ is unmatched by introducing $u^\ast$ to $L_{\shortminus u^\ast}$, which implies $u_k=u$. Then $u$ is the victim of vertices $u_i$ in $L$ for all even $i<k$, which means $u$ has $k/2$ blockers in the case of this bad event and can receive up to $k/2$ copies of compensation.

 \item The first two properties imply that longer alternating paths are more desirable, as they indicate that there are more vertices with backups and that if $u$ is victimized, then it could receive more than one copy of compensation from other vertices. We use this property for $u$ to collect more compensation than it pays in expectation. For example, in \ranking{}, except for a specific case, the alternating path has length $\geq 4$. Consequently, when $u$ is unmatched, it receives at least two copies of compensation, while when it is matched, it pays at most one copy. 

 \item We also formulate a relation between the backup of $u$ with respect to $L_{\shortminus u^\ast}$ and the existence of blockers for $u^\ast$ if $u^\ast$ is unmatched. (See \cref{lem:victim_querycommit_case_u_star}.) We prove that the worse the backup $b$ is\footnote{The worse $b$ is, the later in time $(u,b)$ gets considered in the randomized greedy matching process. See a formal description in \cref{sec:alt-lemma}.}, the more compensation $u^\ast$ will get from its blockers in expectation.
\end{enumerate}

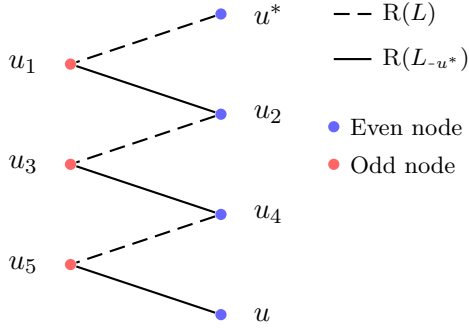
\ifdoublecol

\def\FigShift{-0.5cm}    
\def\picShift{0cm}    
\def\textShift{-1cm}   
\def\textY{-4.5}     
\def\textW{6.5cm}     

\makebox[\linewidth][c]{%
\hspace*{\FigShift}%
\begin{tikzpicture}[
    even/.style={circle, fill=blue!60!white, inner sep=1.5pt},
    odd/.style={circle, fill=red!60!white, inner sep=1.5pt},
    textbox/.style={align=justify, text width=\textW, font=\footnotesize}
]

\begin{scope}[xshift=\picShift]

\foreach \i in {0,...,6} {
    \pgfmathsetmacro{\y}{-(2/3)*\i}
    \ifodd\i
        \node[odd] (u\i) at (-1,\y) {};
        \ifnum\i=0
            \node[anchor=east] at (-1.3,\y) {$u^*$};
        \else
            \node[anchor=east] at (-1.3,\y) {$u_\i$};
        \fi
    \else
        \node[even] (u\i) at (1,\y) {};
        \ifnum\i=0
            \node[anchor=west] at (1.3,\y) {$u^*$};
        \else
            \ifnum\i=6
                \node[anchor=west] at (1.3,\y) {$u$};
            \else
                \node[anchor=west] at (1.3,\y) {$u_\i$};
            \fi
        \fi
    \fi
}

\draw[dash pattern=on 6pt off 3pt, thick] (u0) -- (u1);
\draw[thick] (u1) -- (u2);
\draw[dash pattern=on 6pt off 3pt, thick] (u2) -- (u3);
\draw[thick] (u3) -- (u4);
\draw[dash pattern=on 6pt off 3pt, thick] (u4) -- (u5);
\draw[thick] (u5) -- (u6);

\begin{scope}[shift={(2.5, 0)}]
    \draw[dash pattern=on 6pt off 3pt, thick] (0,0) -- (0.5,0)
        node[right, black, font=\footnotesize] {$\M(L)$};
    \draw[thick] (0,-0.6) -- (0.5,-0.6)
        node[right, black, font=\footnotesize] {$\M(L_{\shortminus u^*})$};
\end{scope}

\begin{scope}[shift={(2.5, -1.5)}]
    \node[even, label={[font=\footnotesize]right:Even node}] at (0,0) {};
    \node[odd, label={[font=\footnotesize]right:Odd node}] at (0,-0.5) {};
\end{scope}

\end{scope}

\begin{scope}[xshift=\textShift]
\node[textbox, anchor=north west] at (-1.35,\textY) {
In this alternating path example, $u$ becomes unmatched when $u^*$ is introduced. Even nodes $u_2$ and $u_4$ have backups $u_3$ and $u_5$, respectively, in $L_{\shortminus u^*}$. $u^*, u_2,$ and $u_4$ are all blockers of $u$ in $L$. As the alternating path gets longer, more vertices have backups and become blockers of the unmatched vertex.
};
\end{scope}

\end{tikzpicture}%
}
\else
\begin{tikzpicture}[
    even/.style={circle, fill=blue!60!white, inner sep=1.5pt},
    odd/.style={circle, fill=red!60!white, inner sep=1.5pt},
    textbox/.style={align=left, text width=7.5cm, font=\footnotesize}
]

\foreach \i in {0,...,6} {
    \pgfmathsetmacro{\y}{-(2/3)*\i}
    \ifodd\i
        \node[odd] (u\i) at (-1,\y) {};
        \ifnum\i=0
            \node[anchor=east] at (-1.3,\y) {$u^*$};
        \else
            \node[anchor=east] at (-1.3,\y) {$u_\i$};
        \fi
    \else
        \node[even] (u\i) at (1,\y) {};
        \ifnum\i=0
            \node[anchor=west] at (1.3,\y) {$u^*$};
        \else
            \ifnum\i=6
                \node[anchor=west] at (1.3,\y) {$u$};
            \else
                \node[anchor=west] at (1.3,\y) {$u_\i$};
            \fi
        \fi
    \fi
}

\draw[dash pattern=on 6pt off 3pt, thick] (u0) -- (u1);
\draw[thick] (u1) -- (u2);
\draw[dash pattern=on 6pt off 3pt, thick] (u2) -- (u3);
\draw[thick] (u3) -- (u4);
\draw[dash pattern=on 6pt off 3pt, thick] (u4) -- (u5);
\draw[thick] (u5) -- (u6);

\begin{scope}[shift={(2.5, 0)}]
    \draw[dash pattern=on 6pt off 3pt, thick] (0,0) -- (0.5,0)
        node[right, black, font=\footnotesize] {$\M(L)$};
    \draw[thick] (0,-0.6) -- (0.5,-0.6)
        node[right, black, font=\footnotesize] {$\M(L_{\shortminus u^*})$};
\end{scope}

\begin{scope}[shift={(2.5, -1.5)}]
    \node[even, label={[font=\footnotesize]right:Even node}] at (0,0) {};
    \node[odd, label={[font=\footnotesize]right:Odd node}] at (0,-0.5) {};
\end{scope}

\node[below=0.3cm of u6, font=\footnotesize]
{Example alternating path.};

\node[textbox, right=6cm of u3] (info) {
In this alternating path example, $u$ becomes unmatched when $u^*$ is introduced.\\[2mm]
Even nodes $u_{2}$ and $u_4$ have backups $u_3$ and $u_5$, respectively, in $L_{\shortminus u^*}$. $u^*,u_2,u_4$ are all blockers of $u$ in $L$.\\[2mm]
As the alternating path gets longer, more vertices have backups and become blockers of the unmatched vertex.
};

\end{tikzpicture}
\fi

Notice that when $u^*$ causes $u$ to become unmatched, the alternating path $P=(u^*,u_1,...,u_{k-1},u)$ is a path between $u^*$ and $u$ of even length $k$. Since we picked $(u,u^*)$ to be a pair of matched vertices in the maximum matching, $P\cup\{(u,u^*)\}$ is an odd cycle of length $k+1$. If the graph has a large odd girth, then $k/2$ is large and hence we can collect more compensations for $u$ and $u^*$. We exploit this property and deduce stronger LP systems for graphs with large odd girths in \cref{sec:odd_girth_graphs}.

In the next two subsections, we will discuss our model-specific ideas separately for \ranking{} and \Franking{}.

\subsection{\ranking{}}
In the context of \ranking{}, we assume that the random permutation $\pi$ is generated by assigning each vertex $v$ an independent \emph{rank} $x_v$, drawn uniformly at random from the interval $[0,1]$, and ordering the vertices by increasing rank. The vertices are then processed sequentially according to this order, with each vertex being matched to the first available neighbor (if any) according to the same permutation $\pi$. We also let $\vecx_{\shortminus u^\ast}$ denote the vector of ranks excluding vertex $u^\ast$.

Our gain-sharing function for \ranking{} follows a similar approach to that designed by~\cite{0.546ranking}. When two vertices $u$ and $v$ are matched, we divide their unit gain using a function $g(x_u, x_v)$ that depends on their ranks. The main difference in our analysis is the addition of the notions of \emph{victim} and a corresponding \emph{compensation function}. Specifically, we assume that any matched vertex $w$ also pays a compensation $h(x_w,x_z)$, where $x_w,x_z$ is the rank of $w$ and its match, respectively. See \cref{sec:gainsharingranking} for more details about this gain-sharing scheme. We then use a factor-revealing linear program (LP) to determine the optimal forms of the functions $g$ and $h$.

\paragraph{Global Analysis via Victim.}
We now discuss how our approach differs from that of~\cite{0.546ranking}, particularly the additional universal analyses that allow us to surpass the $0.5469$ approximation ratio. To compute the expected value of $\text{gain}(u) + \text{gain}(u^\ast)$,~\cite{0.546ranking} introduced the notion of a \emph{profile} of $u$ (\cref{def:profile_ranking}), which classifies $\vecx_{\shortminus u^\ast}$ based on the ranks of $u$, its match, and its backup. All instances of $\vecx_{\shortminus u^\ast}$ with the same profile are treated identically. The authors first computed the worst-case expected gain for each profile, taking the expectation over the randomization of $x_{u^\ast}$, and then integrated over the worst-case profile distribution to obtain the final expected gain.

Our approach preserves this overall structure but strengthens the bound for each fixed profile by introducing the notion of \emph{victim} and establishing a new monotonicity property that was not considered in~\cite{0.546ranking}. The inclusion of the victim and compensation functions makes our analysis more \emph{global}, as it effectively redistributes gain across different values of $x_{u^\ast}$—reducing it for some and increasing it for others. Since the expected total gain increases overall, this refinement leads to an improved approximation ratio.

A direct incorporation of victim analysis into the LP model of~\cite{0.546ranking}, however, introduces a major computational challenge: during the final discretization step, the number of constraints grows by a factor of $\Omega(n)$\footnote{$n$ is the number of pieces that we break $[0,1]$ into in order to numerically approximate the solution of the LP system.} This increase largely negates the benefit of the compensation mechanism due to the explosion in LP size. To address this issue, we introduce a monotonicity constraint on the rank of the match of $u^\ast$, which allows us to incorporate victim analysis using asymptotically the same number of constraints as in~\cite{0.546ranking}. Moreover, this monotonicity constraint eliminates one of the worst-case matching scenarios considered in the previous analysis, further tightening the bound.

\paragraph{Monotonicity for the rank of $u^\ast$'s match.}
We observe that $u^\ast$'s match worsens as its rank increases. In~\cite{0.546ranking}, the rank of $u^\ast$'s match was not assumed to have a strong monotonic relationship with increasing $x_{u^\ast}$ values; consequently, their analysis required case-by-case constraints for all possible ranks of $u^\ast$ in $[0,1]$. This lack of monotonicity is the primary reason why naively incorporating victim analysis leads to an $\Omega(n)$-fold increase in the number of constraints. By enforcing a monotonicity constraint, we can instead partition the unit interval into at most three sub-intervals $[a_i, b_i]$, where the match of $u^\ast$ at any $x_{u^\ast} \in [a_i, b_i]$ is determined by its match at $b_i$. As a result, in our final approximation, we no longer need to consider the match of $u^\ast$ for every $x_{u^\ast}  \in [0,1]$, but only at a few representative values.

\paragraph{Multivariate Compensation Function.}
When a blocker $w$, matched to a vertex $z$, pays compensation to a victim, previous victim analyses~\cite{0.521Franking,0.531RDO} let $w$ transfer $h(y)$ amount of its gain to the victim, where $y$ is the rank of either $z$ or $w$. This design arises because their analyses have limited information about the rank of the other vertex in the pair $(w,z)$. We define the compensation function $h$ as a bivariate function depending on the ranks of both endpoints. This allows us to incorporate more information about the blocker and its match, leading to a tighter bound on the overall approximation ratio.

\paragraph{Large Odd Girth Implies Long Alternating Paths.}
For any graph $G$, the alternating path of interest (a path between $u$ and $u^*$) needs to have length at least $k$, where $k+1$ is the odd girth of $G$. This implies that we can assume $u,u^*$ are victims of at least $k/2$ blockers, hence receiving at least $k/2$ copies of compensation. We exploit this fact and design stronger \ranking{} LPs for graphs of larger odd girths.

\subsection{\Franking{}}

In \Franking{}, the decision order $\pi$ is fixed, but the preference order $\sigma$ is a uniform random permutation which we assume is generated by assigning each vertex $v$ an independent \emph{rank} $x_v$, drawn uniformly at random from the interval $[0,1]$, and ordering the vertices by increasing rank. The vertices are then processed sequentially according to $\pi$, with each vertex being matched to the first available neighbor (if any)  in the order of $\sigma$. When an edge $(u,v)$ joins the matching, $u$ is said to be actively matched if the match happens at $u$'s decision time. In this case, $v$ is called the passive vertex.  Similar to before, we let $\vecx_{\shortminus u^\ast}$ denote the vector of ranks excluding vertex $u^\ast$. 

Our gain-sharing method and factor-revealing LP are similar to what we do for \ranking{} with some key differences. First, since \Franking{} has only one-sided randomization, our gain-sharing function $g(x)$ and compensation function $h(x)$ take only one random variable, which is the rank of the passive vertex in the match. Moreover, the profile of a rank vector $\vecx_{\shortminus u^\ast}$, besides containing the ranks of $u$, its match and its backup, needs to keep track of whether they are active or passive.

Previous \Franking{} analysis~\cite{0.521Franking} is also conducted through randomized primal-dual with gain-sharing. Their analysis introduced the notion of victim but did not include the notions of profile and backup. The analysis considered, for each fixed $\vecx_{\shortminus u^*}$, a system of worst-case matches of $u$ and $u^*$ over the randomization of $u^\ast$. By introducing the notion of backups and profiles, we are able to classify $\vecx_{\shortminus u^*}$ into six different profiles. Each profile gives a different worst-case matching condition for $u$ and $u^*$. Further, we are able to add constraints on the distribution of profiles themselves, allowing a more detailed universal analysis over the randomization of $u,\;u^*$, and a third vertex $v$ (the match of $u$) together.

\paragraph{Backup Characterizes Unique Worst Case Match for $u$.}
By identifying the profile of $u$ for each $\vecx_{\shortminus u^*}$, we have a complete characterization of how $u$ is made worse off as a result of introducing $u^*$ (\cref{fact:backup_uniquely_describe_worse_off}). In particular, we know that if $u$ becomes worse off, then $u$ is exactly matched to its backup $b$ (or unmatched if $u$ has no backup). In the previous analysis~\cite{0.521Franking}, they effectively considered up to $\frac{n}{2}$ possibilities of how different values of $x_{u^*}$ can make $u$ worse off in a different way.
As a result, they required an additional constraint on the compensation function, which limited and complicated their analysis.

\paragraph{Distribution of Profiles Enforces Distribution of Marginal Ranks.}
The marginal rank of $u^\ast$ with respect to a fixed $\vecx_{\shortminus u^\ast}$ is defined as the unique rank $\theta_1$  where $u^*$ transitions from being passively matched to actively finding a match. In previous work, the analysis is conducted by considering the worst case $\theta_1$ for each $\vecx_{\shortminus u^*}$ and without imposing restrictions on the distribution of $\theta_1$. We identified that for those profiles where $u$ is actively matched to $v$, the marginal rank $\theta_1$ has to satisfy the condition $\theta_1\geq x_v$. Hence, this allows us to extend a monotonicity property that we have for $x_v$ (\cref{lemma:monotonicity_Franking}) to a monotonicity constraint for $\theta_1$. This helps us to restrict the worst-case choices of marginal ranks for some $\vecx_{\shortminus u^\ast}$ resulting in an improved approximation ratio.

\paragraph{Different Compensation Rules.}
Previous analysis requires a blocker vertex $w$ to transfer part of its gain to its victim $v$ only when all three of the following conditions are met: i) $v$ is the victim of $w$, ii) $w$ is active, and iii) $w$ is a neighbor of $v$. We modified the compensation rule by removing the third condition. We observe that due to properties related to marginal rank distributions,  this less restrictive compensation rule results in extra gain and an improved approximation ratio.

\section{Preliminaries}
We study undirected general graphs $G=(V,E)$. We denote by $N(v)$ the set of neighbors of a vertex $v$. A matching $M\subseteq E$ on $G$ is a subset of edges such that no two edges share an endpoint. A maximal matching is a matching $M$ that is not properly contained in another matching. A maximum matching $M^*$ is a matching of the largest possible size. A maximum matching $M^*$ is perfect if $|V|=2|M^*|$, i.e. all vertices in $G$ are matched in $M^*$. We assume for this paper that any graph $G$ contains a perfect matching $M^*$ of size $|V|/2$. This is due to the following fact:
\begin{fact}[$G$ Contains a Perfect Matching $M^\ast$~\cite{Randomized_greedy_matching}]\label{assumeperfectmatching}
   Let $A$ be a randomized greedy (query-commit) matching algorithm (\Cref{alg:querycommit}). Suppose $\alpha$ lower bounds the approximation ratio of $A$ on all graphs $G$ that admit perfect matchings, then $\alpha$ is also a lower bound for the approximation ratio of $A$ on all graphs $G$.
\end{fact}
This fact is essentially due to the alternating path lemma (\cref{lem:alt-path}). For anyone interested in a proof, see \nameref{sec:appendix-a}.

In this paper, we fix an arbitrary maximum matching $M^*$ of $G$ and let $(u,u^*)$ be an arbitrary matched edge in $M^*$. We will often use $v$ to denote the match of $u$ when $u^*$ is removed, and we will often use $b$ to denote the backup of $u$. We will often use $w$ to denote the match of $u^*$ when $u^*$ is introduced and causes an alternating path leading to $u$ being unmatched. The symbols $u,v,w$ may also sometimes denote arbitrary vertices in $G$; it will be clear from the context when they denote these special vertices.

\section{Alternating Path Lemma, Backups, and Victims}
\phantomsection\label{sec:Alt_path_backup_victim}
In this section, we formalize a unified framework that applies to all randomized greedy matching algorithms and captures the common structural backbone of their analyses. The framework provides general lemmas that subsume key results used across prior analyses. The framework is built around the alternating path lemma and the related notions of backup and victim. We begin by recalling the query-commit matching algorithm, which is essentially the greedy algorithm that includes edges in the order of $L$.
\begin{definition}[Greedy Matching Algorithm]\label{def:greedyalg}
    Let $L$ be an ordered list of pairs of vertices in $V$. Assume we run the following algorithm to generate a maximal matching:
    
    \begin{algorithm}[H]
  \caption{Greedy (Query-commit) Matching with $L$}
  \phantomsection
  \label{alg:querycommit}
  \KwIn{Graph $G=(V,E) ,\; \text{ ordered list }L\subseteq V\times V$}
  \SetAlgoLined

  \For{each pair of vertices $(u,v) \in L$ according to $L$'s ordering}{
    \If{$(u,v)\in E$ and both $u,v$ are available}{
      Match $u,v$, mark them as unavailable.
    }
  }
\end{algorithm}
We denote $\M(L)$ as the resulting matching of the query-commit algorithm on $L$. 
\end{definition}
A randomized query-commit matching algorithm chooses a distribution $D$ of edge query lists $L$ and runs the query-commit matching algorithm with the distribution $D$. All randomized greedy matching algorithms are randomized query-commit algorithms equipped with different distributions $D$. \nameref{alg:ranking}, for example, performs query-commit matching on randomized list $L$ generated by $\pi\times\pi$, where $\pi$ is a uniform random permutation on $V$. On the other hand, \nameref{alg:Franking} runs query-commit matching on $L=\pi\times\sigma$, where $\pi$ is a fixed adversarial ordering and $\sigma$ is a uniform random permutation of $V$. Other algorithms such as UUR, RDO and MRG can also be interpreted as query-commit algorithms with different choices of distributions $D$.

Since $L$ is an ordered list, we can view the query-commit matching process on $L$ as a timed process. For $(u,v)\in L$, when we refer to \emph{time} $t=(u,v)$, we denote the time right before the edge $(u,v)$ gets queried\footnote{Although an edge $(u,v)$ may be queried more than once, we assume W.L.O.G. that $t=(u,v)$ refers to the first time the edge gets queried, as the second query will never result in a match.} in the query-commit matching process. For a given time $t\in L$, we use $\M^t(L)$ to denote the partial matching constructed at time $t$, and we use $\A^t(L)$ to denote the set of remaining available vertices at time $t$. When $t$ is omitted, we assume that $\M(L)$ and $\A(L)$ represent the final matching and available set of vertices respectively.
\subsection{Alternating Path Lemma}
\phantomsection
\label{sec:alt-lemma}
The alternating path lemma is a key tool in previous analyses on general graphs. Several variants have been proved for different randomized greedy algorithms, sharing the same underlying structure but yielding algorithm-specific outcomes. Here, we present a generalized alternating path lemma applicable to any fixed query list $L$, which will then be used to formalize the roles of backup and victim.
\begin{lemma}[Alternating Path Lemma]
\phantomsection
\label{lem:alt-path}
    Let $L$ be an ordered list of pairs of nodes in $V$. For any vertex $v\in V$, let $L_{\shortminus v}$ denote the list $L$ but with $v$ marked as unavailable from the beginning\footnote{We could as well interpret $L_{\shortminus v}$ as removing all pairs of nodes incident to $v$, but we want to synchronize the time notions in $L$ and $L_{\shortminus v}$. Marking $v$ as unavailable without removing it enables us to discuss time relevant to $v$ in both lists.}. The symmetric difference of partial matchings at time $t$, that is
    $$\M^t(L)\oplus \M^t(L_{\shortminus v})=(\M^t(L)- \M^t(L_{\shortminus v}))\cup( \M^t(L_{\shortminus v})-\M^t(L)),$$ is an alternating path $u_0,\ldots, u_k$ such that
    \begin{enumerate}
        \item $u_0=v$.
        \item  For all even $i$, $(u_i,u_{i+1})\in \M^t(L)$, and for all odd $i$, 
        $(u_i,u_{i+1})\in \M^t(L_{\shortminus v})$.
        \item The time $(u_i,u_{i+1})$ gets queried monotonically increases. i.e., \\$\forall i<k-1$, $(u_i,u_{i+1})<(u_{i+1},u_{i+2})$.
        \item The set of available vertices at time $t$ only differs at $u_k$. \\
        If $k$ is odd, then $\A^t(L_{\shortminus v})=\A^t(L)\cup\{u_k\}$ and if $k$ is even, then $\A^t(L)=\A^t(L_{\shortminus v})\cup\{u_k\}$.
    \end{enumerate}
\end{lemma}
The idea of the proof is not different from previous versions of the alternating path lemma. For anyone who is interested in the proof, please refer to \nameref{sec:appendix-a}. The alternating path lemma gives a depiction of the effect on matching results when a vertex $v$ is added to or removed from the graph while keeping the remaining edge query ordering untouched. The difference in the matchings consists of a path that starts at $v$ (the vertex that is added or removed), and the edges alternate between the two matchings. Further, the alternating path lemma states that the relative query order for the edges in the path monotonically increases along the path. 

For any ordered list $L$, we say a vertex $u$ becomes \emph{worse off by introducing $v$} if the time $u$ gets matched in $\M(L)$ is later than the time $u$ gets matched in $\M(L_{\shortminus v})$ (or if $u$ is unmatched in $\M(L)$). $u$ \emph{becomes better off by introducing v} or \emph{worse off by removing $v$} if the converse holds. Note that being worse (better) off by introducing $v$ is equivalent to being an even (odd) indexed vertex in the alternating path. i.e., for $u\neq v$:
\ifdoublecol
\begin{align*}
    &\text{$u$ becomes worse off by introducing $v$}\\
    \Longleftrightarrow\quad &\text{$u $ is matched at a later time (or unmatched)}\\ 
     &\quad\text{in $\M(L)$ than in $\M(L_{\shortminus v})$}\\
    \Longleftrightarrow\quad &\text{$u$ is an even-indexed vertex in the alternating path }\\
    &\quad \M(L)\oplus \M(L_{\shortminus v}).
\end{align*}
\else
\begin{align*}
    &\text{$u$ becomes worse off by introducing $v$}\\
    \Longleftrightarrow\quad &\text{$u $ is matched at a later time (or unmatched)
     in $\M(L)$ than in $\M(L_{\shortminus v})$}\\
    \Longleftrightarrow\quad &\text{$u$ is an even-indexed vertex in the alternating path $\M(L)\oplus \M(L_{\shortminus v})$}
\end{align*}
\fi
A similar iff relation holds for the notion of being better off. The notions of worse (better) off have different manifestations in different randomized vertex-iterative greedy matching algorithms; we will discuss the relevant notions in the corresponding sections.

\subsection{Backup and Victim in Randomized Primal-Dual}
A common analysis framework for randomized greedy matching algorithms is the randomized primal-dual method~\cite{RandomizedPrimalDual}. For each edge added to the matching, one unit of total \emph{gain} is distributed among the vertices in $V$. To lower bound the final approximation ratio, we can lower bound the expected sum of gain of any pair of perfect matches $(u,u^\ast)$. Assuming that $G$ contains a perfect matching $M^\ast$ (\cref{assumeperfectmatching}), we have: 
\begin{align*}
    \min_{(u,u^*)\in M^*}\left\{\ev_{L}[\text{ gain$(u)$ $+$ gain$(u^\ast)$ }]\right\}
    \leq&\frac{\sum_{(u,u^\ast)\in M^*}\ev_L[\text{ gain$(u)$ + gain$(u^\ast)$ }]}{|M^*|}\\
    =&\frac{\ev_L\left[\sum_{v\in V}\text{ gain$(v)$}\right]}{|M^*|}=\frac{\ev_L[|\M(L)|]}{|M^*|}.
\end{align*}
The expectation is often taken in two levels: first, fix an arbitrary ordering of edges in $G$ without $u^\ast$ (denoted $L_{\shortminus u^\ast}$), and compute the expectation over the randomization of $u^\ast$, assuming the remaining edge ordering is fixed as $L_{\shortminus u^\ast}$; then take the outer expectation over $L_{\shortminus u^\ast}$. Mathematically, the randomized primal-dual method roughly computes the following expectation:
 $$\ev_{L_{\shortminus u^\ast}}\left[\; \ev_{\text{$u^\ast$}}\;\left[\;\vphantom{\int}\text{gain$(u)$ $+$ gain$(u^\ast)$ }\right]\right].$$
In the bipartite graph analysis of \ranking{}, if we know that $u$ is matched in $L_{\shortminus u^\ast}$, then $u$ is also matched in $L$, independent of the randomization of $u^\ast$. With this fact and a cleverly chosen gain-sharing rule, ~\cite{RandomizedPrimalDual,Economicbasedanalysis} were able to prove the optimal $1-\frac{1}{e}$ bound for \ranking{} in bipartite graphs. 

However, this property does not hold for general graphs, which is the main difficulty in establishing performance guarantees for many randomized greedy matching algorithms for general graphs. When randomized primal-dual is performed naively without this property, the best we can get is $0.5$. 

To address this challenge, the notion of \emph{backup} was introduced by~\cite{0.546ranking} for the \ranking{} analysis.  In the backup framework, when a vertex $u$ has a backup $b$, introducing  $u^\ast$ will not cause $u$ to match with a vertex worse than $b$. The existence of a backup ensures that $u$'s matching condition is not affected significantly when $u^\ast$ is added. 

A different matching structure, namely \emph{victim}, was introduced by~\cite{0.521Franking} to help with \Franking{} analysis to address this challenge; this notion is later extended to \cite{0.531RDO} to facilitate the analysis of RDO. The notion of victim operates differently from the backup. Whenever a vertex $w$ has a victim $u$, $w$ in some sense takes away $u$’s opportunity to be matched. Consequently, in~\cite{0.521Franking,0.531RDO}, $w$ transfers some of its gain as compensation to its victim $u$\footnote{These two works enforce additional criteria for when this gain-transfer is conducted, beyond the fact that $w$ causes $u$ to be unmatched. These additional conditions are optimized for \Franking{} and RDO specifically. Regardless of the algorithm, the key feature for $u$ being the victim of $w$ is that $w$ takes away $u$’s chance of being matched.}. By showing that vertices are more likely to receive compensation than to pay it, both works achieve approximation ratios strictly greater than $0.5$.

 In the following, we identify where backups and victims appear along alternating path, thereby establishing a unified connection between these three structural components.
\subsection{Backup}
We first define the notion of backup.
\begin{definition}[Backup]
    \phantomsection
\label{def:backup_querycommit}
We say $b$ is the backup of $u$ with respect to $L$ iff $u$ is matched to some vertex $v$ in $\M(L)$, and $u$ is matched to $b$ in $\M(L_{\shortminus v})$.
\end{definition}
We have the following fact about the backup of a vertex:
\begin{fact}[Backup Uniquely Describes Worse-off]
\phantomsection
\label{fact:backup_uniquely_describe_worse_off}
    Suppose $u$ has a match $u_1$ and a backup $b$ in $L$. If $u$ becomes worse off by removing any vertex $v$, then $u$ is matched to $b$ in $\M(L_{\shortminus v})$. If $u$ has no backup, then it is unmatched by removing $v$.

    Suppose $u$ has a match $u_1$ and a backup $b$ in $L_{\shortminus v}$. If $u$ becomes worse off by introducing $v$, then $u$ is matched to $b$ in $\M(L)$. If $u$ has no backup, then it is unmatched by introducing $v$.
\end{fact}
    Note that we do not assume $u_1=v$. Thus the first property is different from restating the definition of backup. We leave the proof to \nameref{sec:appendix-a}. 

Intuitively, the above fact shows that a backup uniquely determines how a vertex can become worse off when another vertex is removed or introduced. Now consider the alternating path $\M(L)\oplus \M(L_{\shortminus v})$, and let $v,\ldots,u_k$ denote its vertices. Each vertex on this path is worse off in one of the two matchings: odd vertices become worse off when $v$ is removed, while even vertices become worse off when $v$ is introduced. Therefore, by \cref{fact:backup_uniquely_describe_worse_off}, the next vertex along the alternating path must be exactly the backup of the current vertex in the corresponding list. The following lemma formalizes this relationship between backups and the alternating path.
\begin{lemma}[Backups in Alternating Path]
\phantomsection
\label{lem:backup_in_alt_path}
    For every odd vertex $u_i$, $u_{i+1}$ is its backup in $L$. For every even vertex $u_i\neq v$, $u_{i+1}$ is its backup in $L_{\shortminus v}$. If the odd/even $i=k$, then $u_i$ does not have a backup in its respective list $L/L_{\shortminus v}$.
\end{lemma}
\begin{proof}
    By \cref{fact:backup_uniquely_describe_worse_off}, $u_{i+1}$ are exactly the worse off choices for vertex $u_i$ to match after removing (introducing) $v$ for odd (even) $u_i$. 
\end{proof}
Intuitively, backup defines the unique way a vertex can become worse off when a vertex is introduced/removed. Because odd $u_i$ becomes worse off by removing $v$, their backups are their worse choice, that is, their match in $\M(L_{\shortminus v})$. On the other hand, even $u_i$ becomes worse off by introducing $v$, their backups are their worse choice match in $\M(L)$. See the following graph for an example:

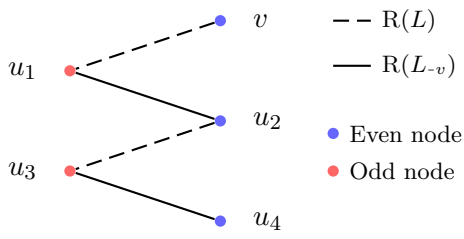
\ifdoublecol

\def\FigShift{-1.5cm}    
\def\picShift{1cm}    
\def\textShift{0cm}   
\def\textY{-3}     
\def\textW{6.5cm}     

\makebox[\linewidth][c]{%
\hspace*{\FigShift}%
\begin{tikzpicture}[
    even/.style={circle, fill=blue!60!white, inner sep=1.5pt},
    odd/.style={circle, fill=red!60!white, inner sep=1.5pt},
    textbox/.style={align=justify, text width=\textW, font=\footnotesize}
]

\begin{scope}[xshift=\picShift]

\foreach \i in {0,...,4} {
    \pgfmathsetmacro{\y}{-(2/3)*\i}
    \ifodd\i
        \node[odd] (u\i) at (-1,\y) {};
        \ifnum\i=0
            \node[anchor=east] at (-1.3,\y) {$v$};
        \else
            \node[anchor=east] at (-1.3,\y) {$u_\i$};
        \fi
    \else
        \node[even] (u\i) at (1,\y) {};
        \ifnum\i=0
            \node[anchor=west] at (1.3,\y) {$v$};
        \else
            \node[anchor=west] at (1.3,\y) {$u_\i$};
        \fi
    \fi
}

\draw[dash pattern=on 6pt off 3pt, thick] (u0) -- (u1);
\draw[thick] (u1) -- (u2);
\draw[dash pattern=on 6pt off 3pt, thick] (u2) -- (u3);
\draw[thick] (u3) -- (u4);

\begin{scope}[shift={(2.5, 0)}]
    \draw[dash pattern=on 6pt off 3pt, thick] (0,0) -- (0.5,0)
        node[right, black, align=left, font=\footnotesize] {$\M(L)$};
    \draw[thick] (0,-0.6) -- (0.5,-0.6)
        node[right, black, align=left, font=\footnotesize] {$\M(L_{\shortminus v})$};
\end{scope}

\begin{scope}[shift={(2.5, -1.5)}]
    \node[even, label={[font=\footnotesize]right:Even node}] at (0,0) {};
    \node[odd, label={[font=\footnotesize]right:Odd node}] at (0,-0.5) {};
\end{scope}

\path (current bounding box.north west) -- (current bounding box.north east)
  node[midway, above=4pt, font=\footnotesize, align=center]
  {Example alternating path for lemma $4.4$\footnotemark};

\end{scope}

\begin{scope}[xshift=\textShift]
\node[textbox, anchor=north west] at (-1.35,\textY) {
Odd node $u_3$ is the backup of even node $u_2$ when $u_2$ is matched with $u_1$ in $\M(L_{\shortminus v})$. Even node $u_2$ is the backup of odd node $u_1$ when $u_1$ is matched to $v$ in $\M(L)$. A similar relationship holds for the node triplet $(u_4,u_3,u_2)$. The last node, which could be even or odd (in this case $u_4$), does not have a backup in the respective matchings $\M(L_{\shortminus v})$ and $\M(L)$.
};
\end{scope}

\end{tikzpicture}%
}

\else

\begin{tikzpicture}[
    even/.style={circle, fill=blue!60!white, inner sep=1.5pt},
    odd/.style={circle, fill=red!60!white, inner sep=1.5pt},
    textbox/.style={align=left, text width=7.5cm, font=\footnotesize}
]

\foreach \i in {0,...,4} {
    \pgfmathsetmacro{\y}{-(2/3)*\i} 
    \ifodd\i
        \node[odd] (u\i) at (-1,\y) {};  
        \ifnum\i=0
            \node[anchor=east] at (-1.3,\y) {$v$};
        \else
            \node[anchor=east] at (-1.3,\y) {$u_\i$};
        \fi
    \else
        \node[even] (u\i) at (1,\y) {};  
        \ifnum\i=0
            \node[anchor=west] at (1.3,\y) {$v$};
        \else
            \node[anchor=west] at (1.3,\y) {$u_\i$};
        \fi
    \fi
}

\draw[dash pattern=on 6pt off 3pt, thick] (u0) -- (u1);
\draw[thick] (u1) -- (u2);
\draw[dash pattern=on 6pt off 3pt, thick] (u2) -- (u3);
\draw[thick] (u3) -- (u4);

\begin{scope}[shift={(2.5, 0)}]
    \draw[dash pattern=on 6pt off 3pt, thick] (0,0) -- (0.5,0) node[right, black, align=left, font=\footnotesize] {$\M(L)$};
    \draw[thick] (0,-0.6) -- (0.5,-0.6) node[right, black, align=left, font=\footnotesize] {$\M(L_{\shortminus v})$};
\end{scope}

\begin{scope}[shift={(2.5, -1.5)}] 
    \node[even, label={ [font=\footnotesize]right:Even node}] at (0,0) {};
    \node[odd, label={[font=\footnotesize] right:Odd node}] at (0,-0.5) {};
\end{scope}

\node[below=0.3cm of u4, font=\footnotesize] {Example alternating path for~\cref{lem:backup_in_alt_path}\footnotemark.};

\node[textbox, right=4cm of u2] (info) {
Odd node $u_{3}$ is the backup of even node $u_2$ when $u_2$ is matched with $u_1$ in $\M(L_{\shortminus v})$. \\[2mm]
Even node $u_{2}$ is the backup of odd node $u_1$ when $u_1$ is matched to $v$ in $\M(L)$. A similar relationship holds for the  node triplet $(u_4,u_3,u_2)$. \\[2mm]
The last node, which could be even or odd (in this case $u_4$), does not have a backup in the respective matchings $\M(L_{\shortminus v})$ and $\M(L)$.  
};

\end{tikzpicture}

\fi

\footnotetext{Query order increases from top edges to bottom edges.}

\subsection{Victim}
 We begin by defining generalized versions of victim with respect to any ordered query list $L$.
\begin{definition}[Victims and Blockers]
\phantomsection
\label{def:victim_querycommit}
We say that $u$ is the victim of $w$ and, conversely, that $w$ is the blocker of $u$, with respect to $L$ if 
\begin{itemize}
    \item $u$ is unmatched in $\M(L)$,
    \item $u$ is matched when $w$ is removed; that is, $u$ is matched in $\M(L_{\shortminus w})$.
\end{itemize}
\end{definition}
A simple observation that follows directly from the definitions is that, with respect to any fixed list $L$, each vertex can have at most one victim, whereas a vertex may become the victim of multiple blocker vertices. This structural asymmetry is a key property that enables our improved analysis for both \ranking{} and \Franking{}, and further, our odd-girth-sensitive analysis.

Now we introduce two lemmas connecting the notion of victim with the alternating path.
\begin{lemma}[Blockers in Alternating Path I]
\phantomsection
\label{lem:victim_querycommit_case_u}
    Let $u,v\in V$ be two arbitrary vertices. Suppose $u$ is unmatched by introducing $v$ to $L_{\shortminus v}$. Let $u_0,\ldots,u_{k}$ be the alternating path $\M(L)\oplus\M(L_{\shortminus v})$ where $u_0=v,u_k=u$. Then $u$ is the victim of vertices $u_i$ in $L$ for all even $i<k$.
\end{lemma}
\begin{lemma}[Blockers in Alternating Path II]
\phantomsection
\label{lem:victim_querycommit_case_u_star}  
Let $v\in N(u)$ be a neighbor of $u$ in $G$. Suppose $u$ becomes worse off in $\M(L)$ compared to $\M(L_{\shortminus v})$. Let $u_0,\ldots,u_{j},u_{j+1},..., u_k$ be the alternating path $\M(L)\oplus\M(L_{\shortminus v})$ where $u_0=v,u_j=u,u_{j+1}=b$ (path ends at $u_j$ in the case where $u$ is unmatched). Let $L'$ be a list such that the following holds:
\begin{itemize}
    \item $v$ is unmatched in $\M(L')$.
    \item $L'_{\shortminus v}=L_{\shortminus v}$\footnote{Technically, by the way we defined $L_{\shortminus v}$ (marking $v$ unavailable instead of removing incident edges), this condition is not accurate. The most correct way to state this requirement is that $L'$ agrees with $L$ on all pairs of nodes except those incident to $v$. But since these two views are essentially equivalent under the query-commit matching view, we abuse notations here for simplicity.}.
    \item $v$ is preferred by $u$ over $b$ in $L'$ if $u$ has a backup $b$ in $L_{\shortminus v}$. i.e., $(u,v)<(u,b)$ in $L'$.
\end{itemize}
Then, for every odd-indexed vertex $u_i$ in $\M(L)\oplus\M(L_{\shortminus v})$ with $i<j$, $v$ is the victim of $u_i$ with respect to $L'$.
\end{lemma}
\ifdoublecol

\def\FigShift{0cm}    
\def\picShift{0cm}    
\def\textShift{-2cm}   
\def\textY{-4}      
\def\textW{6.5cm}     

\makebox[\linewidth][c]{%
\hspace*{\FigShift}%
\begin{tikzpicture}[
    even/.style={circle, fill=blue!60!white, inner sep=1.5pt},
    odd/.style={circle, fill=red!60!white, inner sep=1.5pt},
    textbox/.style={align=justify, text width=\textW, font=\footnotesize}
]

\begin{scope}[xshift=\picShift]

\foreach \i in {0,...,4} {
    \pgfmathsetmacro{\y}{-(2/3)*\i}

    \ifnum\i=0
        \def\nodename{n0}
    \else\ifnum\i=1
        \def\nodename{n1}
    \else\ifnum\i=2
        \def\nodename{n2}
    \else\ifnum\i=3
        \def\nodename{n3}
    \else
        \def\nodename{n4}
    \fi\fi\fi\fi

    \ifnum\i=0
        \def\nodelabel{$v$}
    \else\ifnum\i=2
        \def\nodelabel{$u_{k\shortminus 2}$}
    \else\ifnum\i=3
        \def\nodelabel{$u_{k\shortminus 1}$}
    \else\ifnum\i=4
        \def\nodelabel{$u$}
    \else
        \def\nodelabel{$u_\i$}
    \fi\fi\fi\fi

    \ifodd\i
        \node[odd] (\nodename) at (-1,\y) {};
        \node[below=2pt of \nodename, font=\footnotesize] {\nodelabel};
    \else
        \node[even] (\nodename) at (1,\y) {};
        \node[below=2pt of \nodename, font=\footnotesize] {\nodelabel};
    \fi
}

\draw[dash pattern=on 6pt off 3pt, thick] (n0) -- (n1);
\draw[dotted, line width=1pt, dash pattern=on 1pt off 2pt] (0, {-(2/3)*1}) -- (0, {-(2/3)*2});
\draw[dash pattern=on 6pt off 3pt, thick] (n2) -- (n3);
\draw[thick] (n3) -- (n4);

\begin{scope}[shift={(2.5, 0)}]
    \draw[dash pattern=on 6pt off 2pt, thick] (0,0) -- (0.5,0)
        node[right, black, align=left, font=\footnotesize] {$\M(L)$};
    \draw[thick] (0,-0.6) -- (0.5,-0.6)
        node[right, black, align=left, font=\footnotesize] {$\M(L_{\shortminus v})$};
\end{scope}

\begin{scope}[shift={(2.5, -1.5)}]
    \node[even, label={[font=\footnotesize]right:Even node}] at (0,0) {};
    \node[odd, label={[font=\footnotesize]right:Odd node}] at (0,-0.5) {};
\end{scope}

\draw[line width=1pt]
    ([xshift=-3cm, yshift=0cm]n0.east) -- ([xshift=-3.3cm, yshift=0cm]n0.east);
\draw[dotted, line width=1pt]
    ([xshift=-3cm, yshift=0cm]n0.east) -- ([xshift=-0.2cm, yshift=0cm]n0.east);

\draw[line width=1pt]
    ([xshift=-3.15cm, yshift=0cm]n0.east) -- ([xshift=-3.15cm, yshift=-1.2cm]n4.east);

\draw[line width=1pt]
    ([xshift=-3cm, yshift=-1.2cm]n4.east) -- ([xshift=-3.3cm, yshift=-1.2cm]n4.east);

\draw[dotted, line width=1pt]
    ([xshift=-3cm, yshift=-1.2cm]n4.east) -- ([xshift=-0.1cm, yshift=-1.2cm]n4.east);

\node[anchor=east, font=\footnotesize, align=center]
    at ([xshift=-3.15cm, yshift=-0.45cm]n2.east) {victim\\range};

\node[font=\footnotesize] at (0,0.6)
{Blockers in the alternating path when $u$ has no backup in $L_{\shortminus v}$};

\end{scope}

\begin{scope}[xshift=\textShift]
\node[textbox, anchor=north west] at (-1.35,\textY) {
If $u$ is unmatched in $\M(L)$, then $u$ is the victim of all earlier even (blue) nodes $v,u_2,\ldots,u_{k\shortminus 2}$. If $v$ is unmatched in $\M(L')$, then $v$ is the victim of all odd (red) nodes $u_1,\ldots,u_{k\shortminus 1}$, provided that $v$ is a neighbor of $u$.
};
\end{scope}

\end{tikzpicture}%
}

\else

\begin{tikzpicture}[
    even/.style={circle, fill=blue!60!white, inner sep=1.5pt},
    odd/.style={circle, fill=red!60!white, inner sep=1.5pt},
    textbox/.style={align=left, text width=7.5cm, font=\footnotesize}
]

\foreach \i in {0,...,4} {
    \pgfmathsetmacro{\y}{-(2/3)*\i} 

    \ifnum\i=0
        \def\nodename{n0}
    \else\ifnum\i=1
        \def\nodename{n1}
    \else\ifnum\i=2
        \def\nodename{n2}
    \else\ifnum\i=3
        \def\nodename{n3}
    \else
        \def\nodename{n4}
    \fi\fi\fi\fi

    \ifnum\i=0
        \def\nodelabel{$v$}
    \else\ifnum\i=2
        \def\nodelabel{$u_{k\shortminus 2}$}
    \else\ifnum\i=3
        \def\nodelabel{$u_{k\shortminus 1}$}
    \else\ifnum\i=4
        \def\nodelabel{$u$}
    \else
        \def\nodelabel{$u_\i$}
    \fi\fi\fi\fi

\ifodd\i
    \node[odd] (\nodename) at (-1,\y) {};  
    \node[below=2pt of \nodename] {\nodelabel};
\else
    \node[even] (\nodename) at (1,\y) {};  
    \node[below=2pt of \nodename] {\nodelabel};
\fi
}

\draw[dash pattern=on 6pt off 3pt, thick] (n0) -- (n1);
\draw[dash pattern=on 6pt off 3pt, thick] (n2) -- (n3);
\draw[thick] (n3) -- (n4);

\begin{scope}[shift={(2.5, 0)}]
    \draw[dash pattern=on 6pt off 2pt, thick] (0,0) -- (0.5,0) node[right, black, align=left, font=\footnotesize] {$\M(L)$};
    \draw[thick] (0,-0.6) -- (0.5,-0.6) node[right, black, align=left, font=\footnotesize] {$\M(L_{\shortminus v})$};
\end{scope}

\begin{scope}[shift={(2.5, -1.5)}] 
    \node[even, label={ [font=\footnotesize]right:Even node}] at (0,0) {};
    \node[odd, label={[font=\footnotesize] right:Odd node}] at (0,-0.5) {};
\end{scope}

\node[below=1.5cm of n4, font=\normalsize] {Blockers in the alternating path when $u$ has no backup in $L_{\shortminus v}$};

\node[textbox, right=4cm of n2] (info) {
If $u$ is unmatched in $\M(L)$, then $u$ is the victim of all earlier even (blue) nodes $v,u_2,...,u_{k\shortminus 2}$.\\[2mm]
If $v$ is unmatched in $\M(L')$, then $v$ is the victim of all odd (red) $u_1,...u_{k\shortminus 1}$, provided that $v$ is a neighbor of $u$.
};

\draw[line width=1pt, dotted, dash pattern=on 1pt off 2pt]
    (0, {-(2/3)*1}) -- (0, {-(2/3)*2});

\draw[line width=1pt]
    ([xshift=-3cm, yshift=0cm]n0.east) -- ([xshift=-3.3cm, yshift=0cm]n0.east);
\draw[dotted, line width=1pt]
    ([xshift=-3cm, yshift=0cm]n0.east) -- ([xshift=-0.2cm, yshift=0cm]n0.east);

\draw[line width=1pt]
    ([xshift=-3.15cm, yshift=0cm]n0.east) -- ([xshift=-3.15cm, yshift=-1.2cm]n4.east);

\draw[line width=1pt]
    ([xshift=-3cm, yshift=-1.2cm]n4.east) -- ([xshift=-3.3cm, yshift=-1.2cm]n4.east);

\draw[dotted, line width=1pt]
    ([xshift=-3cm, yshift=-1.2cm]n4.east) -- ([xshift=-0.1cm, yshift=-1.2cm]n4.east);

\node[anchor=east, font=\footnotesize, align=center]
    at ([xshift=-3.15cm, yshift=-0.45cm]n2.east) {victim\\range};

\end{tikzpicture}

\fi
\ifdoublecol
\medskip
\fi

\ifdoublecol

\def\FigShift{0cm}    
\def\picShift{0cm}    
\def\textShift{-2cm}   
\def\textY{-4}      
\def\textW{6.5cm}     

\makebox[\linewidth][c]{%
\hspace*{\FigShift}%
\begin{tikzpicture}[
    even/.style={circle, fill=blue!60!white, inner sep=1.5pt},
    odd/.style={circle, fill=red!60!white, inner sep=1.5pt},
    textbox/.style={align=justify, text width=\textW, font=\footnotesize}
]

\begin{scope}[xshift=\picShift]

\foreach \i in {0,...,5} {
    \pgfmathsetmacro{\y}{-(2/3)*\i}

    \ifnum\i=0
        \def\nodename{n0}
    \else\ifnum\i=1
        \def\nodename{n1}
    \else\ifnum\i=2
        \def\nodename{n2}
    \else\ifnum\i=3
        \def\nodename{n3}
    \else\ifnum\i=4
        \def\nodename{n4}
    \else
        \def\nodename{n5}
    \fi\fi\fi\fi\fi

    \ifnum\i=0
        \def\nodelabel{$v$}
    \else\ifnum\i=2
        \def\nodelabel{$u_{j\shortminus 2}$}
    \else\ifnum\i=3
        \def\nodelabel{$u_{j\shortminus 1}$}
    \else\ifnum\i=4
        \def\nodelabel{$u$}
    \else\ifnum\i=5
        \def\nodelabel{$b$}
    \else
        \def\nodelabel{$u_\i$}
    \fi\fi\fi\fi\fi

    \ifodd\i
        \node[odd] (\nodename) at (-1,\y) {};
        \node[below=2pt of \nodename, font=\footnotesize] at (-1.3,\y) {\nodelabel};
    \else
        \node[even] (\nodename) at (1,\y) {};
        \node[below=2pt of \nodename, font=\footnotesize] at (1.3,\y) {\nodelabel};
    \fi
}

\draw[dash pattern=on 6pt off 2pt, thick] (n0) -- (n1);
\draw[dash pattern=on 6pt off 2pt, thick] (n2) -- (n3);
\draw[thick] (n3) -- (n4);
\draw[dash pattern=on 6pt off 2pt, thick] (n4) -- (n5);

\begin{scope}[shift={(2.5, 0)}]
    \draw[dash pattern=on 6pt off 2pt, thick] (0,0) -- (0.5,0)
        node[right, black, align=left, font=\footnotesize] {$\M(L)$};
    \draw[thick] (0,-0.6) -- (0.5,-0.6)
        node[right, black, align=left, font=\footnotesize] {$\M(L_{\shortminus v})$};
\end{scope}

\begin{scope}[shift={(2.5, -1.5)}]
    \node[even, label={[font=\footnotesize]right:Even node}] at (0,0) {};
    \node[odd, label={[font=\footnotesize]right:Odd node}] at (0,-0.5) {};
\end{scope}

\draw[line width=1pt, dotted, dash pattern=on 1pt off 2pt]
    (0, {-(2/3)*1}) -- (0, {-(2/3)*2});

\coordinate (Ibottom) at (n4.east|-n5.east);

\draw[line width=1pt]
    ([xshift=-3cm, yshift=0cm]n0.east) -- ([xshift=-3.3cm, yshift=0cm]n0.east);

\draw[dotted, line width=1pt]
    ([xshift=-3cm, yshift=0cm]n0.east) -- ([xshift=-0.2cm, yshift=0cm]n0.east);

\draw[line width=1pt]
    ([xshift=-3.15cm, yshift=0cm]n0.east) -- ([xshift=-3.15cm, yshift=0cm]Ibottom);

\draw[line width=1pt]
    ([xshift=-3cm] Ibottom) -- ([xshift=-3.3cm] Ibottom);

\draw[dotted, line width=1pt]
    ([xshift=-3cm] Ibottom) -- ([xshift=-2.2cm] Ibottom);

\node[anchor=east, font=\footnotesize, align=center]
    at ([xshift=-3.15cm, yshift=-0.45cm]n2.east) {victim\\range};

\node[font=\footnotesize, align=center] at (0,0.6)
{Blockers in the alternating path when $u$ has backup $b$ in $L_{\shortminus v}$};

\end{scope}

\begin{scope}[xshift=\textShift]
\node[textbox, anchor=north west] at (-1.35,\textY) {
Assume $u$ is matched to a worse vertex $b$ in $\M(L)$. If $v$ is unmatched in $\M(L')$, then as long as $v$ is preferred by $u$ over $b$, $v$ is the victim of all odd (red) nodes $u_1,\ldots,u_{j-1}$, provided that $v$ is a neighbor of $u$.
};
\end{scope}

\end{tikzpicture}%
}

\else

\begin{tikzpicture}[
    even/.style={circle, fill=blue!60!white, inner sep=1.5pt},
    odd/.style={circle, fill=red!60!white, inner sep=1.5pt},
    textbox/.style={align=left, text width=7.5cm, font=\footnotesize}
]

\foreach \i in {0,...,5} {
    \pgfmathsetmacro{\y}{-(2/3)*\i} 

    \ifnum\i=0
        \def\nodename{n0}
    \else\ifnum\i=1
        \def\nodename{n1}
    \else\ifnum\i=2
        \def\nodename{n2}
    \else\ifnum\i=3
        \def\nodename{n3}
    \else\ifnum\i=4
        \def\nodename{n4}
    \else
        \def\nodename{n5}
    \fi\fi\fi\fi\fi

    \ifnum\i=0
        \def\nodelabel{$v$}
    \else\ifnum\i=2
        \def\nodelabel{$u_{j\shortminus 2}$}
    \else\ifnum\i=3
        \def\nodelabel{$u_{j\shortminus 1}$}
    \else\ifnum\i=4
        \def\nodelabel{$u$}
    \else\ifnum\i=5
        \def\nodelabel{$b$}
    \else
        \def\nodelabel{$u_\i$}
    \fi\fi\fi\fi\fi

    \ifodd\i
        \node[odd] (\nodename) at (-1,\y) {};
        \node[below=2pt of \nodename] at (-1.3,\y) {\nodelabel};
    \else
        \node[even] (\nodename) at (1,\y) {};
        \node[below=2pt of \nodename] at (1.3,\y) {\nodelabel};
    \fi
}

\draw[dash pattern=on 6pt off 2pt, thick] (n0) -- (n1);
\draw[dash pattern=on 6pt off 2pt, thick] (n2) -- (n3);
\draw[thick] (n3) -- (n4);
\draw[dash pattern=on 6pt off 2pt, thick] (n4) -- (n5);

\begin{scope}[shift={(2.5, 0)}]
    \draw[dash pattern=on 6pt off 2pt, thick] (0,0) -- (0.5,0) node[right, black, align=left, font=\footnotesize] {$\M(L)$};
    \draw[thick] (0,-0.6) -- (0.5,-0.6) node[right, black, align=left, font=\footnotesize] {$\M(L_{\shortminus v})$};
\end{scope}

\begin{scope}[shift={(2.5, -1.5)}]
    \node[even, label={ [font=\footnotesize]right:Even node}] at (0,0) {};
    \node[odd, label={[font=\footnotesize] right:Odd node}] at (0,-0.5) {};
\end{scope}

\node[below=1.2cm of n4, font=\normalsize] {Blockers in the alternating path when $u$ has backup $b$ in $L_{\shortminus v}$};

\node[textbox, right=4cm of n2] (info) {
Assume $u$ is matched to a worse vertex $b$ in $\M(L)$. If $v$ is unmatched in $\M(L')$, as long as $v$ is preferred by $u$ over $b$, then $v$ is the victim of all odd (red) $u_1,...u_{j-1}$, provided that $v$ is a neighbor of $u$.
};

\draw[line width=1pt, dotted, dash pattern=on 1pt off 2pt]
    (0, {-(2/3)*1}) -- (0, {-(2/3)*2});

\coordinate (Ibottom) at (n4.east|-n5.east);

\draw[line width=1pt]
    ([xshift=-3cm, yshift=0cm]n0.east) -- ([xshift=-3.3cm, yshift=0cm]n0.east);

\draw[dotted, line width=1pt]
    ([xshift=-3cm, yshift=0cm]n0.east) -- ([xshift=-0.2cm, yshift=0cm]n0.east);

\draw[line width=1pt]
    ([xshift=-3.15cm, yshift=0cm]n0.east) -- ([xshift=-3.15cm, yshift=0cm]Ibottom);

\draw[line width=1pt] ([xshift=-3cm] Ibottom) -- ([xshift=-3.3cm] Ibottom);
\draw[dotted, line width=1pt] ([xshift=-3cm] Ibottom) -- ([xshift=-2.2cm] Ibottom);

\node[anchor=east, font=\footnotesize, align=center]
    at ([xshift=-3.15cm, yshift=-0.45cm]n2.east) {victim\\range};

\end{tikzpicture}

\fi
We leave the proof of both lemmas to \nameref{sec:appendix-a}. The first lemma is saying that if $v$ causes an alternating path that unmatches $u$ in $L$, then removing any even-indexed node in the alternating path will rematch $u$. 

The second lemma states that, if placing $v$ in a favored position (as in $L$) causes one of its neighbors $u$ to become worse off, then when $v$ is permuted to a less favored position (as in $L'$) so that it becomes unmatched, we can remove any odd-indexed vertex in the alternating path $\M(L)\oplus\M(L_{\shortminus v})$ to rematch $v$, provided that $v$ is favored by $u$ over the worse-off choice $b$.

\cite{0.521Franking} proved a special case of \cref{lem:victim_querycommit_case_u_star}, showing that $u_1$ acts as a blocker of $v$ in the \Franking{} algorithm. Our lemma generalizes this result to all precedent odd-indexed vertices and arbitrary randomized greedy matching algorithms.

As we can see from this lemma, the longer the alternating path between $v$ and $u$, the more compensation $u$ and $u^\ast$ (equivalent to $v$) have a chance to collect. We explicitly use this feature in our analyses. In \ranking{}, we show that except in an extreme case\footnote{See \cref{fact:theta_0} for the extreme case in \ranking{}.}, the alternating path will have length $\geq 4$, so both $u$ and $u^\ast$ get to collect at least two copies of compensation when they are unmatched, respectively. In \Franking{}, we implicitly use the fact that length of the alternating path is $\geq 4$ so that we can collect some extra gain for $u$ (when it is unmatched) without assuming $u^\ast$ is the one paying compensation, as if the alternating path has length $2$, then $u^\ast$ is the only provable blocker of $u$, and hence we can only collect compensation for $u$ from $u^\ast$. In such a case, we are equivalently receiving no extra gain.

We can also prove a restricted variant of \cref{lem:victim_querycommit_case_u_star}, which shows that the match of $u$ is a blocker of $v$, assuming $v\in N(u)$ is unmatched and favored by $u$ over the backup vertex $b$ (when $u$ has one). Although this variant only shows the existence of one blocker vertex, the benefit is that we no longer need the existence of a list $L$ in which $v$ makes $u$ worse off. This fact was used by~\cite{0.531RDO} in its victim analysis.
\begin{fact}
\phantomsection
\label{claim:victim_when_no_alt_path}
Assuming $v\in N(u)$ is unmatched in $\M(L')$ and that $u$ is matched to $w$ and does not have a backup in $L_{\shortminus v}'$. Then $v$ is the victim of $w$ in $L'$.

Assuming $v\in N(u)$ is unmatched in $\M(L')$ and that $u$ is matched to $w$ and has a backup $b$ in $L_{\shortminus v}'$. Then $v$ is the victim of $w$ in $L'$, provided that $(u,v)<(u,b)$ in $L'$.
\end{fact}
See a proof in \nameref{sec:appendix-a}.

\section{\ranking}
\phantomsection
\label{sec:ranking}
\ranking{} was first introduced in \cite{karp1990} for the online bipartite matching problem with one-sided vertex arrival. Its general graph version was first studied in \cite{flawedranking} and later shown to achieve an approximation ratio of $0.5469$ through subsequent analyses~\cite{0.523ranking, 0.526ranking, 0.546ranking}. The algorithm works as follows:

\begin{algorithm}[H]
  \caption{\ranking{} algorithm for general graphs.}
  \phantomsection\label{alg:ranking}
  \KwIn{Graph $G=(V,E)$}
  \SetAlgoLined

  Sample a uniform random value $x_v \in [0,1]$ independently for each $v \in V$\;

  \For{each vertex $v \in V$ in increasing order of $x_v$ in $\vecx$}{
    \If{$v$ is unmatched and has an unmatched neighbor}{
      Match $v$ with the first unmatched neighbor $u$ (in the same order $\vecx$)\;
    }
  }
\end{algorithm}

Let $\M(\vecx)$ be the output matching of running \ranking{} with rank vector $\Vec{x}$. This algorithm can also be viewed as running the \nameref{alg:querycommit} on the list of pairs of vertices $(u,v)$ ordered lexicographically by $(x_u,x_v)$. For a set $S\subseteq V$, we denote $\vecx_{\shortminus S}$ as the induced rank vector with nodes in $S$ removed from $\vecx$. For ranks $x_1,\dots,x_k \in [0,1]$, we write $\vecx_{\shortminus S}(s_1=x_1,\dots,s_k=x_k)$ to denote the rank vector obtained by adding nodes $s_i \in S$ with rank $x_i$ to $\vecx_{\shortminus S}$. If $S$ is a singleton $\{v\}$, we use $\vecx_{\shortminus v}$ to denote the rank vector with $v$ removed and $\vecx_{\shortminus v}(x)$ to denote the rank vector with $v$ moved to rank $x$, respectively. When $S=\{u,v\}$, we write $\vecx_{\shortminus uv}$ as shorthand for $\vecx_{\shortminus \{u,v\}}$.

We can view \ranking{} as the VI randomized greedy matching algorithm in which both the vertex arrival order $\pi$ and the preference orders $\sigma(v)$ are given by the same permutation induced by $\vecx$. The following notion of activeness/passiveness of a vertex is commonly used to describe whether the matching of a vertex occurs at its arrival time or at another vertex's arrival time.

\begin{definition}[Passive vs.\ Active Vertices]
\phantomsection
\label{def:active_passive_vertices}
\mbox{}\\[-1em]

\textbf{\textup{(Passive)}} A vertex $u$ is passively matched to $v$ if $(u,v)$ is matched at $v$'s arrival time $\pi(v)$.

\textbf{\textup{(Active)}}\; A vertex $u$ is actively matched to $v$ if $(u,v)$ is matched at $u$'s arrival time $\pi(u)$.

\end{definition}
\subsection{Our Approach vs. Previous Approaches}
In this subsection, we briefly discuss how our approach differs from the randomized primal-dual method in \cite{0.546ranking}, and the additional techniques we added to the victim gain-sharing scheme compared to~\cite{0.521Franking,0.531RDO}.

To analyze the expected value of gain$(u)$ $+$ gain$(u^\ast)$, \cite{0.546ranking} introduced the notion of the \emph{profile} of $u$ (\cref{def:profile_ranking}), which classifies $\vecx_{\shortminus u^\ast}$. All $\vecx_{\shortminus u^\ast}$ with the same profile are treated identically. The authors first computed the worst-case expected gain for each profile, and then integrated it over the worst-case profile distribution to obtain the final expected gain. Our approach preserves this overall structure but improves the bound for each fixed profile by introducing the notion of victim and a new monotonicity property that was not considered in~\cite{0.546ranking}. Specifically, fixing $\vecx_{\shortminus u^\ast}$, previous work conducted a worst-case analysis for each $\vecx_{\shortminus u^\ast}(x)$ in isolation. In contrast, we analyze the worst-case matching conditions of $\M(\vecx_{\shortminus u^\ast}(x))$ holistically across $x\in[0,1]$, which allows us to boost the overall expectation
$$\ev_{x\sim U[0,1]}[\text{ gain($u$) + gain$(u^\ast)$ }\;|\; \vecx=\vecx_{\shortminus u^\ast}(x)].$$
Rather than directly adopting the victim definition and analysis from~\cite{0.521Franking,0.531RDO}, our analysis refines the victim gain-sharing mechanism. Previous analyses defined the compensation rule as a function of the rank of a single vertex, whereas we make it depend on the ranks of multiple vertices, thereby capturing more structural information.

\paragraph{Global Analysis via Victim.}
 One common worst-case matching scenario \ranking{} needs to consider is when $u$ or $u^\ast$ is unmatched in $\M(\vecx_{\shortminus u^\ast}(x))$. When this happens for more $x\in[0,1]$, the primal-dual method tends to yield a worse approximation ratio. If this worst-case scenario occurs for all $x\in[0,1]$, then without additional structural analysis, the best approximation ratio one can obtain is $0.5$. By enforcing a gain-sharing rule that forces blockers to pay some of their gain to the victim, extra gains are collected for $u$ and $u^\ast$ when they are unmatched, effectively alleviating this worst-case matching scenario. In terms of the final $LP$, introducing victim gain-sharing is equivalent to reducing the estimated gain at some $x$ values while boosting it at some other $x$ values. Because in expectation vertices receive more gain than lose it, we obtain an increase in the final approximation ratio.
 
However, naively incorporating victim analysis into the LP model of~\cite{0.546ranking} increases the number of constraints by a factor of $\Omega(n)$\footnote{Here, $n$ denotes the number of intervals into which we partition $[0,1]$ in order to numerically approximate the LP system.} during the final discretization step. Due to computational limitations, the benefit of compensation is largely offset by the increase in LP size. To address this issue, we introduce a monotonicity constraint on the rank of the match of $u^\ast$, which allows us to incorporate victim analysis using asymptotically the same number of constraints as in~\cite{0.546ranking}. Furthermore, this monotonicity constraint rules out one of the worst-case matching scenarios considered in the previous analysis.

\paragraph{Monotonicity for the match of $u^\ast$ for Increasing $x$.}
We observed that the match of $u^\ast$ worsens as the rank of $u^\ast$ increases (\cref{fact:matching_guarantee_ranking}, property 4). In \cite{0.546ranking}, the rank of the match of $u^\ast$ is not assumed to be strongly correlated to increasing $x$ values, hence they incorporated case-by-case constraints for all different $x$ values in $[0,1]$. This is the main reason why naively incorporating the victim analysis results in an $\Omega(n)$ times increase in the number of constraints. By introducing the monotonicity constraint, we were able to partition the unit interval into at most three sub-intervals $[a_i,b_i]$, where the match of $u^\ast$ at any $x\in[a_i,b_i]$ is determined by the match of $u^\ast$ at $b_i$. Hence in our final approximation, we no longer have to consider the match of $u^\ast$ for all $x\in[0,1]$ but instead only consider it at a few $x$ values.

This property also allows us to rule out a particular worst-case matching scenario for certain $x_1 < x_2$ pairs. Previous analyses considered a possible matching combination where $u^\ast$ is matched to a worse vertex $w_1$ in $\vecx_{\shortminus u^\ast}(x_1)$, while the same vertex $u^\ast$ is matched to a better vertex $w_2$ in $\vecx_{\shortminus u^\ast}(x_2)$. Using the monotonicity property, we can directly rule out such matching combinations, as $w_1$ and $w_2$ must satisfy $x_{w_1}\leq x_{w_2}$, given $x_1<x_2$.

\paragraph{Multivariate Compensation Function.}
When a blocker $w$, matched to a vertex $z$, pays compensation to a victim, previous papers~\cite{0.521Franking,0.531RDO} let $w$ transfer $h(y)$ amount of its gain to the victim, where $y$ is the rank of either $z$ or $w$. This design arises because their analyses have limited information about the rank of the other vertex in the pair $(w,z)$. For example, in the \Franking{} analysis~\cite{0.521Franking}, little information can be inferred about the rank of the active vertex in the pair $(w,z)$. In \ranking{}, a vertex’s rank affects both the decision order and the relative preference order: the former is related to the behavior of a vertex when it is matched actively, while the latter is related to the behavior when the vertex is matched passively. We exploit this additional structure. In particular, we define the compensation function $h$ as a bivariate function depending on the ranks of both endpoints. This allows us to incorporate more information about the blocker and its match, leading to a tighter bound on the overall approximation ratio.

\paragraph{Refined Bounds with Longer Alternating Paths.}
In this section, we provide a system of bounds that effectively treats all alternating paths of length $\geq 4$ as having length exactly $4$. Later, in \cref{sec:tigher_bounds}, we perform a refined analysis that also considers the case where alternating paths have length $6$. It turns out that we can already prove an approximation ratio of $0.559$ without this refinement. Since the refinement improves the approximation ratio by only around $0.001$ while requiring intricate structural analysis, we defer it to \cref{sec:tigher_bounds}.

\paragraph{Large Odd Girth Implies Long Alternating Paths.}
For any graph $G$, it turns out that the alternating paths of interest need to have at least length $2k$, where $2k+1$ is the odd girth of $G$. This implies that for graphs of odd girth at least $2k+1$, we can assume $u,u^*$ are victims of at least $k$ blockers, hence receiving at least $k$ copies of compensations. We exploit this fact in \cref{sec:odd_girth_graphs} and design stronger LPs for graphs of larger odd girths.

\subsection{\ranking{} Structural Properties}
Before formally introducing the analysis, we present a few important structural properties of \ranking{} that will be used throughout this section.

We first describe a special property of the \ranking{} algorithm as a corollary of the alternating path lemma: the ranks $x_{u_i}$ of even-indexed and odd-indexed vertices along the path monotonically increase.
\begin{corollary}[Special Property for Alternating Path in \ranking{}~\cite{0.546ranking}]
\phantomsection
\label{fact:monotonicity_for_alt_path_ranking}
    Let $\vecx$ be a rank vector on $V$. For any $v\in V$, denote the alternating path $\M(\vecx)\oplus\M(\vecx_{\shortminus v})$ as $u_0,u_1,...,u_k$. Then for any $0 \le i \le k-2$, we have $x_{u_i} < x_{u_{i+2}}$.
\end{corollary}
\begin{proof}
    By \cref{lem:alt-path} part 3, for each $0\leq i\leq k-2$, $(u_i,u_{i+1})$ has an earlier query time than $(u_{i+1},u_{i+2})$. Since \ranking{} queries $(u_i,u_{i+1})$ at time $\min\{(x_{u_i},x_{u_{i+1}}),(x_{u_{i+1}},x_{u_{i}})\}$ and $(u_{i+1},u_{i+2})$ at time $\min\{(x_{u_{i+1}},x_{u_{i+2}}),(x_{u_{i+2}},x_{u_{i+1}})\}$, we have that $(u_i,u_{i+1})$ has an earlier query time than $(u_{i+1},u_{i+2})$ if and only if $x_{u_{i}}<x_{u_{i+2}}$ by the nature of lexicographic orderings.
\end{proof}
The notion of being worse off in \ranking{} is simplified due to the above property. In short, being worse off is equivalent to being matched to a vertex with larger rank.
\begin{fact}
    $u$  becomes worse off by introducing $v$ into $\vecx_{\shortminus v}$ is equivalent to $u$ being matched to a larger-ranked vertex in $\M(\vecx)$ (or unmatched) compared to $\M(\vecx_{\shortminus v})$, which is also equivalent to $u$ being an even-indexed vertex in $\M(\vecx)\oplus \M(\vecx_{\shortminus v})$.
\end{fact}
We state the following lemma for matching guarantees with respect to $\vecx_{\shortminus u^\ast}(x)$. Properties $1$–$3$ are proved in~\cite{0.546ranking}, while property $4$ is a commonly used monotonicity property in various \ranking-like algorithms~\cite{0.526ranking,0.531RDO}:
\begin{fact}[Matching Guarantees for \ranking{}]
\phantomsection
\label{fact:matching_guarantee_ranking}
    Let $x$ be the rank of $u^\ast$ when $u^\ast$ is inserted into $\vecx_{\shortminus u^\ast}$. Assuming $u$ is matched to $v$ in $\M(\vecx_{\shortminus u^\ast})$, then:
    \begin{enumerate}
        \item If $x<x_v$, then $u^\ast$ is matched to a vertex of rank $\leq x_u$ in $\M(\vecx_{\shortminus u^\ast}(x))$.
        \item If $x>x_u$, then $u$ is matched to a vertex of rank $\leq x_v$ in $\M(\vecx_{\shortminus u^\ast}(x))$.
        \item If $u$ is made worse off by introducing $u^\ast$, then $u^\ast$ is matched to a vertex of rank $\leq x_v$ in $\M(\vecx_{\shortminus u^\ast}(x))$.
        \item If $u^\ast$ is matched to a vertex $w$ in $\M(\vecx_{\shortminus u^\ast}(x_0))$, then $u^\ast$ is matched to a vertex of rank $\leq x_w$ in $\M(\vecx_{\shortminus u^\ast}(x))$ for all $x\in[0,x_0]$. In other words, demoting the rank of a vertex can only make it match to a worse vertex.
    \end{enumerate}
\end{fact}
\begin{proof}
    We prove here only the part not presented in \cite{0.546ranking}, namely property $4$. For the proof of properties $1$–$3$, see \nameref{sec:appendix-a}. 

Assume $u^\ast$ is matched to some vertex $w$ when $u^\ast$ is placed at rank $x$. We show that demoting the rank of $u^\ast$ to $x_0 > x$ cannot make it better off. There are two cases: either $u^\ast$ is picked by $w$ in $\M(\vecx_{\shortminus u^*}(x))$ ($u^\ast$ is passive) or $w$ is picked by $u^\ast$ in $\M(\vecx_{\shortminus u^*}(x))$ ($u^\ast$ is active). 

Case $u^\ast$ being passively matched: Any vertex $v$ with rank $x_v < x_w$ does not have $u^\ast$ as its smallest ranked neighbor when it makes its choice. Since demoting $u^\ast$ only makes it less preferred, these $v$ still do not match $u^\ast$ in $\M(\vecx_{\shortminus u^\ast}(x_0))$. This implies that $x_w$ is the best possible match for $u^\ast$ in this case.

Case $u^\ast$ being actively matched: If $u^\ast$ actively matches $w$ in $\M(\vecx_{\shortminus u^*}(x))$, demoting $u^\ast$ does not change the choices made by any vertex $v$ with rank $x_v < x$. At time $(x,x)$, the best available neighbor for $u^\ast$ is $w$ with rank $x_w$. Since $u^\ast$ is matched in $\vecx_{\shortminus u^\ast}(x_0)$ at time later than $(x,x)$ (if it is matched at all), its match is at best $w$ with rank $x_w$.
\end{proof}
These guarantees describe to what extent $u$ and $u^\ast$ are matched in $\M(\vecx_{\shortminus u^\ast}(x))$. We will build our approximation constraints based on these matching guarantees.

\subsection{Gain-Sharing Scheme}\label{sec:gainsharingranking}
In the randomized primal-dual framework, for each edge $\ranking{}$ matches, we distribute one unit of gain among vertices in $V$. We introduce the notions of \emph{product} and \emph{buyer}, recap the definition of \emph{victim} (\cref{def:victim_querycommit}) in the language of \ranking{}, and define our gain-sharing rules:
\begin{definition}[Product, Buyer, and Victim in \ranking{}]
    For a rank vector $\vecx$ on vertices $V$ and bi-partition $\chi:V\to\{P,B\}$, we define product and buyer as the following:
    \begin{itemize}
        \item $v$ is a product if $\chi(v)=P$.
        \item $v$ is a buyer if $\chi(v)=B$.
    \end{itemize}
    We say $v$ is a victim of $w$ if the following holds:
    \begin{itemize}
        \item $v$ is unmatched in $\M(\vecx)$.
        \item Removing $w$ from the graph makes $v$ matched, i.e. $v$ is matched in $\M(\vecx_{\shortminus w})$.
    \end{itemize}
\end{definition}
And we perform gain-sharing as follows:
\begin{definition}[Gain-Sharing Rules]
\phantomsection
\label{gain_sharing_scheme_ranking}
     For any gain function $g:[0,1]\times[0,1]\to[0,1]$, compensation function $h:[0,1]\times[0,1]\to[0,1]$, rank vector $\vecx$ and bi-partition $\chi$, we perform gain-sharing as follows:
        \begin{enumerate}
        \item When $u$ matches $v$ with $u$ as buyer and $v$ as product, $u$ gets $1-g(x_u,x_v)$ amount of gain and $v$ gets $g(x_u,x_v)$ amount of gain.
        \item If $w$ matches $z$ and further has a victim $v$, $w$ sends $h(x_w, x_z)$ amount of compensation to $v$. In the case where $w$ is a buyer, gain of $w$ is $1-g(x_w,x_z)-h(x_w, x_z)$, and in the case where $w$ is a product, the gain of $w$ is $g(x_z,x_w)-h(x_w, x_z)$. $v$ gets $h(x_w, x_z)$ amount of extra gain as a result. The gain of an unmatched vertex is the sum of all the compensation it receives.
    \end{enumerate}
\end{definition}
In \cite{0.546ranking}, fact 6.8, the authors proved that there exists a randomized bi-partitioning $\chi$ such that gain-sharing rule $1$ is well defined with respect to $\chi$, that is, any pair of matched vertices are always a (buyer, product) pair. Further, $\chi$ has the good property that each pair of perfect match $(u,u^\ast)$ is also a (buyer, product) pair, and the assignment of buyer/product role for each vertex is independent from the underlying rank vector $\vecx$.

The reason these properties hold is that the randomized bi-partition $\chi$ is performed after the matching process. For each fixed $\vecx$, because $\M(\vecx)$ and $M^\ast$ are both matchings, their union $\M(\vecx)\cup M^\ast$ is a bipartite graph. So we can randomly assign each side of the bipartite graph as buyer/product with $\frac{1}{2}$ probability. Because we perform random role assignment with respect to each $\vecx$, the role assignment (for each vertex) is thus independent from $\vecx$. It is not hard to show that we can add gain-sharing rule 2 above gain-sharing rule $1$ and still have the gain-sharing rule well defined.
\begin{fact}[Randomized Bi-Partition]
\phantomsection\label{fact:randomized_bipartition}
    There exists a randomized bi-partitioning $\chi$ such that the worst-case expected value of gain($u$) $+$ gain$(u^\ast)$  after performing gain-sharing rule 1,2 and conditioning on the event
    $$\text{$u$ is a buyer $\wedge$ $u^\ast$ is a product $\wedge$ the match of $u$ is a product $\wedge$ the match of $u^\ast$ is a buyer,}$$
    lower bounds the approximation ratio of \ranking{}. That is, denoting $PD(v),\; BU(v)$ as the event $v$ is a product, buyer respectively, we have 
    \[
\min_{(u,u^*)\in M^*}
\left\{
\ev_{\vecx}\!\left[
\text{gain}(u) + \text{gain}(u^\ast)
\,\middle|\,
\begin{aligned}
&PD(u^\ast)\wedge BU(\text{match of $u^\ast$}) \\
&\wedge PD(\text{match of $u$})\wedge BU(u)
\end{aligned}
\right]
\right\}
\leq \frac{\ev_{\vecx}[|\M(\vecx)|]}{|M^*|}
\]
    Further, the conditional distribution $\vecx$ for any pair of $(u,u^*)$ is the same as the initial uniform random rank vector.
\end{fact}
The construction of $\chi$ is the same as that in \cite{0.546ranking}, Definition 6.3. For anyone who's interested, see a construction and proof sketch in \nameref{sec:appendix-a}. \textbf{From now on, we will always assume $u$ is a buyer, $u^\ast$ is a product, the match of $u$ is a product and the match of $u^\ast$ is a buyer.}
\paragraph{}
We will also apply the following restrictions to our gain-sharing function $g$ and compensation function $h$:
\begin{definition}[Function Constraints]
\phantomsection
For \ranking{}, we assume the gain function $g$ and the compensation function $h$ satisfy the following constraints:
\label{def:function_constraint_ranking}
    \begin{enumerate}
    \item For each fixed $x_0$, $g(x_0,y),h(x_0,y)$ are monotonically increasing with respect to $y$.
    \item For each fixed $y_0$, $g(x,y_0),h(x,y_0)$ are monotonically decreasing with respect to $x$.
    \item For all $x$, $h(x,0)=0$,
    \item $\forall x,y$ we have $1-g(x,y)-h(x,y)\geq 4h(0^+,1)$, where $h(0^+,1)=\lim_{x\to 0^+} h(x,1)$.
    \item $\forall x,y$ we have $g(x,y)-h(y,x)\geq 4h(0^+,1)$.
\end{enumerate}
\end{definition}
The first two constraints correspond to the intuition that smaller-ranked vertices are preferred in \ranking{}, as illustrated by the following fact.
\begin{fact}[Monotonicity of Gains]
\phantomsection
\label{fact:monotonicity_of_gains1}
    If vertex $u$ is matched to a vertex of rank $\leq x$, then 
    \begin{itemize}
        \item If $u$ is a product, gain$(u)\geq g(x,x_u)-h(x_u,x)$.
        \item If $u$ is a buyer, gain$(u)\geq 1-g(x_u,x)-h(x_u,x)$.
    \end{itemize}
\end{fact}
\begin{proof}
    This is because the functions $g(x,x_u), 1-g(x_u,x),$ and $-h(x_u,x)$ monotonically decrease as $x$ increases.
\end{proof}
With this property, we define a lower bound for $u$ being a buyer/product while paying compensation as the following functions:
\begin{definition}[Gain Functions Assuming Compensation]
\phantomsection
\label{def:gain_with_comp_function}
Let $g$ and $h$ be the gain and compensation functions. Define $g_B(x_u,x_v),\;g_P(x_u,x_v)$ as the functions representing the gain of a buyer/product vertex $u$ when matched to vertex $v$ assuming $u$ pays compensation. That is:
\begin{itemize}
    \item $g_B(x_u,x_v)=1-g(x_u,x_v)-h(x_u,x_v).$
    \item $g_P(x_v,x_u)=g(x_v,x_u)-h(x_u,x_v).$
\end{itemize}
\end{definition}
\begin{corollary}[Monotonicity of Gains]
\phantomsection
\label{fact:monotonicity_of_gains}
    By \cref{fact:monotonicity_of_gains1}, if vertex $u$ is matched to a vertex of rank $\leq x$, then
    \begin{itemize}
        \item If $u$ is a buyer, gain$(u)\geq g_B(x_u,x)$.
        \item If $u$ is a product, gain$(u)\geq g_P(x,x_u)$.
    \end{itemize}
\end{corollary}
The third constraint in \cref{def:function_constraint_ranking} is used to avoid over-approximating compensation in some degenerate cases. It will be clear when we encounter these specific cases. Constraints $4,\;5$ allow us to assume that being matched is always better than being unmatched, assuming the unmatched vertex receives $\leq 4$ copies of compensation. In our preliminary analysis, we will assume each unmatched vertex receives at most $2$ copies of compensation. In the refined analysis in \cref{sec:tigher_bounds}, we will assume each unmatched vertex receives at most $4$ copies of compensation. These bounds allow us to take being unmatched as the worst-case scenario when a vertex could be either matched or unmatched.
\begin{fact}[Assume Unmatched when Possible]
\phantomsection
\label{fact:assume_unmatched_worst_case}
    By \cref{def:function_constraint_ranking}, constraints $4$ and $5$, when a vertex $u$ can be either matched or unmatched, we can always calculate a lower bound for $gain(u)$ assuming $u$ is unmatched if the collected compensation is $\leq4$ copies. 
\end{fact}
\begin{proof}
    W.L.O.G., we can assume that no vertex $v$ has rank $x_v=0$, as the probability measure of such an event is $0$. Then by the monotonicity constraints of $h$, the largest value of one copy of compensation is $h(0^+,1)$. Since for any $x,y\in[0,1]$, we set both $1-g(x,y)-h(x,y)$ (minimum gain when a vertex is matched actively) and $g(x,y)-h(y,x)$ (minimum gain when a vertex is matched passively) to be at least as large as $4h(0^+,1)$, we can always use $\leq 4$ copies of compensations $\sum_i h(x_i,y_i)$ for arbitrary $x_i\in (0,1],\;y_i\in (0,1],$ to lower bound the expected gain of $u$ when it is matched.
\end{proof}

\subsection{Approximating With Profiles}
We lower bound the expected value of gain($u$) $+$ gain$(u^\ast)$ through the randomization of $u$ and $u^\ast$ successively. We first fix each $\vecx_{\shortminus u^\ast}$ and calculate the expected gain over $x\sim U[0,1]$ where $x$ is the rank of $u^\ast$. Then we take the expectation over the worst-case distribution of $\vecx_{\shortminus u^\ast}$. i.e., we perform the following approximations:
$$\ev[\text{gain}(u) + \text{gain}(u^\ast)]=\ev_{\vecx_{\shortminus u^\ast}}\left[\vphantom{\int_0^1} \ev_{x\sim U[0,1]}\left[\text{gain}(u) + \text{gain}(u^\ast)\mid \vecx=\vecx_{\shortminus u^\ast}(x)\right]\right].$$
For each fixed $\vecx_{\shortminus u^\ast},$ we collect up to 3 pieces of information about the matching condition of $u$, namely the rank of $u$, the rank of its match, and the rank of its backup. The collection of these pieces of information is defined as the \emph{profile}~\cite{0.546ranking} of $u$ in $\vecx_{\shortminus u^\ast}$. Before defining the notion of profile, we first define the notion of backup of a vertex for \ranking{}. It is the same as \cref{def:backup_querycommit}, we restate it here in the language of rank vector $\vecx$.
\begin{definition}[Backup]
    For any rank vector $\vecx_{\shortminus u^\ast}$, if $u$ is matched to some vertex $v$ in $\M(\vecx_{\shortminus u^\ast})$, we say a vertex $b$ is the backup vertex of $u$ if $u$ is matched to $b$ when $v$ is removed. I.e., $(u,b)\in\M(\vecx_{\shortminus vu^\ast})$. If there's no such $b$, we say $u$ does not have a backup.
\end{definition}
Backup describes to what extent introducing another vertex can make a vertex becomes worse off. In fact, the backup always has a larger rank than the current match, and further that introducing a new vertex into the graph will only cause $u$ to match to its backup in the worst-case.
\begin{fact}[\cref{fact:backup_uniquely_describe_worse_off} in \ranking{}(claim 4.3 in~\cite{0.546ranking})]
\phantomsection
\label{facf:backup_ranking}
    If $u$ is matched to $v$ and has a backup $b$ in $\vecx_{\shortminus u^\ast}$, then $x_v< x_b$ and if $u$ is made worse off in $\M(\vecx)$, it has to be matched to $b$.

    If $u$ is matched to $v$ and does not have a backup $b$ in $\vecx_{\shortminus u^\ast}$, then $u$ being worse off in $\M(\vecx)$ implies that $u$ is unmatched.
\end{fact}
\begin{proof}
    Most parts of the fact are just a direct translation of \cref{fact:backup_uniquely_describe_worse_off} except for the observation $x_v<x_b$. To show $x_v<x_b$, by \cref{lem:backup_in_alt_path}, we know that $v$ and $b$ are some pair of vertices $u_i$ and $u_{i+2}$ in the alternating path $\M(\vecx_{\shortminus u^\ast})\oplus \M(\vecx)$. By the monotonicity property (\Cref{fact:monotonicity_for_alt_path_ranking}), we have $x_v<x_b$. 
\end{proof}
We are now ready to define the \emph{profile}.
\begin{definition}[Profile of $u$~\cite{0.546ranking}]
\phantomsection\label{def:profile_ranking}
    The profile of $u$ with respect to $\vecx_{\shortminus u^\ast}$ is a triplet $(x_u,x_v,x_b)$ such that
    \begin{enumerate}
        \item $x_u$ is the rank of $u$.
        \item $x_v$ is the rank of $v$, the match of $u$ in $\M(\vecx_{\shortminus u^\ast})$. If $u$ is unmatched, $x_v=\bot$.
        \item $x_b$ is the rank of $b$, the backup of $u$ with respect to $\vecx_{\shortminus u^\ast}$. If $u$ does not have a backup, $x_b=\bot$.
    \end{enumerate}
\end{definition}
We will first provide lower bounds for the expected value of gain($u$) $+$ gain$(u^\ast)$ for each fixed profile of $u$. For each fixed profile $(x_u,x_v,x_b)$, we will devise a function $G$ such that
$$G(x_u,x_v,x_b)\leq\ev\left[ \text{ gain($u$) $+$ gain($u^\ast$)}\;|\; \text{profile of $u$ in $\vecx_{\shortminus u^\ast}$ is $(x_u,x_v,x_b)$}\right].$$
After deriving the case-by-case lower bound, we then take the expected lower bound by considering the worst-case distribution of profiles $D$ to lower bound the expected gain.
$$\ev_{(x_u,x_v,x_b)\sim D}\left[G(x_u,x_v,x_b)\right]\leq \ev\left[\text{ gain}(u) + \text{gain}(u^\ast)\;\right].$$
The profiles are divided into three cases and we devise lower bounds for them separately:
\begin{itemize}
    \item Case $u$ not matched, with profile $(x_u,\bot,\bot)$.
    \item Case $u$ matched with $v$ but has no backup, with profile $(x_u,x_v,\bot)$.
    \item Case $u$ matched with $v$ and has backup $b$, with profile $(x_u,x_v,x_b)$.
\end{itemize}
Next, we prove a very useful monotonicity about profiles as follows. Intuitively, when $u$ is matched to $v$ and has a backup $b$, $b$ is the second-best choice of match for $u$, hence, even if we demote the rank of $v$, as long as it is still better than $b$, $u$ will continue to match with $v$, thus resulting in the same overall matching. The following monotonicity properties formally establish this intuition.

\begin{lemma}[Monotonicity Properties for Profiles {\protect\cite[Corollaries 4.8 and 4.10]{0.546ranking}}]
\phantomsection
\label{monoto_before_backup_ranking}
Fix $x_u,x_b\in[0,1]$. Suppose $u$ is matched to $v$ with rank $x_v$ in $\M(\vecx_{\shortminus u^*})$.
\begin{itemize}
    \item If $u$ does not have a backup, then for any $x \in [x_v,1]$, demoting the rank of $v$ to $x$ does not change the resulting matching. That is,
    $$\M(\vecx_{\shortminus vu^*}(v=x))=\M(\vecx_{\shortminus vu^*}(v=x_v))$$
Consequently, the probability density function of $x_v$ is monotonically increasing on $[0,1]$ with respect to profiles $(x_u,x_v,\bot)$. That is, for any $0 \le x_1 < x_2 \le 1$,
\[
\Pr\;\!\left[\vecx_{\shortminus u^\ast}
\text{ has profile } (x_u,x_1,\bot)\right]
\le
\Pr\;\!\left[\vecx_{\shortminus u^\ast}
\text{ has profile } (x_u,x_2,\bot)\right].
\]

\item If $u$ has a backup $b$ at rank $x_b$, then for any $x \in [x_v,x_b)$, demoting the rank of $v$ to $x$ does not change the resulting matching. That is,
$$\M(\vecx_{\shortminus vu^*}(v=x))=\M(\vecx_{\shortminus vu^*}(v=x_v))$$
Consequently, the probability density function of $x_v$ is monotonically increasing on $[0,x_b)$ with respect to profiles $(x_u,x_v,x_b)$. In particular, for any
$0 \le x_1 < x_2 < x_b$,
\[
\Pr\;\!\left[\vecx_{\shortminus u^\ast}
\text{ has profile } (x_u,x_1,x_b)\right]
\le
\Pr\;\!\left[\vecx_{\shortminus u^\ast}
\text{ has profile } (x_u,x_2,x_b)\right].
\]
\end{itemize}
\end{lemma}
\begin{proof}
    We give a simpler proof than \cite{0.546ranking}. Observe that $\vecx_{\shortminus u^\ast}$ can be viewed as inserting $v$ at rank $x_v$ into $\vecx_{\shortminus vu^\ast}$. Since $u^\ast$ plays no role in this lemma, we omit it and write $\vecx_{\shortminus v}$ for $\vecx_{\shortminus vu^\ast}$.

Suppose $u$ is matched to $v$ in $\M(\vecx_{\shortminus v}(x_v))$. We prove that $\M(\vecx_{\shortminus v}(x))= \M(\vecx_{\shortminus v}(x_v))$ for any $x > x_v$, provided $x < x_b$ when $u$ has a backup $b$.

First note that if $(u,v) \in \M(\vecx_{\shortminus v}(x_v))$, then
\[
\M(\vecx_{\shortminus v}(x_v)) - \{(u,v)\}
= \M(\vecx_{\shortminus uv}),
\]
since the remaining partial matching is equal to the matching with both $u$ and $v$ removed from the beginning. The same holds for $\M(\vecx_{\shortminus v}(x))$. Therefore, once we show $(u,v)$ is also matched in $\M(\vecx_{\shortminus v}(x))$, it follows that
\[
\M(\vecx_{\shortminus v}(x))
= \M(\vecx_{\shortminus uv}) \cup \{(u,v)\}
= \M(\vecx_{\shortminus v}(x_v)).
\]

Now it suffices to show that $(u,v)$ is matched in $\M(\vecx_{\shortminus v}(x))$. Assume $x$ is as stated. Notice that if $u$ has a backup $b$, then it is matched to $b$ in $\M(\vecx_{\shortminus v})$, if it does not have a backup, then it is unmatched in $\M(\vecx_{\shortminus v})$. By \cref{fact:matching_guarantee_ranking} property $1$ and \cref{matching_guarantee_ranking_no_macth}, it follows that $v$ is matched to a vertex with rank $\leq x_u$ in $\vecx_{\shortminus v}(x)$. On the other hand, by \cref{fact:matching_guarantee_ranking} property $4$, since $v$ is matched to $u$ in $\vecx_{\shortminus v}(x_v)$, its match at $x\geq x_v$ is at most as good as $u$. So $v$ is matched to a vertex with rank $\geq x_u$ in $\vecx_{\shortminus v}(x)$. The two facts combined show that $v$ is matched to a vertex with rank exactly $x_u$. We can also assume that no two vertices share the same rank from the beginning, as such an event has zero measure in the probability space. Therefore, we conclude that $v$ is matched to $u$ in $\vecx_{\shortminus v}(x)$.
\end{proof}

\def\FigShift{1.5cm} 

\makebox[\linewidth][c]{%
\hspace*{\FigShift}%
\begin{tikzpicture}[
  even/.style={circle, fill=blue!60!white, inner sep=1.5pt},
  odd/.style={circle, fill=blue!60!white, inner sep=1.5pt},
  legendtext/.style={font=\footnotesize, align=left},
]

\def\xLR{0.90}
\def\sep{6}
\def\legGap{3.0}   
\def\legTextW{3.4}

\pgfmathsetmacro{\yStep}{(2/1.7320508)*\xLR}

\pgfmathsetmacro{\xA}{-\sep/2}
\pgfmathsetmacro{\xB}{\sep/2}
\pgfmathsetmacro{\xL}{\xB + \legGap}

\pgfmathsetmacro{\sysLeft}{\xA - \xLR }
\pgfmathsetmacro{\sysRight}{\xL + \legTextW }
\pgfmathsetmacro{\sysCenter}{(\sysLeft+\sysRight)/2}
\pgfmathsetmacro{\textWidth}{\sysRight - \sysLeft}

\begin{scope}[shift={({-\sysCenter},0)}]
\begin{scope}[xshift= 1cm]
\node[odd]  (uA) at ({\xA-\xLR}, {-1*\yStep}) {};
\node[anchor=east] at ({\xA-\xLR-0.25}, {-1*\yStep}) {$u$};

\node[even] (vA) at ({\xA+\xLR}, {0}) {};
\node[anchor=west] at ({\xA+\xLR+0.25}, {0}) {$v$};

\node[even] (bA) at ({\xA+\xLR}, {-2*\yStep}) {};
\node[anchor=west] at ({\xA+\xLR+0.25}, {-2*\yStep}) {$b$};

\draw[thick] (uA) -- (vA);
\draw[dash pattern=on 6pt off 3pt, thick] (uA) -- (bA);

\node[font=\footnotesize] at (\xA, {-2*\yStep - 0.38}) {$v$ at $x_v$};

\node[odd]  (uB) at ({\xB-\xLR}, {-1*\yStep}) {};
\node[anchor=east] at ({\xB-\xLR-0.25}, {-1*\yStep}) {$u$};

\node[even] (vB) at ({\xB+\xLR}, {0}) {};
\node[anchor=west] at ({\xB+\xLR+0.25}, {0}) {$v$};

\node[even] (bB) at ({\xB+\xLR}, {-2*\yStep}) {};
\node[anchor=west] at ({\xB+\xLR+0.25}, {-2*\yStep}) {$b$};

\node[even] (vpB) at ({\xB+\xLR}, {-1.35*\yStep}) {};
\node[anchor=west] at ({\xB+\xLR+0.25}, {-1.35*\yStep}) {$v$};

\draw[thick] (uB) -- (vB);
\draw[dash pattern=on 6pt off 3pt, thick] (uB) -- (bB);
\draw[thick] (uB) -- (vpB);

\node[font=\footnotesize] at (\xB, {-2*\yStep - 0.38}) {Demoting $v$ to $x\in[x_v,x_b)$};

\draw[->, thick]
  ([xshift=4pt,yshift=-4pt]vB.south) -- ([xshift=4pt,yshift=4pt]vpB.north)
  node[midway, right, font=\footnotesize] {Demoted};
\pgfmathsetmacro{\yArrow}{-1*\yStep}

\draw[->, thick]
  ({\xA+\xLR+1.3}, \yArrow) -- ({\xB-\xLR-1.3}, \yArrow);

\node[font=\footnotesize, align=center]
  at ({(\xA+\xB)/2}, {\yArrow+0.32})
  {Demoting $v$};

\pgfmathsetmacro{\yL}{0.10}

\draw[thick] (\xL,\yL) -- ({\xL+0.6},\yL);
\node[legendtext, anchor=west, text width=\legTextW cm]
  at ({\xL+0.85},\yL) {Matched};

\draw[dash pattern=on 6pt off 3pt, thick] (\xL,{\yL-0.6}) -- ({\xL+0.6},{\yL-0.6});
\node[legendtext, anchor=west, text width=\legTextW cm]
  at ({\xL+0.85},{\yL-0.6}) {Unmatched};

\end{scope}
\pgfmathsetmacro{\yText}{-2*\yStep - 1.05}

\node[anchor=north west, text width=\textWidth cm, font=\footnotesize]
  at (\sysLeft,\yText)
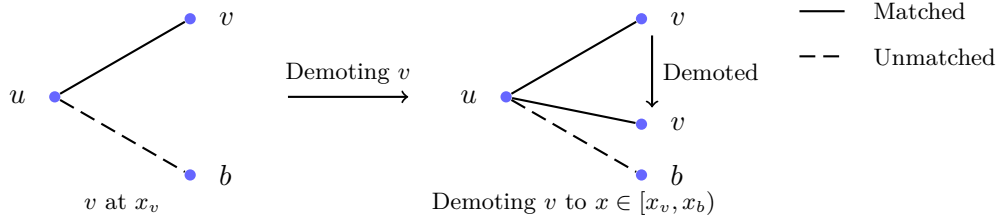
{An example of \cref{monoto_before_backup_ranking} where $u$ has a backup $b$. Permuting $v$ to any rank $x\in[x_v,x_b)$ while keeping all other ranks unchanged will not affect the final matching.};

\end{scope}
\end{tikzpicture}
}

Now we are ready to deduce the case-by-case lower bounds for all the profiles, starting with the profile $(x_u,\bot,\bot)$.
\subsection{Lower Bound for $(x_u,\bot,\bot)$}
In this case, since $u$ is not matched in $\M(\vecx_{\shortminus u^*})$, $u^\ast$ will always be matched to some vertex at least as good as $u$ by the following simple fact:
\begin{fact}
\phantomsection
\label{matching_guarantee_ranking_no_macth}
    If $u$ is not matched in $\vecx_{\shortminus u^\ast}$, then $u^\ast$ is matched to a vertex of rank $\leq x_u$ in $\vecx_{\shortminus u^\ast}(x)$.
\end{fact}
\begin{proof}
    This is a variant of property $1$ in \cref{fact:matching_guarantee_ranking}. The proof is essentially the same. Assume $u^\ast$ is not matched at time $t=(u,u^\ast)$. Then the alternating path $\M^t(\vecx)\oplus \M^t(\vecx_{\shortminus u^\ast})$ is degenerated and hence $A^t(\vecx)= A^t(\vecx_{\shortminus u^\ast})\cup\{u^\ast\}$. Since $u\in A(\vecx_{\shortminus u^\ast})\subseteq A^t(\vecx_{\shortminus u^\ast})$, it is also in $A^t(\vecx)$, so $(u,u^\ast)$ will be matched in this case.
\end{proof}
The gain is at least $0$ for $u$. The gain for $u^*$ is at least  $g_P(x_u,x)$ in $\vecx_{\shortminus u^\ast}(x)$ by the above fact and the monotonicity of gains (\cref{fact:monotonicity_of_gains}). We can define the $G$ function and bound the expected gain as follows:
\begin{claim}
\phantomsection
\label{G_ranking_nomatch}
    We define the lower-bound function $G$ when the profile of $u$ with respect to $\vecx_{\shortminus u^\ast}$ equals $(x_u,bot,\bot)$ by:
    $$G(x_u,\bot,\bot)= \int_{0}^1 g_P(x_u,x)\; dx.$$
    We have 
    $$G(x_u,\bot,\bot)\leq \ev_x\left[\text{gain}(u) + \text{gain}(u^\ast)\;|\; \text{profile of $u$}=(x_u,\bot,\bot)\right].$$
\end{claim}

\subsection{Lower Bound for $(x_u,x_v,\bot)$}
To better bound the expected gain for this case, we define an auxiliary rank $\theta_0\in [0,1]$, called the \emph{impacting rank}. The impacting rank $\theta_0\in[0,1]$ is a rank such that if $u^\ast$ is placed before $\theta_0$, $u$ can be made worse off; if $u^\ast$ is placed after $\theta_0$, $u$ will not be made worse off. Intuitively, before $\theta_0$, $u^*$ has a negative impact on the matching status of $u$; after $\theta_0$, it exerts no impact on the matching status of $u$. Formally, we define $\theta_0$ as follows:
\begin{definition}[Impacting Rank $\theta_0$]
\phantomsection
\label{def:theta_0_ranking}
    Assume $u$ is matched in $\M(\vecx_{\shortminus u^*})$. The impacting rank $\theta_0$ with respect to $\vecx_{\shortminus u^\ast}$ is
    $$\theta_0=\sup_{x\in[0,1]}\{x\;|\; u \text{ is made worse off in $\M(\vecx_{\shortminus u^\ast}(x))$}\}.$$
\end{definition}
We show the following simple but important facts about $\theta_0$, where $\theta_0^-$ denotes $\theta_0-\epsilon$ for some small enough $\epsilon$:
\begin{fact}[Properties of $\theta_0$, I]
\phantomsection
\label{fact:theta_0}
Fixing $\vecx_{\shortminus u^*}$, we have the following properties with respect to the impacting rank $\theta_0$:
\begin{enumerate}
    \item $\theta_0\leq x_u$, where $x_u$ is the rank of $u$.
    \item If $0<\theta_0<x_u$, for any $x\in[0,\theta_0]$ such that $u$ is made worse off by inserting $u^\ast$ at $x$, we have that $u^*$ is matched in $\M(\vecx_{\shortminus u^\ast}(x))$ and its match is not $v$ (the match of $u$ in $\vecx_{\shortminus u^\ast}$).
\end{enumerate}
\end{fact}
\begin{proof}
    Property 1: By \cref{fact:matching_guarantee_ranking}, $u$ is not made worse off in $\vecx_{\shortminus u^\ast}(x)$ for $x>x_u$. Hence, we have $\theta_0\leq x_u$.
    
Property 2: Assume $0<\theta_0<x_u$, fix $x\in[0,\theta_0]$ be as defined. 
Since $u$ is made worse off in $\M(\vecx_{\shortminus u^\ast}(x))$, $u^\ast$ must be matched. 
Assume that the match of $u^\ast$ at this time is $v$, the match of $u$ in $\vecx_{\shortminus u^\ast}$. 
This implies that the alternating path 
$\M(\vecx_{\shortminus u^\ast}(x))\oplus \M(\vecx_{\shortminus u^\ast})$ 
consists entirely of three vertices: $u^*, v, u$. 
 This implies that $u$ is the backup of $v$ with respect to $\vecx_{\shortminus u^\ast}(x)$. By \cref{monoto_before_backup_ranking},  it follows that demoting the rank of $u^*$ to any $x'\in [x,x_u)$ does not change the matching. This implies $\theta_0\ge x_u$ as there exists $x'$ arbitrarily close to $x_u$ where $u$ is unmatched in $\M(\vecx_{\shortminus u^\ast}(x'))$. Thus we have a contradiction as we assumed $\theta_0<x_u$.
\end{proof}

\Cref{fact:theta_0} implies that assuming $u$ is made worse off by $u^*$, except in one extreme case (the case where $\theta_0=x_u$), the alternating path $\M(\vecx_{\shortminus u^\ast}(\theta_0^-))\oplus\M(\vecx_{\shortminus u^\ast})$ will have length $\geq 4$. By \cref{lem:victim_querycommit_case_u} and \cref{lem:victim_querycommit_case_u_star}, the length of alternating path is positively correlated with how many copies of compensation $u$ and $u^\ast$ receive when they are unmatched. In the case where the alternating path has length $\geq 4$, $u$ and $u^\ast$ receive $\geq 2$ copies of compensation each. When the path has length $\geq 2$, in the worst-case $u$ and $u^\ast$ receive $1$ copy of compensation each when unmatched.

We also have the following property with respect to the alternating paths $\M(\vecx_{\shortminus u^*}(x))\oplus \M(\vecx_{\shortminus u^*}$ for any rank $x\in[0,\theta_0)$ that belongs to the set of ranks that define $\theta_0$.
\begin{fact}[Properties of $\theta_0$, II]
\phantomsection
\label{fact:theta_0_II}
Fixing $\vecx_{\shortminus u^*}$, let $\theta_0$ be the impacting rank as defined in \cref{def:theta_0_ranking} with respect to $\vecx_{\shortminus u^*}.$ Assuming $\theta_0 > 0$, which is the non-trivial case, let $S$ be the collection of ranks $x$ that defines $\theta_0$. I.e.
    $$S=\{x\;|\; \text{$u$ is made worse off in $\M(\vecx_{\shortminus u^*}(x))$}\}.$$
    Then,
\begin{enumerate}
    \item For any $x\in S$, the third vertex $u_2$ in the alternating path $\M(\vecx_{\shortminus u^*}(x))\oplus \M(\vecx_{\shortminus u^*})$ has rank $x_{u_2}\leq \theta_0$.
    \item There exists a small enough punctured left neighborhood of $\theta_0$, denoted as $(\theta_0^-,\theta_0)$, such that $(\theta_0^-,\theta_0)\subseteq S$. Further, for any $x\in(\theta_0^-,\theta_0)$, the third vertex $u_2$ in the alternating path $\M(\vecx_{\shortminus u^*}(x))\oplus \M(\vecx_{\shortminus u^*})$ has rank $x_{u_2}= \theta_0$.
\end{enumerate}
\end{fact}
\begin{proof}
    Property 1: Fix any $x\in S$. Let the alternating path be $\M(\vecx_{\shortminus u^*}(x))\oplus \M(\vecx_{\shortminus u^*})=(u^*,w,u_2,...)$. By \cref{lem:backup_in_alt_path}, $u_2$ is the backup of $w$ with respect to $\vecx_{\shortminus u^*}(x)$. By the monotonicity property of the match before backup (\cref{monoto_before_backup_ranking}), demoting the rank of $u^*$ until $x_{u_2}$ does not change the matching. This implies that there exists $y$ arbitrarily close to $x_{u_2}$ from the left such that $u$ is made worse off in $\M(\vecx_{\shortminus u^*}(y))$. These $y$ thus belong to the set of points that define $\theta_0$, which means $x_{u_2}\leq \theta_0$.

    Property 2: We show that there exists a rank $x\in S$ such that $u$ is made worse off in $\M(\vecx_{\shortminus u^*}(x))$ and the third vertex in the alternating path has rank exactly $\theta_0$. Then we can define the left neighborhood as $(x,\theta_0)$. Since by monotonicity of match before backup (\cref{monoto_before_backup_ranking}), we have $\M(\vecx_{\shortminus u^*}(y))=\M(\vecx_{\shortminus u^*}(x))$ for any $y\in  (x,\theta_0)$.

    Assume towards a contradiction that for any $x\in S$, the third vertex has rank $x_{u_2}<\theta_0$. Again, since $u_2$ is the backup of $w$ with respect to $\M(\vecx_{\shortminus u^*}(x))$ with $w$ being the match of $u^*$, by \cref{facf:backup_ranking}, $x<x_{u_2}$. This implies
    \begin{align*}
        \theta_0=&\sup_{x\in S}\{x\}\\
        \leq &\sup_{x\in S}\{x_{u_2}\;|\; \text{ $x_{u_2}$ is the rank of the third vertex in $\M(\vecx_{\shortminus u^*}(x))\oplus \M(\vecx_{\shortminus u^*})$}\}<\theta_0,
    \end{align*}
    which is a contradiction. The last $<$ inequality is because each $x_{u_2}$ is the rank of some vertex in the graph, hence the collection of $x_{u_2}$ is a finite set. Since each $x_{u_2}<\theta_0$, their supremum cannot be $\theta_0$.
\end{proof}
These two facts help us better bound the value of the compensations received when $u$ and $u^*$ are unmatched. In the bounding function, we often need to bound $h(y,z)$ where $y$ or $z$ is the rank of the third vertex $u_2$ in the alternating path.
\paragraph{Preliminary Bounds for $(x_u,x_v,\bot)$}
We first provide preliminary bounds for the expected gain of $u$ and the expected gain of $u^\ast$ individually. In this system of bounds, we effectively assume that alternating paths have length $2$ or $4$, thus collecting $1$ or $2$ copies of compensation when a vertex is unmatched. We will later provide a refined analysis where alternating paths with length $6$ are also taken into account. Theoretically, we could further consider cases where alternating paths have length $>6$, but as we take longer alternating paths into consideration, the system of inequalities becomes more intricate and is already very complex with length $6$. So we will stop at length-$6$ alternating paths.

We start the preliminary case by lower bounding the expected gain of $u$. By the definition of $\theta_0$, we know that when $u^*$ is inserted at $x>\theta_0$, $u$ is matched to a vertex at least as good as $v$. If $x<\theta_0$, when $u$ does not have a backup, in the worst-case $u$ will be unmatched due to the introduction of $u^\ast$. We show that: 
\begin{claim}
\phantomsection
\label{fact:lowerbound_u_nobackup}
    Conditioning on the event that the profile of $u$ equals $(x_u,x_v,\bot)$:
$$
\ev_x[\text{gain}(u)] \ge
\left\{
\begin{array}{lr}
  \int_0^{\theta_0}h(x,x_{\text{match of $u^*$}})\; dx\;\;+\; \\
  (1-\theta_0)\cdot g_B(x_u,x_v),
  & \text{if } x_u > x_v \text{ and } \theta_0 = x_u,\\[8pt]
  \int_0^{\theta_0}h(x,x_{\text{match of $u^*$}})\; dx\;+\;\theta_0\cdot  h(x_u, x_v)\; \\
  + (1-\theta_0)\cdot g_B(x_u,x_v), 
  & \text{otherwise.}
\end{array}
\right.
$$
\end{claim}
\begin{proof}
    By the definition of $\theta_0$, when $x\in(\theta_0,1]$, $u$ is not made worse off. Hence, it is matched to a vertex of rank at most $x_v$. So $u$ gets at least $g_B(x_u,x_v)$ amount of gain by the monotonicity property of gains (\cref{fact:monotonicity_of_gains}).

    When $x\in [0,\theta_0)$, $u$ can be unmatched by $u^\ast$. By \cref{fact:assume_unmatched_worst_case}, we derive the lower bound assuming $u$ is unmatched.  When $u$ is unmatched, the alternating path $\M(\vecx)\oplus \M(\vecx_{\shortminus u^\ast})$ is some path $u_0,\ldots,u_k$ such that $u_0=u^\ast$ and $u_k=u$. We look at the last three nodes, $u_{k-2},\; u_{k-1}=v,$ and $ u_k=u$. \cref{lem:victim_querycommit_case_u} shows that $u$ is the victim of $u_{k-2}$ and $u^\ast$.

As the match of $u_{k-2}$ in $\M(\vecx)$ is $v$, $u$ receives $h(x_{u_{k-2}}, x_v)$ amount of compensation from $u_{k-2}$. Notice that $(u_{k-2},v)$ is queried before $(u,v)$, which happens iff $x_{u_{k-2}}<x_u$. So the compensation $u$ received from $u_{k-2}$ is at least $h(x_u,x_v)$ by the monotonicity constraints of $h$ ($h$ monotonically decreases in the first coordinate). $u$ also receives at least $h(x, x_{\text{match of }u^*})$ amount of compensation from $u^\ast$.

To exclude the double counting case where $u_{k-2}=u^*$, we show that $u_{k-2}\neq u^\ast$ unless $x_u>x_v$ and $\theta_0=x_u$. Assume W.L.O.G. that $x_u\neq x_v$. If $x_u<x_v$, then $x\leq\theta_0\leq x_u< x_v$ implies $u^\ast$ is matched to some vertex with rank $\leq x_u<x_v$ in $\M(\vecx)$ by \cref{fact:matching_guarantee_ranking}, property $1$. Hence, $u_{k-2}\neq u^\ast$, as $u_{k-2}$ is matched to $v$ (with rank $x_v$). If $x_u>x_v$ and $\theta_0<x_u$, by \cref{fact:theta_0}, property $2$, $u^\ast$ is not matched to $v$ and hence $u^\ast\neq u_{k-2}$. 
\end{proof}
We now approximate the minimum expected gain of $u^\ast$. We split it into four cases by the relative order of $\theta_0,x_u,$ and $x_v$. We show that:
\begin{claim}
\phantomsection
\label{fact:lowerbound_u_star_nobackup}
        Conditioning on the event that profile of $u= (x_u,x_v,\bot)$:
        $$\ev_x[\text{gain}(u^\ast)] \ge \begin{cases}
  \int_0^{x_v} g_P(x_u,x) \;dx \;+\; (1-x_v)\cdot(h(x_u,\theta_0)+h(x_v,x_u)), & \text{if } \theta_0\leq x_u<x_v,\\[8pt]
    \int_0^{\theta_0} g_P(x_v,x) \;dx \;+\;\int_{\theta_0}^{x_v} g_P(x_u,x) \;dx \; \\
    \quad+\;(1-x_v)\cdot(h(x_v,\theta_0)+h(x_v,x_u)), & \text{if } \theta_0< x_v<x_u,\\[8pt]
  \int_0^{\theta_0} g_P(x_v,x) \;dx \;+\; (1-\theta_0)\cdot(h(x_v,\theta_0)+h(x_v,x_u)), & \text{if }  x_v\leq \theta_0<x_u,\\[8pt]
    \int_0^{\theta_0} g_P(x_v,x) \;dx \;+\; (1-\theta_0)\cdot h(x_v, x_u), & \text{if } x_v<\theta_0=x_u.\\
\end{cases}$$
\end{claim}

\begin{proof}
    Recall that by \cref{fact:theta_0}, we have $\theta_0\leq x_u$. Further, W.L.O.G., we can assume that $x_u\neq x_v$. So the above four cases cover all possible orderings of $\theta_0,x_u,x_v.$
    
    \noindent\textbf{Case $\theta_0\leq x_u<x_v$:} By \cref{fact:matching_guarantee_ranking}, $u^\ast$ is matched to a vertex of rank $\leq x_u$ when $x\in[0,x_v)$, so we have the integral part. For $x\in(x_v,1]$, we again assume $u^\ast$ is unmatched to calculate the worst-case gain. Let $y\in (\theta_0^-,\theta_0)$ be some rank in a small enough left neighborhood of $\theta_0$. In this case, by \cref{fact:theta_0_II}, $u$ is made worse off in $\M(\vecx_{\shortminus u^*}(y))$.\footnote{When $\theta_0=0$, $\theta_0^-$ is not well defined. But since we set $h(x,0)=0$ for all $x$, we only have to show $u^\ast$ is the victim of $v$ to get $h(x_u,x_v)$ amount of compensation. By \cref{claim:victim_when_no_alt_path}, this is true for all the four cases.}. Let $u_0,...u_k$ be the alternating path $\M(\vecx_{\shortminus u^*}(y))\oplus \M(\vecx_{\shortminus u^*})$ where $u_0=u^\ast,u_1=w, u_{k-1}=v$ and $u_k=u$. By \cref{lem:victim_querycommit_case_u_star}, when $u^\ast$ is demoted to $x\in(x_v,1]$, if it is unmatched, then it is the victim of $w$ and $v$. $v$ is matched to $u$ in this case, hence it gives $h(x_v, x_u)$ amount of compensation to $u^*$. 
    
    For $w$ to pay $h(x_u,\theta_0)$ amount of compensation, we show that its rank is $\leq x_u$ and its match $u_2$ in $\M(\vecx_{\shortminus u^*}(x))$ has rank $\geq \theta_0$. By \cref{fact:theta_0_II}, property 2, and the choice of $y$, $u_2$ has rank exactly $\theta_0$. Since $y<\theta_0\leq x_u<x_v$ with $w$ being the match of $u^*$ in $\M(\vecx_{\shortminus u^*}(y))$, by \cref{fact:matching_guarantee_ranking}, property 1, $w$ has rank $\leq x_u$.
    
    By the same line of reasoning as in the proof for \cref{fact:lowerbound_u_nobackup}, $w\neq v$ when $\theta_0\leq x_u<x_v$, so we are not double counting when we add two copies of compensation.

    \noindent\textbf{Case $\theta_0< x_v<x_u$:} Again, fix $y\in (\theta_0^-,\theta_0)$ as described in \cref{fact:theta_0_II}, property 2. Since $u$ is made worse off in $\M(\vecx_{\shortminus u^*}(y))$, by \cref{fact:matching_guarantee_ranking}, property $3$, $u^\ast$ is matched to a vertex of rank $\leq x_v$. Hence, by property $4$ of \cref{fact:matching_guarantee_ranking}, $u^*$ is matched to a vertex of rank $\leq x_v$ in the whole segment $x\in[0,\theta_0)$. 
    
    When $x\in(\theta_0,x_v)$, it is matched to a vertex of rank $\leq x_u$ by \cref{fact:matching_guarantee_ranking}, property $1$. For $x\in(x_v,1]$, $u^\ast$ can be unmatched, by the same line of reasoning as in the previous case, $u^\ast$ receives compensation from $w,v$ where $\M(\vecx_{\shortminus u^*}(y))\oplus \M(\vecx_{\shortminus u^*})=(u^*,w,u_2,...,v,u)$. We now upper bound the rank of $w$ by property $2$ of \cref{fact:matching_guarantee_ranking}. Since $w$ is the match of $u^*$ in $\M(\vecx_{\shortminus u^*}(y))$ in which $u$ is made worse off, $w$ will have rank $\leq x_v$. By our choice of $y$, similarly as before, we have $x_{u_2}=\theta_0$. Hence $w$ will send at least $h(x_v,\theta_0)$ amount of compensation to $u^*$. We are again not double counting by \cref{fact:theta_0}, as $w\neq v$ when $\theta_0<x_u$.

    \noindent\textbf{Case $x_v\leq \theta_0<x_u$:} This is almost the same as the previous case. Now since $\theta_0\geq x_v$, the interval $(\theta_0,x_v)$ does not exist. $u^\ast$ still receives two copies of compensation as $\theta_0$ is still $<x_u$.
    
    \noindent\textbf{Case $x_v<\theta_0=x_u$:} Most of the bounds follow the same line of reasoning as the previous case. Except that now we can no longer assume $w\neq v$, so we only collect one copy of compensation from $v$ for $u^*$ when it is unmatched in $(\theta_0,1]$.
\end{proof}
We now define a preliminary lower bound function $G$ when profile of $u=(x_u,x_v,\bot)$:
\begin{claim}
\phantomsection
\label{G_ranking_nobackup}
    For profile of $u=(x_u,x_v,\bot)$, we define $G(x_u,x_v,\bot)$ as
\begin{align*}
&\min_{\theta_0\leq x_u}\left\{
    \qquad\qquad\int_0^{x_v} g_P(x_u,x)\,dx
    +(1-x_v)\cdot\bigl(h(x_u,\theta_0)+h(x_v,x_u)\bigr)\right. \\
&\qquad\qquad\qquad\quad\quad
    +\int_0^{\theta_0} h(x,x_u)\,dx
    +\theta_0\cdot\,h(x_u,x_v)
    +(1-\theta_0)\cdot\,g_B(x_u,x_v)
&(1)\Biggr\}
&\text{, (if } x_u<x_v), \\[1em]
&\min_{\theta_0\leq x_u}\Biggl\{
    \;\;\inf_{\theta_0<x_u}\;\;\Biggl[\;
        \int_0^{\theta_0} g_P(x_v,x)\,dx
        +\int_{\theta_0}^{\max\{x_v,\theta_0\}} g_P(x_u,x)\,dx \\
&\qquad\qquad\qquad\qquad
        +(1-\max\{x_v,\theta_0\})\cdot\bigl(h(x_v,\theta_0)+h(x_v,x_u)\bigr) \\
&\qquad\qquad\qquad\qquad
        +\int_0^{\theta_0} h(x,x_v)\,dx
        +\theta_0\cdot\,h(x_u,x_v)
        +(1-\theta_0)\cdot\,g_B(x_u,x_v)
    \Biggr], &(2)\phantom{\Biggr\}}& \\[0.75em]
&\qquad\qquad
    \min_{\theta_0=x_u}\;\;\Biggl[\;
        \int_0^{\theta_0} g_P(x_v,x)\,dx
        +(1-\theta_0)\cdot\,h(x_v,x_u) \\
&\qquad\qquad\qquad\qquad
        +\int_0^{\theta_0} h(x,x_v)\,dx
        +(1-\theta_0)\cdot\,g_B(x_u,x_v)
    \;\Biggr]
&(3)\Biggr\}
&\text{, (if } x_v<x_u).
\end{align*}
\end{claim}
\begin{proof}
    This is a combination of \cref{fact:lowerbound_u_nobackup} and \cref{fact:lowerbound_u_star_nobackup}, taking the worst-case scenario for $\theta_0\leq x_u$. We combined the case $\theta_0<x_v<x_u$ and $x_v\leq\theta_0<x_u$ into one case $\theta_0<x_u$. The expected gain of $u^*$ (\cref{fact:lowerbound_u_star_nobackup}), for the above two cases, can be expressed as the unified expression
    $$\int_0^{\theta_0}g_P(x_v,x)dx +\int_{\theta_0}^{\max\{\theta_0,x_v\}}g_P(x_u,x)dx+(1-\max\{\theta_0,x_v\})\cdot(h(x_v,\theta_0)+h(x_v,x_u)).$$
     The expected gain of $u$ for the above two cases, on the other hand, has the same expression.
    
    We also replaced $h(x,x_{\text{match of $u^*$}})$ in \cref{fact:lowerbound_u_nobackup} with $h(x,x_u)$ and $h(x,x_v)$ as we assume that $u^*$ is matched to a vertex of rank $\leq x_u$ and $\leq x_v$, respectively in each of the cases $x_u<x_v$ and $x_v<x_u$. In the prior case, $u^*$ in $x\in[0,\theta_0)$ is matched to a vertex $w$ of rank $x_w \in [0,x_u)$. The gain of $u^*$ plus compensation received by $u$ from $u^*$ (assuming $u$ unmatched), is at least
    $$1-g(x_w,x)-h(x,x_w)+h(x,x_w)=1-g(x_w,x)$$
    which is a function that monotonically decreases as $x_w$ increases, so we can effectively lower bound the three parts together, assuming the worst-case match of $u^*$ in this range is at rank $x_u$. Similar reasoning applies to the latter case as well.
\end{proof}
\subsection{Lower Bound for $(x_u,x_v,x_b)$}
We define $\theta_0$ the same way as in \cref{def:theta_0_ranking} and bound the expected gains of $u$ and $u^\ast$ individually, starting from $u$. 
\begin{claim}
\phantomsection
\label{lowerbound_u_backup_ranking}
    Conditioning on the event profile of $u= (x_u,x_v,x_b)$:
$$\ev_x[\text{gain}(u)] \ge
  \theta_0\cdot g_B(x_u,x_b) \;+\; (1-\theta_0)\cdot g_B(x_u,x_v).$$   
\end{claim}
\begin{proof}
   By \cref{facf:backup_ranking}, $u$ is matched to a vertex with rank $\leq x_b$ when $x\in[0,\theta_0)$; and by \cref{fact:matching_guarantee_ranking}, $u$ is matched to a vertex with rank $\leq x_v$ when $x\in(\theta_0,1]$. The bound then follows from the monotonicity property of gains \cref{fact:monotonicity_of_gains}.
\end{proof}
For $u^\ast$, we split into four subcases by the relative order of $\theta_0,x_u$, and $x_v$.
\begin{claim}
\phantomsection
\label{lowerbound_u_star_backup_ranking}
        Conditioning on the event that the profile of $u= (x_u,x_v,x_b)$:
    $$\ev_x[\text{gain}(u^\ast)] \ge \begin{cases}
  \int_0^{x_v} g_P(x_u,x) \;dx \;+\;\max\{x_b-x_v,0\}\cdot(h(x_u,\theta_0)+h(x_v,x_u)), & \text{if } \theta_0\leq x_u<x_v,\\[8pt]
    \int_0^{\theta_0} g_P(x_v,x) \;dx \;+\;\int_{\theta_0}^{x_v} g_P(x_u,x) \;dx \\
    \quad+\; \max\{x_b-x_v,0\}\cdot(h(x_v,\theta_0)+h(x_v,x_u)), & \text{if } \theta_0< x_v<x_u,\\[8pt]
  \int_0^{\theta_0} g_P(x_v,x) \;dx \;+\; \max\{x_b-\theta_0,0\}\cdot(h(x_v,\theta_0)+h(x_v,x_u)), & \text{if }  x_v\leq \theta_0<x_u,\\[8pt]
    \int_0^{\theta_0} g_P(x_v,x) \;dx \;+\; \max\{x_b-\theta_0,0\}\cdot h(x_v,x_u), & \text{if } x_v<\theta_0=x_u.\\
\end{cases}$$
\end{claim}
\begin{proof}
 Compare \cref{lowerbound_u_star_backup_ranking} with \cref{fact:lowerbound_u_star_nobackup}, and notice that the only difference lies in terms related to the expected gain of $u^\ast$ in the range $x\in(\max\{x_v,\theta_0\},1]$, where $u^\ast$ is assumed to be unmatched. In \cref{fact:lowerbound_u_star_nobackup}, $u^*$ will receive compensation throughout this interval by \cref{lem:victim_querycommit_case_u_star}. The same lemma, applied to the case where $u$ has a backup $b$, can only guarantee that $u^\ast$ is a victim if $u^\ast$ is a better choice for $u$ compared to the backup $b$. This will happen if and only if $x<x_b$. Hence we can only assume $u^\ast$ receives compensation when $x\in (\max\{\theta_0,x_v\},x_b)$. This accounts for the transition from $(1-x_v)$ to $\max\{x_b-x_v,0\}$, and from $(1-\theta_0)$ to $\max\{x_b-\theta_0,0\}$. All the other terms are identical and follow from the same proof as \cref{fact:lowerbound_u_star_nobackup}.
\end{proof}
Similarly, we define $G$ for the profiles of $u$ when $u$ has a backup $b$.
\begin{claim}
\phantomsection
\label{G_ranking_backup}
For the profile of $u=(x_u,x_v,x_b)$ where $x_v,x_b\neq \bot$, we define $G(x_u,x_v,x_b)$ as
\begin{align*}
&\min_{\theta_0\leq x_u}\left\{
    \qquad\qquad\quad\int_0^{x_v} g_P(x_u,x) \;dx
    \;+\; \max\{x_b-x_v,0\}\cdot(h(x_u,\theta_0)+h(x_v,x_u))
\right. \\
&\qquad\qquad\qquad\qquad\;\;\;
    +\theta_0\cdot g_B(x_u,x_b)
    \;+\; (1-\theta_0)\cdot g_B(x_u,x_v)
    \phantom{\int}
&(1)\Biggr\}
&\ \text{if } x_u<x_v, \\[1em]
&\min_{\theta_0\leq x_u}\Biggl\{
    \quad\;\;\inf_{\theta_0<x_u}\;\;\Bigl[\;
        \int_0^{\theta_0} g_P(x_v,x) \;dx
        \;+\;\int_{\theta_0}^{\max\{x_v,\theta_0\}} g_P(x_u,x) \;dx
\Bigr. \\
&\qquad\qquad\qquad\qquad\;\;\;
        +\max\{x_b-\max\{x_v,\theta_0\},0\}\cdot(h(x_v,\theta_0)+h(x_v,x_u))
\\
&\qquad\qquad\qquad\qquad\;\;\;
        +\theta_0\cdot g_B(x_u,x_b)
        \;+\; (1-\theta_0)\cdot g_B(x_u,x_v)
    \Bigl.\;\Bigr],
&(2)\phantom{\Biggr\}}& \\[0.8em]
&\qquad\qquad\quad
    \min_{\theta_0=x_u}\;\;\Bigl[\;
        \int_0^{\theta_0} g_P(x_v,x) \;dx
        \;+\max\{x_b-\theta_0,0\}\cdot h(x_v,x_u)
\Bigr. \\
&\qquad\qquad\qquad\qquad\;\;\;\Bigl.
        +\theta_0\cdot g_B(x_u,x_b)
        \;+\; (1-\theta_0)\cdot g_B(x_u,x_v)
    \;\Bigr]
&(3)\Biggr\}
& \ \text{if } x_v<x_u.
\end{align*}
     We have 
    $$G(x_u,x_v,x_b)\leq \ev_{x}\left[\text{ gain}(u) + \text{gain}(u^\ast)\;|\; \text{profile of $u$}=(x_u,x_v,x_b)\right].$$
\end{claim}
\begin{proof}
        This is a combination of \cref{lowerbound_u_backup_ranking} and \cref{lowerbound_u_star_backup_ranking}, taking the worst-case scenario for $\theta_0$. Similar to the function $G(x_u,x_v,\bot)$, we unified the two cases $\theta_0<x_v<x_u$ and $x_v\leq\theta_0<x_u$.
 \end{proof}
\subsection{Expectation Over Worst Case Distribution of Profiles}
Because of these monotonicity properties, we can effectively lower-bound the expected value of gain($u$) $+$ gain$(u^\ast)$ for each fixed rank of $u$ by the worst-case $x_v$-uniformly distributed profiles. 
\begin{claim}
\phantomsection
\label{claim:ranking_lower_bound_uniform_profiles}
    Let $G$ be the lower bound function we defined in \cref{G_ranking_nomatch}, \cref{G_ranking_nobackup}, and \cref{G_ranking_backup}. For each fixed $x_u$, we have: 
\begin{align*}
  \min\;\left\{
    \begin{array}{c}
      G(x_u,\bot,\bot), \\[4pt]
      \displaystyle \inf_{0\le v_0<1}\biggl[\frac{1}{1-v_0}\int_{v_0}^{1} G(x_u,x_v,\bot)\,dx_v\biggr], \\[4pt]
      \displaystyle \inf_{0\le v_0<b_0\leq 1}\biggl[\frac{1}{b_0 - v_0}\int_{v_0}^{b_0} G(x_u,x_v,b_0)\,dx_v\biggr]
    \end{array}
  \right\}\leq   \ev\bigl[\text{ gain}(u) + \text{gain}(u^\ast) \mid \text{rank of $u = x_u$}\bigr].
\end{align*}
We denote the LHS by $G_u(x_u)$.
\end{claim}
For profiles $(x_u,x_v,b_0)$, intuitively, this is because integrating over a monotonically increasing probability density function $p$ of $x_v\in[0,b_0)$, i.e.,
$$\int_0^{b_0} p(x_v)\cdot G(x_u,x_v,b_0) dx_v,$$
is equivalent to horizontally integrating over slices of uniformly distributed intervals $x_v\in(v_0,b_0]$ for some other probability distribution $p'$ for $v_0\in[0,b_0)$:
$$\int_0^{b_0} p'(v_0)\cdot \left[\frac{1}{b_0-v_0}\int_{v_0}^{b_0} G(x_u,x_v,b_0) dx_v\right] dv_0.$$
Thus, it can be lower-bounded by the worst-case uniform distribution over all possible $[v_0,b_0]$. The same reasoning also applies to profiles $(x_u,x_v,\bot)$. For those interested, please refer to \nameref{uniform-profiles} for a complete proof.

Now, to calculate the final approximation ratio, we simply take the expected value of $G_u$.
\begin{claim}
Let $G_u$ be the function defined in \cref{claim:ranking_lower_bound_uniform_profiles}. Then, the approximation ratio of \ranking{} is lower-bounded by the expected value of $G_u$, i.e.,
    $$\int_0^1 G_u(x)dx\leq \ev[\text{gain$(u)$ $+$ gain$(u^\ast)$}].$$
\end{claim}
\section{\Franking{}}

\Franking{} was introduced by \cite{0.521Franking} for the fully online matching problem\footnote{See a detailed discussion of the fully online matching problem and its relation to \Franking{} in \cref{Fully-Online_Section}.}. Unlike in the uniform random arrival matching settings, where the decision time $\pi$ is a uniform random permutation, in the fully online matching setting, the decision time is fixed by an oblivious adversary. \Franking{} is the algorithm that picks a random common rank (preference) list for all vertices, and lets each vertex pick the smallest-ranked available neighbor at its decision time. The algorithm can be equivalently viewed as performing query-commit matching algorithm with preference list $L=\pi\times \sigma$ where $\pi$ is the adversarial decision order and $\sigma$ is a uniform random permutation of vertices. We now state the \Franking{} algorithm.

\begin{algorithm}[H]
  \caption{\Franking{} algorithm for general graphs.}
  \phantomsection
  \label{alg:Franking}
  \KwIn{Graph $G=(V,E)$, adversarial decision order $\pi$}
  \SetAlgoLined

 Sample an independent uniform random value $x_v \in [0,1]$ for each $v \in V$\;

  \For{each vertex $v \in V$ in increasing order of $\pi(v)$}{
    \If{$v$ is available and has an available neighbor}{
      Match $v$ with the smallest-ranked available neighbor $u$ according to $\vecx$\;
    }
  }
\end{algorithm}

Fix $\pi$ as an arbitrary adversarial decision order, $\vecx$ a random rank vector, and $(u,u^*)$ a pair of nodes in the perfect matching $M^\ast$, we use the same notations $\M(\vecx)$, $\vecx_{\shortminus u^\ast},$ etc., as before, omitting $\pi$ (as $\pi$ is fixed throughout the analysis). 
\textbf{Additionally, throughout the \Franking{} analysis, we assume W.L.O.G. that $u$ has an earlier decision time than $u^*$.}
\subsection{Our Approach vs. Previous Approach}
Previous \Franking{} analysis~\cite{0.521Franking} is also conducted through randomized primal-dual with gain-sharing. Their analysis introduced the notion of victim but did not include the notions of profile and backup. The analysis considered, for each fixed $\vecx_{\shortminus u^*}$, a system of worst-case matches of $u$ and $u^*$ over the randomization of $x$, the rank of $u^*$. By introducing the notion of backups and profiles, we were able to classify $\vecx_{\shortminus u^*}$ into $6$ different profiles taking the activeness/passiveness of the vertices $u,v,b$ into account\footnote{Recall the definition of active and passive vertices (\cref{def:active_passive_vertices}).}. Each profile gives a different worst-case matching condition for $u$ and $u^*$. Further, we were able to add constraints on the distribution of profiles themselves, allowing a more detailed universal analysis over the randomization of $u,\;u^*$, and a third vertex $v$ (the match of $u$) together.

\paragraph{Backup Characterizes Unique Worst Case Match for $u$.}
By identifying the profile of $u$ for each $\vecx_{\shortminus u^*}$, we have a complete characterization of how $u$ is made worse off as a result of introducing $u^*$ (\cref{fact:backup_uniquely_describe_worse_off}). In particular, we know that if $u$ becomes worse off, then $u$ is exactly matched to its backup $b$ (or unmatched if $u$ has no backup). In the previous analysis~\cite{0.521Franking}, they effectively considered up to $\frac{|V|}{2}$ possible rank pairs $(\theta_i,\tau_i)$, where $\theta_i$ is a decreasing sequence of rank values for $u^\ast$, and $\tau_i$ is an increasing sequence of rank values for the match of $u$. When $u^*$ has rank $[\theta_{i+1},\theta_i)$, the worst-case match of $u$ is of rank $\tau_i$. To avoid considering up to $\Omega(|V|)$ rank pairs, the previous work added an additional constraint requiring the compensation function $h(x)$ to have at most a linear growth rate, i.e., $\frac{h(x)}{x}$ is non-increasing (\cite{0.521Franking}, Definition 4.1). By this constraint, previous work was able to bound an arbitrary number of $(\theta_i,\tau_i)$ pairs by at most one pair of $(\theta_i,\tau_i)$. We show that, by the nature of the matching structures, there is at most one such pair of $(\theta_i,\tau_i)$, which is the pair $(\theta_0,x_b)$\footnote{Relevant notions are defined in~\cref{def:theta_0_franking} and \cref{def:backup_querycommit}.} in our setting. Thus, we can remove the additional constraint without worrying about a growing number of rank pairs for free.

\paragraph{Distribution of Profiles Enforces Distribution of Marginal Ranks.}
Previous analysis also introduced the notion of marginal rank (\cref{theta_1_franking}) as the unique rank $\theta_1$ with respect to a fixed $\vecx_{\shortminus u^\ast}$, where $u^*$ transitions from being passively matched to actively finding a match. Then the analysis is conducted by considering the worst-case $\theta_1$ for each $\vecx_{\shortminus u^*}$ and without imposing restrictions on the distribution of $\theta_1$. We identified that for those profiles where $u$ is actively matched to $v$, the marginal rank $\theta_1$ has to satisfy the condition $\theta_1\geq x_v$. Hence, the monotonicity property of $x_v$ (\cref{lemma:monotonicity_Franking}) directly imposes a monotonicity constraint for $\theta_1$, allowing us to restrict the worst case choices of marginal ranks for some $\vecx_{\shortminus u^\ast}$.

\paragraph{Different Compensation Rules.}
Previous analysis requires a blocker vertex $w$ to transfer part of its gain to its victim $v$ only when all three of the following conditions are met: i) $v$ is the victim of $w$, ii) $w$ is active, and iii) $w$ is a neighbor of $v$. We modified the compensation rule by removing the third condition. In terms of the final LP constraints, it turns out that we obtained extra gains roughly proportional to the following parameter:
$$\text{extra gain } \propto\; \theta_0-(1-\theta_1)=\theta_0+\theta_1-1,$$
where $\theta_1$ is the marginal rank (\cref{theta_1_franking}) and $\theta_0$ is the impacting rank (\cref{def:theta_0_franking}). As discussed in the previous paragraph, we enforced a constraint that $\theta_1$ has a larger distribution near $1$ than near $0$. Furthermore, $\theta_0$ is a parameter that is favored near a larger value\footnote{We prove that $\theta_0$ is a parameter in $[0,\theta_1)$, and assuming a fixed $\theta_1$, increasing $\theta_0$ will only cause a decrease in the number of matched vertices in expectation. Roughly speaking, when the LP tries to minimize the approximation ratio, $\theta_0$ is favored to be a number close to $\theta_1$, which is relatively large.}. As a result, the overall impact is positive for this modification.

\subsection{\Franking{} Structural Properties}
A key difference between \ranking{} and \Franking{} is the relative lack of monotonicity between the rank of a vertex and the rank of its match. Denote $v$ as the match of $u$. In \ranking{}, \cref{fact:matching_guarantee_ranking} shows that increasing $u$'s rank will lead to an increase in $v$'s rank. This fact, however, no longer holds in \Franking{}. It is possible that some vertex $v_1$ with a large rank picks $u$ with a small rank because $v_1$ has an extremely early decision time; and when we demote $u$'s rank, $v_1$ no longer prefers $u$, which makes $u$ match to a vertex $v_2$ with rank $x_{v_2}<x_{v_1}$. In general, we have limited knowledge about the rank of $v$ when $u$ is passively matched to $v$. For example, for the vertex $v$ with the earliest decision time ($\pi(v)=1$), we have no information about $x_v$ since $v$ is always the first vertex to pick. 

 In our analysis, we will show that most of the properties in \ranking{} still hold when $u$ is active. We will employ constraints that force being passively matched to be a better scenario than being actively matched to compensate for the lack of information when a vertex is matched passively.

We now prove the matching guarantees for \Franking{} that will be used in our approximating LP. These results can be found in \cite{0.521Franking}.
\begin{fact}[Matching Guarantees for \Franking{}]
\phantomsection
\label{fact:matching_guarantee_franking}
    Let $x$ be the rank of $u^*$ when $u^*$ is inserted into $\vecx_{\shortminus u^*}$. Assuming $u$ is matched to $v$ in $\M(\vecx_{\shortminus u^*})$, then:
    \begin{enumerate}
        \item If $x<x_v$ and $u$ actively matches $v$ in $\M(\vecx_{\shortminus u^*})$ , then $u^*$ is passively matched in $\M(\vecx_{\shortminus u^*}(x))$.
        \item If $u^*$ is unmatched or matches actively to a vertex in $\M(\vecx_{\shortminus u^*}(x))$, then $u$ is matched to $v$ in $\M(\vecx_{\shortminus u^*}(x))$. i.e., $u$'s match is not affected by $u^\ast$.
        \item If $u^*$ is passively matched in $\M(\vecx_{\shortminus u^*}(x_0))$, then it is also passively matched in $\M(\vecx_{\shortminus u^*}(x))$ for all $x\in[0,x_0]$.
    \end{enumerate}
\end{fact}
Property $3$ can be found in \cite{0.521Franking}, Lemma 2.2. With some effort, one can also translate properties $1$ and $2$ to Lemma 2.4 in~\cite{0.521Franking}. These are essentially the \Franking{} version of \cref{fact:matching_guarantee_ranking}, properties $1$, $2$, and $4$. For anyone interested in a proof, see~\nameref{sec:appendix-a}.

\subsection{Gain-Sharing Scheme}
\phantomsection
\label{sec:gain-sharing_franking}
We adopt the notion of victim as defined in \cref{def:victim_querycommit}. In \Franking{}, we do not partition vertices into buyer and product groups; instead, for each matched pair, one vertex is matched passively and the other is matched actively. Accordingly, our gain-sharing scheme distinguishes between passive and active matches and is defined as follows:
\begin{definition}[Gain Sharing for \Franking{}]
    For any gain function $g:[0,1]\to[0,1]$, compensation function $h:[0,1]\to[0,1]$, rank vector $\vecx$, we perform the following gain-sharing scheme:
    \begin{enumerate}
        \item When $u$ actively matches $v$, $u$ receives $1-g(x_v)$ amount of gain and $v$ receives $g(x_v)$ amount of gain.
        \item If $w$ actively matches $z$ and further has a victim $v$, then $w$ sends $h(x_z)$ amount of compensation to $v$. In this case, the gain of $w$ is $1-g(x_z)-h(x_z)$ and $v$ gains an additional $h(x_z)$ amount of gain. The gain of an unmatched vertex is the sum of all compensation it receives.
    \end{enumerate}
\end{definition}
Notice that we only make active vertices pay compensations, and the amount of compensation paid depends on the rank value of the passively matched neighbor. Similarly, when performing gain-sharing rule $1$ between a pair of matched vertices, unlike in the \ranking{} gain-sharing scheme (\cref{gain_sharing_scheme_ranking}), we only consider the rank value of the passive vertex. This is because, in general, we have limited knowledge of the rank of an active vertex, as discussed previously. 

In our final approximations, we will never collect more than one copy of compensation. However, we are still actively using the fact that the alternating path has length $\geq 4$ whenever possible. For example, in the case $u$ is unmatched due to $u^*$, we do not ask $u^\ast$ to pay compensation but still collect one copy of compensation from a third vertex $w$ (see \cref{franking_lower_bound_Abot}) for $u$. We need the fact that $w\neq u^\ast$ (equivalently, alternating path has length $\geq 4$) for such gain estimation. 
 
We also enforce the following constraints to $g,h$ to simplify the analysis:
\begin{definition}[Function Constraints for \Franking{}]
\phantomsection
\label{def:franking_func_constraints}
For \Franking{}, we assume the gain function $g$ and the compensation function $h$ satisfy the following constraints:
    \begin{enumerate}
    \item $g(x)$ is monotonically increasing.
    \item $h(x)$ is monotonically increasing.
    \item $h(0)=0$.
    \item $\forall x$, $1-g(x)-h(x)\geq h(1)$.
    \item $\forall x$,  $g(x)\geq h(1)$.
\end{enumerate}
\end{definition}

It turns out that in the \Franking{} analysis, we only collect up to one copy of compensation. Hence constraints 4 and 5 only require the LHS to be $h(1)$ to ensure that being unmatched represents the worst-case scenario.

\subsection{Approximation With Profiles}
Again, we calculate the expected value of gain$(u)$ $+$ gain$(u^\ast)$ through randomization of $\vecx_{\shortminus u^*}$ and $u^*$ successively. For each fixed $\vecx_{\shortminus u^*}$, besides the rank information for $u,v,b$, we also distinguish whether $u$ is active or passive in each match $(u,v),(u,b)$ and incorporate this information into the profile of $u$. We define the notion of backup as before in \cref{def:backup_querycommit}. We briefly discuss what being worse off means in \Franking{} and its effect on the matching condition of $(u,b)$. A similar discussion can be found in \cite{0.521Franking}, Lemma 2.5.
\begin{fact}[Worse-Off Scenarios in \Franking{}]
\phantomsection
\label{fact:franking_backup}
Suppose $u$ is matched to $v$. If $u$ becomes worse off by removing or introducing a vertex $w$, then:
\begin{enumerate}
    \item If $u$ is passively matched to $v$, it may become passive, active, or unmatched in the worse-off case.
    \item If $u$ is actively matched to $v$, then it either actively matches $b$ with $x_v < x_b$ or becomes unmatched in the worse-off case.
\end{enumerate}
Consequently, if $u$ has a backup $b$ with respect to $\vecx_{\shortminus u^*}$, following from \cref{fact:backup_uniquely_describe_worse_off}:
\begin{itemize}
    \item If $u$ is passively matched to $v$ in $\M(\vecx_{\shortminus u^\ast})$, then $u$ may be passively or actively matched to $b$ in $\M(\vecx_{\shortminus vu^*})$.
    \item If $u$ is actively matched to $v$ in $\M(\vecx_{\shortminus u^\ast})$, then $u$ is actively matched to $b$ in $\M(\vecx_{\shortminus vu^*})$ and the ranks of $v,b$ satisfy $x_v < x_b$.
\end{itemize}
If $u$ has no backup, being worse off means $u$ is unmatched.
\end{fact}
   This is directly due to the definition of being worse off equals being matched at a later time. If $u$ is passively matched to $v$, then it is matched at time $(\pi(v),\;x_u)$, where $\pi(v)<\pi(u)$. $u$ being matched at a later time to $b$ only suggests that $\pi(b)>\pi(v)$ and we can infer little information on the rank relation between $x_v$ and $x_b$. In contrast, if $u$ is actively matched to $v$, we know $u$ is matched at time $(\pi(u),\;x_v)$, and $u$ is matched to $b$ at a later time $(\pi(u),\;x_{b})$, which implies $x_b>x_v$. 
   
This fact implies that there are six subcases for the profile of $u$ with respect to $\vecx_{\shortminus u^*}$, taking into account whether $u$ is active or passive in each match.
\begin{definition}[Profile for $\vecx_{\shortminus u^\ast}$ in \Franking{}]
\phantomsection
\label{def:profile_franking}
For each $\vecx_{\shortminus u^\ast}$, denote the profile of $\vecx_{\shortminus u^\ast}$ by a triplet $(x_u,x_v,x_b)$, where $x_v$ denotes the rank of the match of $u$ in $\M(\vecx_{\shortminus u^\ast})$ (i.e., $v$), and $x_b$ denotes the rank of the backup of $u$ with respect to $\vecx_{\shortminus u^\ast}$ (i.e., $b$). The variables $x_v$ and $x_b$ may carry superscripts $A$ or $P$, where the superscripts ${A, P}$ indicate whether $u$ is active or passive when matched to $v$ and $b$, respectively. By \cref{fact:franking_backup}, there are in total six possible profiles:
    \begin{enumerate}
    \item $(x_u,\bot,\bot)$ where $u$ is not matched in $\M(\vecx_{\shortminus u^\ast})$ and does not have a backup. 
    \item $(x_u,x_v^P,\bot)$  where $u$ is passively matched to $v$ in $\M(\vecx_{\shortminus u^\ast})$ but does not have a backup.
    \item $(x_u,x_v^P,x_b^P)$ where $u$ is passively matched to $v$ in $\M(\vecx_{\shortminus u^\ast})$ and is passively matched to $b$ in $\M(\vecx_{\shortminus vu^\ast})$.
     \item $(x_u,x_v^P,x_b^A)$ where $u$ is passively matched to $v$ in $\M(\vecx_{\shortminus u^\ast})$ and is actively matched to $b$ in $\M(\vecx_{\shortminus vu^\ast})$.
      \item $(x_u,x_v^A,\bot)$  where $u$ is actively matched to $v$ in $\M(\vecx_{\shortminus u^\ast})$ but does not have a backup.
      \item $(x_u,x_v^A,x_b^A)$ where $u$ is actively matched to $v$ in $\M(\vecx_{\shortminus u^\ast})$ and is actively matched to $b$ in $\M(\vecx_{\shortminus vu^\ast})$.
\end{enumerate}
\end{definition}
 We next derive a case-by-case lower bound for all six subcases. We denote the corresponding lower bounding function as $GF$, where $F$ stands for \Franking{}.
\subsection{Lower Bound for $(x_u,\bot,\bot)$}
 We can lower bound the expected gain in this case as follows:
\begin{claim}
\phantomsection
\label{franking_lower_bound_botbot}
Define
    $$GF(x_u,\bot,\bot)=\int_0^1 g(x) dx,$$
    and we have
    $$GF(x_u,\bot,\bot)\leq \ev[\text{ gain($u$) $+$ gain($u^\ast$)}\;|\; \text{profile of $u=(x_u,\bot,\bot)$}].$$
\end{claim}
\begin{proof}
In this subcase, $u$ is unmatched in $\M(\vecx_{\shortminus u^*})$. If $u^*$ is still unmatched by the time $u$ picks its match in $\M(\vecx_{\shortminus u^*}(x))$, $u$ will actively match $u^\ast$. Consequently, $u^*$ is always matched passively in $\M(\vecx_{\shortminus u^\ast}(x))$ and receives $g(x)$ amount of gain.
\end{proof}
\subsection{Lower Bound for $(x_u,x_v^P,x_b^P)$}
\begin{claim}
\phantomsection
\label{franking_lower_bound_PP}
Define
    $$GF(x_u,x_v^P,x_b^P)=g(x_u),$$
    and we have
    $$GF(x_u,x_v^P,x_b^P)\leq \ev[\text{ gain($u$) $+$ gain($u^\ast$)}\;|\; \text{profile of $u=(x_u,x_v^P,x_b^P)$}].$$
\end{claim}
\begin{proof}
In this subcase, $u$ is passively matched to $v$ and is also passively matched to $b$. Introducing $u^*$ may cause $u$ to become worse off or not. If $u$ becomes worse off, the active backup $b$, by \cref{fact:franking_backup}, ensures that $u$ is still matched passively in $\M(\vecx)$. On the other hand, by the definition of being better off for \Franking{} \cref{fact:franking_backup}, if $u$ is passively matched to $v$ and does not become worse off by introducing $u^\ast$, then $u$ is still passively matched to its new match. In either case, $u$ remains passively matched and receives $g(x_u)$ amount of gain.
\end{proof}
\subsection{Lower Bound for $(x_u,x_v^P,\bot)$}
We first define the impacting rank $\theta_0$ the same as before in \cref{def:theta_0_ranking}:
\begin{definition}[Impacting Rank $\theta_0$]
\phantomsection
\label{def:theta_0_franking}
     For \Franking{}, $\theta_0$ with respect to $\vecx_{\shortminus u^*}$ is defined as
    $$\theta_0=\sup_{x\in[0,1]}\{x\;|\; u \text{ is made worse off in $\M(\vecx_{\shortminus u^*}(x))$}\}.$$
\end{definition}

We have the following property about $\theta_0$ for \Franking{}, similar to \cref{fact:theta_0_II}, property $2$:

\begin{fact}[Property of $\theta_0$, \Franking{}]
\phantomsection
\label{fact:theta_0_franking}
Fix $\vecx_{\shortminus u^*}$, let $\theta_0$ be the impacting rank as defined in \cref{def:theta_0_franking} with respect to $\vecx_{\shortminus u^*}.$ Assuming $\theta_0 > 0$, which is the non-trivial case, let $S$ be the collection of ranks $x$ that defines $\theta_0$. I.e.
    $$S=\{x\;|\; \text{$u$ is made worse off in $\M(\vecx_{\shortminus u^*}(x))$}\}.$$
    Then, there exists a small enough punctured left neighborhood of $\theta_0$, denoted as $(\theta_0^-,\theta_0)$, such that $(\theta_0^-,\theta_0)\subseteq S$. For any $x\in(\theta_0^-,\theta_0)$, the third vertex $u_2$ in the alternating path $\M(\vecx_{\shortminus u^*}(x))\oplus \M(\vecx_{\shortminus u^*})$ has rank $x_{u_2}= \theta_0$. Further, $u_2$ is passively matched to $u_1$ in $\M(\vecx_{\shortminus u^*})$.
\end{fact}
\begin{proof}
The construction of $(\theta_0^-,\theta_0)$ is the same as \cref{fact:theta_0_II}, property $2$. We only need to show that $u_1$ actively matches $u_2$ in $\M(\vecx_{\shortminus u^*})$. Fix $x\in (\theta_0^-,\theta_0)$. Notice that $u_2$ is the backup of $u_1$ with respect to $\vecx_{\shortminus u^*}(x)$, when $u_1$ is matched to $u^*$. By \cref{fact:matching_guarantee_franking}, $u^*$ being inserted at $x$ causes an alternating path that makes $u$ worse off implies that $u_1$ actively matches $u^*$. By \cref{fact:franking_backup}, this implies $u_1$ also actively matches the backup $u_2$ in $\M(\vecx_{\shortminus u^*})$.
\end{proof}

\begin{claim}Define
\phantomsection
\label{franking_lower_bound_Pbot}
     $$GF(x_u,x_v^P,\bot)=\min_{\theta_0}\left[\int_0^{\theta_0} g(x)\; dx +(1-\theta_0)\cdot h(\theta_0)+(1-\theta_0)\cdot g(x_u)\right],$$
    and we have
    $$GF(x_u,x_v^P,\bot)\leq \ev[\text{ gain($u$) $+$ gain($u^\ast$)}\;|\; \text{profile of $u=(x_u,x_v^P,\bot)$}].$$
\end{claim}
\begin{proof}
    The expected gain of $u^*$ can be lower bounded by:
$$\int_{0}^{\theta_0} g(x) \;dx\;+\;(1-\theta_0)\cdot h(\theta_0).$$
When $x=\theta_0^-$, $u^\ast$ triggers an alternating path that makes $u$ worse off. Since $u$ has an earlier decision time, the alternating path lemma (\cref{lem:alt-path}) implies that $u^\ast$ must be passively matched to some vertex $w$, otherwise $u$ cannot be affected by $u^\ast$. This fact and \cref{fact:matching_guarantee_franking} property 3 imply that $u^*$ is always passively matched for $x\in [0,\theta_0)$. Therefore, for $x\in[0,\theta_0)$, $u^*$ receives at least $g(x)$ amount of gain. 

For $x\in(\theta_0,1]$, $u^\ast$ is unmatched in the worst case. Fix $x_{\theta_0}\in (\theta_0^-,\theta_0)$ as given by \cref{fact:theta_0_franking}. By \cref{lem:victim_querycommit_case_u_star}, if $u^*$ is unmatched, then it is the victim of $w$, its match in $\M(\vecx_{\shortminus u^\ast}(x_{\theta_0}))$, which is $u_1$ in the alternating path $\M(\vecx_{\shortminus u^\ast}(x_{\theta_0}))\oplus \M(\vecx_{\shortminus u^*})$. Since the match of $w$ in $\M(\vecx_{\shortminus u^*}(x))$, assuming $u^*$ is unmatched, is $u_2$, which by \cref{fact:theta_0_franking}, has rank $\theta_0$. $w$ will pay $h(\theta_0)$ amount of compensation to $u^*$.

When $\theta_0=0$, such an argument fails as $x_{\theta_0}$ does not exist. However, since we set $h(0)=0$ in the function constraints (\cref{def:franking_func_constraints}), we are not over-approximating the expected gain. For those $x\in(\theta_0,1]$ such that $u^*$ is matched in $\M(\vecx_{\shortminus u^\ast}(x))$, we can still use $h(\theta_0)$ as a lower bound on the gain, since we assumed that a matched vertex always receives more gain than when it is unmatched by the function constraints.

The expected gain of $u$ is at least $(1-\theta_0)\cdot g(x_u)$. Since $u$ is not made worse off for $x\in(\theta_0,1]$, it remains passively matched and receives gain $g(x_u)$.
\end{proof}
\subsection{Lower Bound for $(x_u,x_v^P,x_b^A)$}
\begin{claim}
\phantomsection
\label{franking_lower_bound_PA}
Define
\begin{align*}
    GF(x_u,x_v^P,x_b^A)=&\min_{\theta_0}\left[\int_0^{\theta_0} g(x) dx +\max\{x_b-\theta_0,0\}\cdot h(\theta_0)\right. \\
    &\quad\quad\quad\left.+\theta_0\cdot (1-g(x_b)-h(x_b))+(1-\theta_0)\cdot g(x_u)\vphantom{\int}\right],
\end{align*}
    and we have
    $$GF(x_u,x_v^P,x_b^A)\leq \ev[\text{ gain($u$) $+$ gain($u^\ast$)}\;|\; \text{profile of $u=(x_u,x_v^P,x_b^A)$}].$$
\end{claim}
\begin{proof}
    The proof is similar to the previous case. For the expected gain of $u^\ast$: When $x\in[0,\theta_0)$, $u^*$ has gain of at least $g(x)$, as $u^\ast$ is passively matched in this range. For $x\in(\theta_0,1]$, \cref{lem:victim_querycommit_case_u_star} guarantees that $u^*$ is a victim only if $u$ prefers $u^*$ over its backup $b$, i.e., when $x<x_b$. Thus, $u^*$ is the victim of $w$ in the range $(\theta_0,x_b)$, assuming it is unmatched, and its expected gain in this range is lower bounded by $\max\{x_b-\theta_0,0\}\cdot h(\theta_0)$. Similar to before, we can always use $h(\theta_0)$ as a lower bound for the gain, as being matched always receives more gain than being unmatched by the function constraints (\cref{def:franking_func_constraints}).

    For the expected gain of $u$: In the range $x\in(\theta_0,1]$, $u$ is not made worse off and hence has gain $g(x_u)$. When $x\in [0,\theta_0)$, either $u^*$ is not made worse off, in which case it receives at least $g(x_u)$ amount of gain, or it is made worse off, and by \cref{fact:franking_backup}, it is matched actively to $b$. So $u^*$ has gain at least $\min\{g(x_u),1-g(x_b)-h(x_b)\}$. If $g(x_u)$ is the smaller one, then the expected gain of $u$ alone is at least $1\cdot g(x_u)$, which is the same as $GF(x_u,x_v^P,x_b^P)$. In the final approximation (\cref{claim:Franking_lower_bound_uniform_profiles_P}), we will consider the worst-case scenario among $\{GF(x_u,x_v^P,x_b^A), GF(x_u,x_v^P,x_b^P),...\}$,
    hence, we can safely assume $1-g(x_b)-h(x_b)$ is the worse-case choice and use $1-g(x_b)-h(x_b)$ as a lower bound for the expected gain for $x \in [0,\theta_0)$.
\end{proof}
\subsection{Lower Bound for $(x_u,x_v^A,\bot)$}
If $u$ is actively matched to $v$ before $u^*$ is introduced, it becomes possible that introducing $u^*$ makes $u$ better off (i.e., become passively matched) while obtaining a worse gain overall, in the case where $g(x_u)\leq 1-g(x_v)-h(x_v)$. Therefore, we need to take this situation into account. By \cref{fact:matching_guarantee_franking}, property 2, this situation occurs only when $u^*$ is matched passively in $\M(\vecx_{\shortminus u^\ast}(x))$. So we define the notion of marginal rank $\theta_1$, the same as in \cite{0.521Franking}, Definition 2.3. By defining $\theta_1$, we only need to consider the above situation when $x\in[0,\theta_1)$.
\begin{definition}[Marginal Rank $\theta_1$]
\phantomsection
\label{theta_1_franking}
    The marginal rank $\theta_1$ with respect to $\vecx_{\shortminus u^*}$ is defined as:
    $$\theta_1=\sup_{x\in[0,1]}\{x\;|\; \text{$u^*$ is matched passively in $\M(\vecx_{\shortminus u^*}(x))$}\}.$$
\end{definition}
\begin{fact}[Properties of $\theta_1$]
\phantomsection
\label{fact:theta_1}
We have the following properties for $\theta_1$:
\begin{enumerate}
    \item $\theta_0\leq \theta_1$.
    \item If $u$ actively matches $v$ in $\M(\vecx_{\shortminus u^*})$, then $\theta_1\geq x_v$.
\end{enumerate}
\end{fact}
\begin{proof}
     $(1)$ holds by property 1 of the matching guarantees (\cref{fact:matching_guarantee_franking}), which says that $u$ will not become worse off when $u^*$ is matched actively. $(2)$ holds by property $2$ of the matching guarantees. When $u$ is matched actively to $v$, placing $u^*$ before $v$ always guarantees that $u^*$ is matched passively.
\end{proof}
With the new threshold $\theta_1$ defined, we can now construct function $GF$ for profile $(x_u,x_v^A,\bot)$.
\begin{claim}
\phantomsection
\label{franking_lower_bound_Abot}
For each fixed $x_u, x_v$, define
\begin{align*}
    GF(x_u,x_v^A,\bot)=&\min_{\theta_0\leq\theta_1\;,\; x_v\leq \theta_1}\left[\int_0^{\theta_1} g(x) dx + (1-\theta_1)\cdot h(\theta_0)\right.\\
    &\quad\quad\quad\quad\quad\quad +\;\theta_0\cdot h(x_v)\\
    &\quad\quad\quad\quad\quad\quad +\;(\theta_1-\theta_0)\cdot \min\{g(x_u),1-g(x_v)-h(x_v)\}\\
    &\quad\quad\quad\quad\quad\quad + (1-\theta_1)\cdot (1-g(x_v)-h(x_v)) \left.\vphantom{\int}\right].
\end{align*}
    We have
    $$GF(x_u,x_v^A,\bot)\leq \ev[\text{ gain($u$) $+$ gain($u^\ast$)}\;|\; \text{profile of $u=(x_u,x_v^A,\bot)$}].$$
\end{claim}
\begin{proof}
    By \cref{fact:theta_1} we know that ranks $\theta_0,\theta_1$ satisfy the property $\theta_0\leq \theta_1$ and $x_v\leq \theta_1$, so we impose these restrictions for the worst-case analysis. 
    
     Fix any $\theta_0,\theta_1$. We first calculate the expected gain of $u^*$, which is
     $$\int_0^{\theta_1} g(x) dx + (1-\theta_1)\cdot h(\theta_0)$$
     We can apply the same proof as in \cref{franking_lower_bound_Pbot} and conclude this is the worst-case expected gain of $u^\ast$.

     The expected gain of $u$ consists of three parts. When $x\in[0,\theta_0)$, $u$ can be unmatched due to $u^*$. Let $w$ be the match of $v$ in $\M(\vecx_{\shortminus u^*}(x))$, assuming $u$ is unmatched. We show that, in this case, $w$ sends $h(x_v)$ amount of compensation to $u$. By \cref{lem:victim_querycommit_case_u}, $u$ is the victim of $w$ in this case. So we only need to show that $w$ is active in $\M(\vecx_{\shortminus u^\ast}(x))$. By \cref{lem:backup_in_alt_path}, $u$ is the backup of $v$ in $\vecx_{\shortminus u^*}(x)$, hence $u$ is a worse off choice for $v$ compared to $w$, since $v$ is matched to $u$ passively, a better off choice for $v$ will also be an active vertex. Hence, $w$ is also active.
     
     When $x\in(\theta_0,\theta_1)$, $u$ is not made worse off, but it is possible that $u$ is made better off, this means that $u$ is either actively matched to a vertex of rank $\leq x_v$ or becomes passively matched. Therefore, $u$ receives at least $\min\{g(x_u),1-g(x_v)-h(x_v)\}$ amount of gain in this range. 
     
     For $x\in(\theta_1,1]$, by \cref{fact:matching_guarantee_franking}, $u^*$ does not affect the match of $u$, so $u$ receives at least $1-g(x_v)-h(x_v)$ amount of gain.
\end{proof}
\subsection{Lower Bound for $(x_u,x_v^A,x_b^A)$}
\begin{claim}
\phantomsection
\label{franking_lower_bound_AA}
Define
\begin{align*}
    GF(x_u,x_v^A,x_b^A)=&\min_{\theta_0\leq\theta_1\;,\; x_v\leq \theta_1}\left[\int_0^{\theta_1} g(x) dx + \max\{x_b-\theta_1,0\}\cdot h(\theta_0)\right.\\
    &\quad\quad\quad\quad\quad\quad +\;\theta_0\cdot \min\{g(x_u),1-g(x_b)-h(x_b)\}\\
    &\quad\quad\quad\quad\quad\quad +\;(\theta_1-\theta_0)\cdot \min\{g(x_u),1-g(x_v)-h(x_v)\}\\
    &\quad\quad\quad\quad\quad\quad + (1-\theta_1)\cdot (1-g(x_v)-h(x_v)) \left.\vphantom{\int}\right].
\end{align*}
    and we have
    $$GF(x_u,x_v^A,x_b^A)\leq \ev[\text{ gain($u$) $+$ gain($u^\ast$)}\;|\; \text{profile of $u=(x_u,x_v^A,x_b^A)$}].$$
\end{claim}
\begin{proof}
We only prove the part that differs from the case $(x_u,x_v^A,\bot)$; the rest follows from the same proof. For $x\in(\theta_1,1]$, if $u^\ast$ is unmatched, similar to \cref{franking_lower_bound_PA}, the condition for $u^\ast$ to be a victim requires $u^*$ to be a better choice for $u$ compared to $b$ by \cref{lem:victim_querycommit_case_u_star}. Thus, $u^*$ receives $h(\theta_0)$ amount of compensation in the range $x\in(\theta_1,x_b)$.

For $x\in[0,\theta_0)$, since $u$ has a backup, the worst-case passive match of $u$ is $b$. The minimum amount of gain $u$ receives is the smaller one between $g(x_u)$ ($u$ being passively matched) and $1-g(x_b)-h(x_b)$ ($u$ being actively matched to $b$). Therefore, the lower bound for the gain of $u$ is $\min\{g(x_u),1-g(x_b)-h(x_b)\}$.
\end{proof}
\subsection{Expectation Over Worst Case Distribution of Profiles}
We have the same monotonicity properties for the distribution of $x_v$ in \Franking{} as in \ranking{}, except that we further need $u$ to be active. 
\begin{lemma}[Monotonicity Properties of Profiles in \Franking{}]
\phantomsection
\label{lemma:monotonicity_Franking}
Fix $x_u,x_b\in[0,1]$. Suppose $u$ is actively matched to $v$ with rank $x_v$ in $\M(\vecx_{\shortminus u^*})$.
\begin{itemize}
    \item If $u$ does not have a backup, then for any $x \in [x_v,1]$, demoting the rank of $v$ to $x$ does not change the resulting matching.
    
    Consequently, $x_v$ has monotonically increasing probability density function in $[0,1]$ with respect to profiles $(x_u, x_v^A,\bot).$
    \item If $u$ has a backup $b$ at rank $x_b$, then for any $x \in [x_v,x_b)$, demoting the rank of $v$ to $x$ does not change the resulting matching.
    
    Consequently, $x_v$ has monotonically increasing probability density function in $[0,x_b)$ with respect to profiles $(x_u, x_v^A,x_b^A).$
\end{itemize}
\end{lemma}
See \nameref{sec:appendix-a} for a proof. 

We split the profile of $u$ into two large subclasses $P$ and $A$, depending on the active/passiveness of $u$ when it is matched to $v$ in $\M(\vecx_{\shortminus u^*})$. Class $P$ consists of profiles of type $(x_u,x_v^P,\bot)$, $(x_u,x_v^P,x_b^P)$, and $(x_u,x_v^P,x_b^A)$; class $A$ consists of profiles of type $(x_u,x_v^A,\bot)$, $(x_u,x_v^A,x_b^A)$, and $(x_u,\bot,\bot)$. We have the following property with respect to class $P$ and $A$:
\begin{fact}[Marginal Rank for $u$]
\phantomsection
\label{fact:theta_1_u}
    Fix an arbitrary $\vecx_{\shortminus uu^*}$, there exists a marginal rank $\theta_1^u$ for $u$ such that if $x_u<\theta_1^u$, then the profile of $u$ with respect to $\vecx_{\shortminus u^*}$ is of class $P$; if $x_u>\theta_1^u$, then the profile of $u$ with respect to $\vecx_{\shortminus u^*}$ belongs to class $A$.
\end{fact}
This comes easily from \cref{theta_1_franking} together with \cref{fact:matching_guarantee_franking}, property 3, applied to rank vector $\vecx_{\shortminus uu^*}$. We simply define $\theta_1^u$ as the marginal rank of $u$ with respect to $\vecx_{\shortminus uu^*}$. Before $\theta_1^u$, $u$ is matched passively and after $\theta_1^u$, $u$ is matched actively or unmatched. We use the case-by-case bounding functions $GF$ for $x_u<\theta_1^u$ and $x_u>\theta_1^u$ respectively. Fix arbitrary $\theta_1^u\in[0,1]$.
\begin{claim}
\phantomsection
\label{claim:Franking_lower_bound_uniform_profiles_P}
    If $x_u<\theta_1^u$, then the expected value of gain$(u)$ $+$ gain$(u^\ast)$ is lower bounded by \footnote{$GF$ is independent of $x_v^P$ and $x_b^P$, so we might aswell assume $x_v^P=x_b^P=0$ when not mentioned.}
    \begin{align*}
  GF_u^P(x_u) = \min\;\left\{
    \begin{array}{c}
     \ GF(x_u,x_v^P,\bot), \\[4pt]
      GF(x_u,x_v^P,x_b^P), \\[4pt]
      \displaystyle \inf_{0<x_b^A\leq1}\biggl[GF(x_u,x_v^P,x_b^A)\biggr].
    \end{array}
  \right\}.
\end{align*}
\end{claim}
\begin{proof}
    When $x_u<\theta_1^u$, the profile of $u$ with respect to $\vecx_{\shortminus u^*}$ falls in class $P$. By taking the worst case bound for all 3 types of profiles in class $P$ (\cref{franking_lower_bound_Pbot,franking_lower_bound_PP,franking_lower_bound_PA}), we get a the above bound.
\end{proof}

\begin{claim}
\phantomsection
\label{claim:Franking_lower_bound_uniform_profiles_A}
    If $x_u>\theta_1^u$, then the expected value of gain$(u)$ $+$ gain$(u^\ast)$ is lower bounded by
\begin{align*}
  GF_u^A(x_u) = \min\;\left\{
    \begin{array}{c}
      GF(x_u,\bot,\bot), \\[4pt]
      \displaystyle \inf_{0\le v_0<1}\biggl[\frac{1}{1-v_0}\int_{v_0}^{1} GF(x_u,x_v^A,\bot)\,dx_v\biggr], \\[4pt]
      \displaystyle \inf_{0\le v_0<b_0\leq 1}\biggl[\frac{1}{b_0 - v_0}\int_{v_0}^{b_0} GF(x_u,x_v^A,x_b^A=b_0)\,dx_v\biggr]
    \end{array}
  \right\}.
\end{align*}
\end{claim}
\begin{proof}
    When $x_u>\theta_1^u$, the profile of $u$ with respect to $\vecx_{\shortminus u^*}$ falls in class $A$. By the monotonicity property (\cref{lemma:monotonicity_Franking}) and \nameref{uniform-profiles}, we can take the worst case uniform distribution to bound the 3 types of profiles in class $A$ (\cref{franking_lower_bound_botbot,franking_lower_bound_Abot,franking_lower_bound_AA}), which results in the above bound.
\end{proof}
Now we can lower bound the expected gain over $x_u$ assuming the worst case $\theta^u_1$. i.e.
\begin{claim}
    The expected value of gain$(u)$ $+$ gain$(u^\ast)$ for \Franking{} is lower bounded by
    $$\min_{\theta_1^u}\left[\int_0^{\theta_1^u} GF_u^P(x_u)\;dx_u + \int_{\theta_1^u}^{1} GF_u^A(x_u) \;dx_u\right]$$
\end{claim}
\begin{proof}
    For each fixed $\vecx_{\shortminus uu^*}$, LHS is a lower bound for the expected gains, as by \cref{fact:theta_1_u}, there exists a $\theta_1^u$ for each $\vecx_{\shortminus uu^*}$ such that
    $$\int_0^{\theta_1^u} GF_u^P(x_u)\;dx_u + \int_{\theta_1^u}^{1} GF_u^A(x_u) \;dx_u\leq \ev_{\vecx}[\text{ gain($u$) $+$ gain($u^\ast$)}\;|\; \text{$\vecx$ agrees $\vecx_{\shortminus uu^*}$ }]$$
\end{proof}

\section{Discretization}
In both the \ranking{} and \Franking{} LPs, there are two systems of variables that could be optimized. They are: 1. the gain function $g$ and the compensation function $h$, and 2. the distribution of profiles and various cutting points $(\theta_0,\theta_{u^\ast},\dots)$. We need to fix one of them to make the inequality system an LP so that we can solve it using LP solvers. In this work, we follow the approach of fixing 2 and taking 1 as a variable.

Fix a size parameter $n$. The unit interval $(0,1]$\footnote{We assume all vertices $v$ have $x_v>0$, as the event $x_v=0$ has 0 probability measure for any vertex $v$.} can be uniformly discretized into $n$ pieces; we will use $i\in {1,\dots,n}$ to indicate the interval $(\frac{i-1}{n},\frac{i}{n}]$. For a fixed size $n$, we will only consider the functions $g,h$ as piecewise constant functions on each grid block $(\frac{i-1}{n},\frac{i}{n}]\times(\frac{j-1}{n},\frac{j}{n}]$. Since $g$ and $h$ are piecewise constant within each grid block, our final LP will also consist of all the variables $x_u,x_v,x_b,\theta_0...$ fixed at each grid point $\frac{i}{n}$ or $(\frac{i}{n})^+$ for $i\in\{0,...,n\}$). 

Intuitively, this is because $g$ and $h$ are constant on each block, their values depend only on which sub-interval each parameter lies in, and not on the exact value within the sub-interval. Therefore, once the sub-intervals containing the parameters $x_u,x_v,\dots$ are fixed, all occurrences of $g$ and $h$ become constants. As a result, the objective and constraints become linear functions of the parameters. Hence, the worst-case value over each such region is attained at the boundary, i.e., when each parameter is set to either $(\frac{i}{n})^+$ or $\frac{i}{n}$.
\subsection{Discretized \ranking{} LP}
\label{subsec:discre-ranking_LP}
Fix an arbitrary integer $n>0$. The piecewise gain function $g$ and the compensation function $h$ for the \ranking{} LP can be described by variables $g(i,j)$ and $h(k,l)$, where $1\leq i,j\leq n$ and $0\leq k,l\leq  n$. For $x,y\in (0,1]$, $g(x,y)=g(i,j)$ if $x\in(\frac{i-1}{n},\frac{i}{n}]$ and $y\in (\frac{j-1}{n},\frac{j}{n}]$. Similarly, $h(x,y)=h(k,l)$ if $x\in(\frac{k-1}{n},\frac{k}{n}]$ and $y\in(\frac{l-1}{n},\frac{l}{n}]$. Again, we reserve $h(x,0)$ for the degenerate cases $\theta_0=0$. The function constraints (\cref{def:function_constraint_ranking}) translate to
\begin{align*}
    &1.\quad\forall i\;,\forall j<n && g(i,j) \le g(i,j+1),\\
    &\phantom{1.}\quad\forall k\;,\forall l<n  && h(k,l) \le h(k,l+1),\\
    &2.\quad\forall i<n,\forall j && g(i,j) \ge g(i+1,j),\\
    &\phantom{2.}\quad\forall k<n,\forall l  && h(k,l) \ge h(k+1,l),\\
    &3.\quad\forall k &&h(k,0)=0,\\
    &4.\quad\forall i,j && 1 - g(i,j) - h(i,j) \ge 4h(1,n),\\
    &5.\quad\forall i,j && g(i,j) - h(j,i) \ge 4h(1,n).
\end{align*}
The buyer, product gain functions assuming compensation, $g_B,g_P$ (\cref{def:gain_with_comp_function}) are translated to:
\begin{align*}
    &\forall i,j  &&g_B(i,j)=1-g(i,j)-h(i,j),\\
    &\forall i,j  &&g_P(i,j)=g(i,j)-h(j,i).
\end{align*}

Now we continue to discretize the $G$ functions. \textbf{We assume all $\bm{i}$ indices take values in $\bm{\{1,...,n\}}$, except $\bm{i_{\theta_0}}$, which may take $\bm{0}$}.

For $G(x_u,\bot,\bot)$ (\cref{G_ranking_nomatch}), where $x_u\in(\frac{i_u-1}{n},\frac{i_u}{n}]$ the bound directly translates to
$$\forall i_u,\quad\quad G(i_u,\bot,\bot)\leq \frac{1}{n}\left(\sum_{j=1}^n g_P(i_u,j)\right),$$
by directly integrating over the step function $g_P.$ 

For $G(x_u,x_v,\bot)$ (\cref{G_ranking_nobackup}), where $x_u\in(\frac{i_u-1}{n},\frac{i_u}{n}]$ and $x_v\in (\frac{i_v-1}{n},\frac{i_v}{n}]$, the bound translates to: 
\begin{align*}
\shortintertext{$\forall i_{\theta_0}\leq i_u\leq i_v$:}
\qquad G(i_u,i_v,\bot)\leq{}&
\frac{1}{n}\left[\left(\sum_{j=1}^{i_v-1} g_P(i_u,j)\right)+\frac{1}{2}g_P(i_u,i_v)\right. \\
&\quad +(n-i_v)\cdot(h(i_u,i_{\theta_0})+h(i_v,i_u))\\
&\quad +\left(\sum_{j=1}^{i_{\theta_0}} h(j,i_u)\right)+i_{\theta_0}\cdot h(i_u,i_v)+ (n-i_{\theta_0})\cdot g_B(i_u,i_v)\left. \vphantom{\int} \right],& (1) \\[0.6em]
\shortintertext{$\forall  i_{\theta_0},i_v\leq i_u$:}
\qquad G(i_u,i_v,\bot)\leq{}&
\frac{1}{n}\left[\left(\sum_{j=1}^{i_{\theta_0}} g(i_v,j)\right)+\left(\sum_{j=i_{\theta_0}+1}^{i_{v}-1} g_P(i_u,j)\right) \right.\\
&\quad +(n-\max\{i_{\theta_0},i_v-1\})\cdot(h(i_v,i_{\theta_0})+h(i_v,i_u)) \\
&\quad +i_{\theta_0}\cdot h(i_u,i_v)+ (n-i_{\theta_0})\cdot g_B(i_u,i_v)\left. \vphantom{\int} \right],& (2)\\[0.6em]
\shortintertext{$\forall i_v\leq i_u \wedge i_u-1\leq  i_{\theta_0}\leq i_u$:}
\qquad G(i_u,i_v,\bot)\leq{}&
\frac{1}{n}\left[\left(\sum_{j=1}^{i_{\theta_0}} g(i_v,j)\right)+ (n-i_{\theta_0})\cdot h(i_v,i_u) \right.\\
&\quad + (n-i_{\theta_0})\cdot g_B(i_u,i_v)\left. \vphantom{\int}\right],& (3)
\end{align*}
Where $(1),(2),(3)$ correspond to equations $(1),(2),(3)$ in \cref{G_ranking_nobackup}.

\noindent\textbf{Expression (1).} Instead of considering adversarial choice of $x_u,x_v$ within the sub-interval, $G(i_u,i_v,\bot)$ is a lower bound for the expected total gain conditioning on the event $x_u\leq x_v$, $x_u\in(\frac{i_u-1}{n},\frac{i_u}{n}]$, and $x_v\in (\frac{i_v-1}{n},\frac{i_v}{n}]$. Conditioning on this, we know that $x_v$ has a monotonically increasing pdf through the sub-interval $(\frac{i_v-1}{n},\frac{i_v}{n}]$ by \cref{monoto_before_backup_ranking}. Hence, we have that the expected value of the first term satisfies
$$\ev_{x_v}[\int_0^{x_v}g_P(x_u,x)dx]\geq \int_{\frac{i_v-1}{n}}^\frac{i_v}{n}\Bigl(\int_0^{x_v}g_P(x_u,x)\;dx\Bigr)\;dx_v = \frac{1}{n}\left[\left(\sum_{j=1}^{i_v-1} g_P(i_u,j)\right)+\frac{1}{2}g_P(i_u,i_v)\right]$$
As in the worst-case, $x_v$ is a uniform distribution from $\frac{i_v-1}{n}$ to $\frac{i_v}{n}$, and further, this term is independent of $\theta_0$. 

We bound the remaining terms in $(1)$ by considering worst-case choices for $\theta_0,x_v$. Fixing the sub-intervals  $\theta_0,x_v$ lie in, the remaining terms are monotonically decreasing with $x_v$ (the only term that is affected is $(1-x_v)$, so we take $x_v=\frac{i_v}{n}$ (the right-side boundary value). Further, for each fixed $i_{\theta_0}\in\{0,..,i_u\}$, grid-point choices are $\theta_0=\frac{i_{\theta_0}}{n}$ (right boundary value) and $\theta_0=(\frac{i_{\theta_0}-1}{n})^+$(left boundary value). Notice that for each grid point $i_{\theta_0}$, when comparing the equation value for $\theta_0=\frac{i_{\theta_0}}{n}$ and $(\frac{i_{\theta_0}}{n})^+$, we always  have $\theta_0=\frac{i_{\theta_0}}{n}$ resulting in a worse-case value. As the equation satisfies
{\small
\begin{align*}
    &(1-x_v)\cdot\Bigl(h(x_u,(\frac{i_{\theta_0}}{n})^+)+h(x_v,x_u)\Bigr)
    +\int_0^{(\frac{i_{\theta_0}}{n})^+} h(x,x_u)\,dx
    +(\frac{i_{\theta_0}}{n})^+\cdot\,h(x_u,x_v)+(1-(\frac{i_{\theta_0}}{n})^+)\cdot\,g_B(x_u,x_v)\\
    =&(1-x_v)\cdot\Bigl(h(x_u,(\frac{i_{\theta_0}+1}{n}))+h(x_v,x_u)\Bigr)
    +\int_0^{\frac{i_{\theta_0}}{n}} h(x,x_u)\,dx
    +\frac{i_{\theta_0}}{n}\cdot\,h(x_u,x_v)
    +(1-\frac{i_{\theta_0}}{n})\cdot\,g_B(x_u,x_v)\\
    \geq &(1-x_v)\cdot\Bigl(h(x_u,\frac{i_{\theta_0}}{n})+h(x_v,x_u)\Bigr)
    +\int_0^{\frac{i_{\theta_0}}{n}} h(x,x_u)\,dx
    +\frac{i_{\theta_0}}{n}\cdot\,h(x_u,x_v)
    +(1-\frac{i_{\theta_0}}{n})\cdot\,g_B(x_u,x_v).
\end{align*}
}
Intuitively, the choice of $\frac{i_{\theta_0}}{n}$ or $(\frac{i_{\theta_0}}{n})^+$ only affects the value of the $h$ function, where a smaller value is achieved when $\theta_0=\frac{i_{\theta_0}}{n}$. Hence we choose $\theta_0=\frac{i_{\theta_0}}{n}$ and the resulting system of bounds is $(1)$.

\noindent For expressions $(2),(3)$, we combined 
$$\int_0^{\theta_0} g_P(x_v,x)dx+\int_0^{\theta_0} h(x,x_v)dx=\int_0^{\theta_0} g(x_v,x)-h(x,x_v)\;dx+\int_0^{\theta_0} h(x,x_v)dx$$
into one term $$\int_0^{\theta_0} g(x_v,x)dx.$$ This is where terms of the form $g(i_v,j)$ arise.

\noindent\textbf{Expression (2).} For this case, we can no longer use the monotonicity of $x_v$, as the first term in the expression is no longer independent of $\theta_0$. So we consider the worst-case choice of $\theta_0$ and $x_v$ within each sub-interval for all the terms altogether. By the same line of reasoning as case $(1)$, we choose $\theta_0=\frac{i_{\theta_0}}{n}$ for $\theta_0\in\{0,...,i_{u}\}$ as it results in worse-case bounds. For $x_v$, we will choose $x_v=(\frac{i_{v}-1}{n})^+$. This is because when the exact value of $\theta_0,x_u$, and the sub-interval of $x_v$ are fixed, decreasing $x_v$ results in decreasing function values. 

To elaborate on this, when $x_v>\theta_0$, decrease $x_v$ by $\delta$ for $ \delta\in[0, x_v-\max\{\theta_0,\frac{i_v-1}{n}\})$\footnote{Decrease $x_v$ while making sure that it is still within the sub-interval $(\frac{i_v-1}{n},\frac{i_v}{n}]$ and also $\geq\theta_0$}, the function changes by a non-positive value 
$$\delta\cdot \Big[ (h(i_v,i_{\theta_0})+h(i_v,i_u))-g_P(i_u,i_v)\Big]\leq 0,$$
as the summation of two copies of compensation is smaller than or equal to a copy of matched gain (\cref{def:function_constraint_ranking}). On the other hand, if $x_v\leq\theta_0$, then decreasing $x_v$ within the sub-interval does not change the equation evaluation.

\noindent\textbf{Expression (3).} In this case, we consider two grid points $\theta_0=\frac{i_u}{n}$ and $\theta_0=(\frac{i_u-1}{n})^+$. The inequality range $i_u-1\leq i_{\theta_0}\leq i_u$ corresponds to the two cases considered.

For $G(x_u,x_v,x_b)$ (\cref{G_ranking_backup}), where $x_u\in(\frac{i_u-1}{n},\frac{i_u}{n}]$,  $x_v\in (\frac{i_v-1}{n},\frac{i_v}{n}]$ and $x_b\in (\frac{i_b-1}{n},\frac{i_b}{n}]$ ,the bounds are translated to : 
\begin{align*}
\shortintertext{$\forall i_{\theta_0}\leq i_u\leq i_v<i_b$:}
\qquad G(i_u,i_v,i_b)\leq{}&
\frac{1}{n}\left[\left(\sum_{j=1}^{i_v-1} g_P(i_u,j)\right)+ \frac{1}{2} g_P(i_u,i_v) \right.\\
&\quad +\max\{i_b-i_v-1,0\}\cdot(h(i_u,i_{\theta_0})+h(i_v,i_u)) \\
&\quad +i_{\theta_0}\cdot g_B(i_u,i_b)+ (n-i_{\theta_0})g_B(i_u,i_v)\left. \vphantom{\int} \right],& (1.1)\\[0.6em]
\shortintertext{$\forall i_{\theta_0}\leq i_u\leq i_v=i_b$:}
\qquad G(i_u,i_v,i_b)\leq{}&
\frac{1}{n}\left[\left(\sum_{j=1}^{i_v-1} g_P(i_u,j)\right)\right.\\
&\quad +i_{\theta_0}\cdot g_B(i_u,i_b)+ (n-i_{\theta_0})g_B(i_u,i_v)\left. \vphantom{\int} \right],& (1.2)\\[0.6em]
\shortintertext{$\forall i_{\theta_0}, i_v\leq i_u$:}
\qquad G(i_u,i_v,i_b)\leq{}&
\frac{1}{n}\left[\left(\sum_{j=1}^{i_{\theta_0}} g_P(i_v,j)\right)+\left(\sum_{j=i_{\theta_0}+1}^{i_v-1} g_P(i_u,j)\right) \right.\\
&\quad + \max\{i_b-1-\max\{i_{\theta_0},i_v-1\},0\}\cdot(h(i_v,i_{\theta_0})+h(i_v,i_u)) \\
&\quad +i_{\theta_0}\cdot g_B(i_u,i_b)+ (n-i_{\theta_0})g_B(i_u,i_v)\left. \vphantom{\int} \right],& (2)\\[0.6em]
\shortintertext{$\forall i_v\leq i_u\; \wedge\; i_u-1\leq i_{\theta_0}\leq i_u$:}
\qquad G(i_u,i_v,i_b)\leq{}&
\frac{1}{n}\left[\left(\sum_{j=1}^{i_{\theta_0}} g_P(i_v,j)\right)+ \max\{i_b-i_{\theta_0}-1,0\}\cdot h(i_v,i_u) \right.\\
&\quad +i_{\theta_0}\cdot g_B(i_u,i_b)+ (n-i_{\theta_0})g_B(i_u,i_v)\left. \vphantom{\int}\right].& (3)
\end{align*}
Constraints $(1.1),(1.2)$ together correspond to constraint family $(1)$ in \cref{G_ranking_backup}. $(2),(3)$ correspond to constraints $(2),(3)$ in the original bound respectively.

For all of the cases, fixing all parameters except $x_b$, the gain monotonically decreases as we decrease $x_b$ in $(\frac{i_b-1}{n},\frac{i_b}{n}]$, as the only term that is affected is 
$$\max\{x_b-\max\{\theta_0,x_v\},0\}\cdot C,$$
for some positive constant $C$. Hence we assume $x_b$ always takes value $(\frac{i_b-1}{n})^+$. 

\noindent\textbf{(1.1), (1.2).} Similar to equation $(1)$ in the $(x_u,x_v,\bot)$ case, we could effectively use the monotonicity of $x_v$ in $(\frac{i_v-1}{n},\frac{i_v}{n}]$ if we know that $i_v<i_b$. So in $(1.1)$ we have the term 
$$\left(\sum_{j=1}^{i_v-1} g_P(i_u,j)\right)+ \frac{1}{2} g_P(i_u,i_v)$$
In $(1.2)$, when $i_v=i_b$, we no longer have the monotonicity property as it is possible that $x_b=(\frac{i_v-1}{n})^+$. So we again perform the boundary worst-case analysis. The worst-case happens when $x_v,x_b$ both take the value $(\frac{i_v-1}{n})^+$, which results in equation $(1.2)$.

\noindent\textbf{(2), (3).} We again choose $\theta_0$ as $\frac{i_{\theta_0}}{n}$ and $x_v=(\frac{i_v-1}{n})^+$ based on the same arguments as the previous discretization of $G(x_u,x_v,\bot)$.

Combining the above three discretized constraints, we have the following:
\begin{claim}
    $G(i_u,\bot,\bot), G(i_u,i_v,\bot),G(i_u,i_v,i_b)$ lower bound $G(x_u,\bot,\bot ), G(x_u,x_v,\bot),G(x_u,x_v,x_b)$ respectively, conditioning on the event that $x_u\in (\frac{i_u-1}{n},\frac{i_u}{n}]$, $x_v\in (\frac{i_v-1}{n},\frac{i_v}{n}]$ and $x_b\in (\frac{i_b-1}{n},\frac{i_b}{n}]$.
\end{claim}
\begin{claim}
    \phantomsection
    \label{G_u_Ranking_discretized}
For $x_u\in(\frac{i_u-1}{n},\frac{i_u}{n}]$, the bounds for uniformly distributed profiles (\cref{claim:ranking_lower_bound_uniform_profiles}) are translated as: 
\begin{align*}
    &\forall i_u &G_u(i_u)\leq\; & G(i_u,\bot,\bot)\\
    &\forall i_u, i_{\text{start of $v$}} &G_u(i_u)\leq\; &\frac{1}{n+1-i_{\text{start of $v$}}} \sum_{j=i_{\text{start of $v$}}}^n G(i_u,j,\bot)\\
    &\forall i_u, \forall i_{\text{start of $v$}}\leq i_b &G_u(i_u)\leq\; &\frac{1}{i_b+1-i_{\text{start of $v$}}} \sum_{j=i_{\text{start of $v$}}}^{i_b}G(i_u,j,\min\{i_b+1,n\})\\
    &\forall i_u, \forall i_{\text{start of $v$}}\leq i_b &G_u(i_u)\leq\; &\frac{1}{i_b+1-i_{\text{start of $v$}}} \sum_{j=i_{\text{start of $v$}}}^{i_b}G(i_u,j,i_b)
\end{align*}
\end{claim}
\begin{proof}
    We are only considering intervals that are multiples of $\frac{1}{n}$ because for $x_v$ uniformly distributed in $(a,b)$ for some $a\in (\frac{i-1}{n},\frac{i}{n}]$ and $b\in (\frac{j-1}{n},\frac{j}{n}]$, it sums up a sequence of $G$ function values $\{G_k\}$ (where $G_k=G(i_u,i_k,\bot)$ or $G(i_u,i_k,i_b)$) with weights $w_k$ for $k\in \{i,...,j\}$. All terms $w_k$ in the middle $k\in \{i+1,...,j-1\}$ attains the same weight $\frac{1}{n}$, the bound in turn, has the form
$$\frac{w_iG_i+(\frac{1}{n}\sum_{k=i+1}^{j-1} G_k )+w_jG_j}{w_i+\frac{j-1-i}{n}+w_j}.$$
Fixing all $G_k$, the bound is minimized when $w_i$ and $w_j$ take their largest or smallest possible values in $[0,\frac{1}{n}]$, yielding four possible choices. The four choices will result in a uniform distribution in one of the four possible intervals: $[\frac{i-1}{n},\frac{j}{n}]$, $[\frac{i-1}{n},\frac{j-1}{n}]$, $[\frac{i}{n},\frac{j-1}{n}]$ and $[\frac{i}{n},\frac{j}{n}]$.
 
For each fixed $i_{\text{start of $v$}}$, we assume that the profiles are uniformly distributed in 
$(\frac{i_{\text{start of $v$}}-1}{n}, 1]$ for profiles $(x_u,x_v,\bot)$ and $(\frac{i_{\text{start of $v$}}-1}{n}, \frac{i_b}{n}]$ for profiles $(x_u,x_v,x_b)$. We take $i_b$ as $\min\{i_b+1,n\}$ and $i_b$ corresponds to $x_b$ lying on the left and right ends $(\frac{i_b}{n})^+$ and $\frac{i_b}{n}$, respectively.
\end{proof}

\begin{claim}
    For $x_u\in(\frac{i_{u}-1}{n},\frac{i_u}{n}]$, we have 
    $$G_u(i_u)\leq \ev\Big[\text{ gain$(u)$ + gain$(u^*)$}\;|\;x_u\in(\frac{i_{u}-1}{n},\frac{i_u}{n}]\;\Big].$$
\end{claim}
With $G_u$ defined, the final expected value is taken over uniformly distributed $x_u$. 
\begin{claim}
    \phantomsection
    \label{clm:ranking_discretization}
    The approximation ratio for \ranking{} on general graphs is lower bounded by
    $$\frac{1}{n}\sum_{i_u=1}^n G_u(i_u).$$
\end{claim}
This directly follows from the fact that $G_u(i_u)$ is a lower bounding function for the expected gains conditioned on the event $x_u\in(\frac{i_u-1}{n},\frac{i_u}{n}]$.

\subsection{Discretized \Franking{} LP}
\label{subsec:discre-Franking_LP}
Fix an arbitrary integer $n>0$. The piecewise gain function $g$ and compensation function $h$ for the \Franking{} LP can be described by variables $g(i)$ and $h(k)$ where $1\leq i\leq n$ and $0\leq k\leq n$. For $x_u\in (\frac{i-1}{n},\frac{i}{n}]$, $g(x_u)=g(i)$ and $h(x_u)=h(i)$. Again, we save $h(0)$ for the degenerate case $\theta_0=0$. The function constraints translate to
\begin{align*}
    &1.\quad\forall i<n\; && g(i) \le g(i+1),\\
    &2.\quad\forall i<n && h(i) \le h(i+1),\\
    &3.&& h(0)=0\\
    &4.\quad\forall i && 1 - g(i) - h(i) \ge h(n),\\
    &5.\quad\forall i && g(i)  \ge h(n).
\end{align*}
Now we translate the $GF$ functions for all $6$ cases. \textbf{Again, unless otherwise specified, we assume all $\bm{i}$ take values in $\bm{\{1,...,n\}}$ except $\bm{i_{\theta_0}}$, which may take value $\bm{0}$}. We will also use non-superscript $i_v,i_b$ as numbers from $1$ to $n$; while at the same time, we use $i_v^{A/P},i_b^{A/P}$ as inputs to $GF$ to let $GF$ know which profile to consider.

\paragraph{Case $(x_u,\bot,\bot)$ (\cref{franking_lower_bound_botbot}):} For $x_u\in (\frac{i_u-1}{n},\frac{i_u}{n}]$
$$\forall i_u,\quad\quad GF(i_u,\bot,\bot)\leq \frac{1}{n}\sum_{k=1}^n g(k),$$
which is simply integrating the step function $g$ from $0$ to $1$.

\paragraph{Case $(x_u,x_v^P,x_b^P)$ (\cref{franking_lower_bound_PP}):} $GF$ in this case is independent of $x_v^P,x_b^P$, so for $x_u\in (\frac{i_u-1}{n},\frac{i_u}{n}]$, we have bounds:
$$\forall i_u,\quad\quad GF(i_u,i_v^P,i_b^P)\leq g(i_u).$$

\paragraph{Case $(x_u,x_v^P,\bot)$ (\cref{franking_lower_bound_Pbot}):} The bound translates to:
$$\forall i_u,i_{\theta_0}, \quad\quad GF(i_u,i_v^P,\bot)\leq \frac{1}{n}\left[\left(\sum_{k=1}^{i_{\theta_0}}g(k)\right)+(n-i_{\theta_0})\cdot (h(i_{\theta_0})+g(i_u))\right]$$
For $\theta_0 \in (\frac{i_{\theta_0}-1}{n}, \frac{i_{\theta_0}}{n}]$, increasing it by $\epsilon$ results in a change of $\epsilon \cdot \big(g(i_{\theta_0}) - h(i_{\theta_0}) - g(i_u)\big)$ amount of gain. Depending on whether this term is positive or not, the minimum bound is obtained at one of the ends of the interval. Hence, the system of bounds would consist of $\theta_0 = \frac{i_{\theta_0}}{n}$ and $\theta_0 = (\frac{i_{\theta_0}}{n})^+$. Notice that for each choice of $\frac{i_{\theta_0}}{n}$ or $(\frac{i_{\theta_0}}{n})^+$, the only difference in the RHS of the equation is in the value of $h(\theta_0)$: one takes $h(i_{\theta_0})$, the other takes $h(i_{\theta_0}+1)$. Since $h$ is non-decreasing, the case $\theta_0 = \frac{i_{\theta_0}}{n}$, where $h(\theta_0)$ takes the value $h(i_{\theta_0})$, gives a stricter bound. That is, assuming $\theta_0$ always lies at the right-hand side of the interval is sufficient for the whole system of bounds.

\paragraph{Case $(x_u,x_v^P,x_b^A)$ (\cref{franking_lower_bound_PA}):} The bounds translates to:
\begin{align*}
    \forall i_u, i_b, i_{\theta_0}, &&GF(i_u,i_v^P,i_b^A)\leq\; &\frac{1}{n}\left[ \left(\sum_{k=1}^{i_{\theta_0}} g(k)\right) +\max\{i_b-i_{\theta_0}-1,0\}\cdot h(i_{\theta_0})\right.\\
    &&&\left.\vphantom{\int}\quad\quad+ i_{\theta_0}\cdot (1-g(i_b)-h(i_b))+(n-i_{\theta_0})\cdot g(i_u)\right]
\end{align*}
We choose $x_b=(\frac{i_b-1}{n})^+$ as decrease in $x_b$ value results in an overall decrease in the RHS (with decreasing $x_b$, the term $\max\{x_b-\theta_0,0\}$ decreases while the other terms are kept the same). We choose $\theta_0=\frac{i_{\theta_0}}{n}$ for the worst-case for the same reason as the previous case.

\paragraph{Case $(x_u,x_v^A,\bot)$ (\cref{franking_lower_bound_Abot}):} $\forall i_u, i_v, i_{\theta_0}, i_{\theta_1}\; s.t.\; i_{\theta_0}< i_{\theta_1}\wedge i_{v} \leq i_{\theta_1}$, the bound translates to:
\begin{align*}
        GF(i_u,i_v^A,\bot)\leq\; &\frac{1}{n}\left[ \left(\sum_{k=1}^{i_{\theta_1}} g(k)\right) +(n-i_{\theta_1})\cdot h(i_{\theta_0})\right.\\
    &\left.\vphantom{\int}\quad\quad+ i_{\theta_0}\cdot h(i_v)+(i_{\theta_1}-i_{\theta_0})\cdot g(i_u)+(n-i_{\theta_1})\cdot (1-g(i_v)-h(i_v))\right]& (1)\\
    GF(i_u,i_v^A,\bot)\leq\; &\frac{1}{n}\left[ \left(\sum_{k=1}^{i_{\theta_1}} g(k)\right) +(n-i_{\theta_1})\cdot h(i_{\theta_0})\right.\\
    &\left.\vphantom{\int}\quad\quad+ i_{\theta_0}\cdot h(i_v)+(n-i_{\theta_0})\cdot (1-g(i_v)-h(i_v))\right]& (2)\\
      GF(i_u,i_v^A,\bot)\leq\; &\frac{1}{n}\left[ \left(\sum_{k=1}^{i_{\theta_1}-1} g(k)\right) +(n-i_{\theta_1}+1)\cdot h(i_{\theta_0})\right.\\
    &\left.\vphantom{\int}\quad\quad+ i_{\theta_0}\cdot h(i_v)+(i_{\theta_1}-i_{\theta_0}-1)\cdot g(i_u)+(n-i_{\theta_1}+1)\cdot (1-g(i_v)-h(i_v))\right]& (3)\\    
    GF(i_u,i_v^A,\bot)\leq\; &\frac{1}{n}\left[ \left(\sum_{k=1}^{i_{\theta_1}-1} g(k)\right) +(n-i_{\theta_1}+1)\cdot h(i_{\theta_0})\right.\\
    &\left.\vphantom{\int}\quad\quad+ i_{\theta_0}\cdot h(i_v)+(n-i_{\theta_0})\cdot (1-g(i_v)-h(i_v))\right]& (4)
\end{align*}
 $\forall i_u, i_v, i_{\theta_0}, i_{\theta_1}\; s.t.\; i_{\theta_0}= i_{\theta_1}\wedge i_{v} \leq i_{\theta_1}$, the bound translates to:
 \begin{align*}
    GF(i_u,i_v^A,\bot)\leq\; &\frac{1}{n}\left[ \left(\sum_{k=1}^{i_{\theta_1}} g(k)\right) +(n-i_{\theta_1})\cdot h(i_{\theta_0})\right. & (1)\\
    &\left.\vphantom{\int}\quad\quad+ i_{\theta_0}\cdot h(i_v)+(i_{\theta_1}-i_{\theta_0})\cdot g(i_u)+(n-i_{\theta_1})\cdot (1-g(i_v)-h(i_v))\right]\\
    GF(i_u,i_v^A,\bot)\leq\; &\frac{1}{n}\left[ \left(\sum_{k=1}^{i_{\theta_1}} g(k)\right) +(n-i_{\theta_1})\cdot h(i_{\theta_0})\right.\\
    &\left.\vphantom{\int}\quad\quad+ i_{\theta_0}\cdot h(i_v)+(n-i_{\theta_0})\cdot (1-g(i_v)-h(i_v))\right]& (2)
\end{align*}
Notice inequalities $(3),(4)$ in the $i_{\theta_0}<i_{\theta_1}$ case do not extend to the $i_{\theta_0}=i_{\theta_1}$ case. This difference arises because when $i_{\theta_0}=i_{\theta_1}$, assuming $\theta_0=\frac{i_{\theta_0}}{n}$ (i.e. $\theta_0$ lying in the right-hand side of the interval), forces $\theta_1$ to also lie in the right-hand side of the interval, eliminating the case $\theta_0=\frac{i_{\theta_0}}{n}$ while $\theta_1=(\frac{i_{\theta_1}-1}{n})^+$ which is a possible scenario when $i_{\theta_0}<i_{\theta_1}$.

In case $(x_u,x_v^A,\bot)$, fixing every other parameter except $\theta_0$, increasing $\theta_0$ within the interval $(\frac{i_{\theta_0}-1}{n},\min\{\frac{i_{\theta_0}}{n},\theta_1\}]$ will result in an overall decrease in the equation value. This is because every other term is unchanged but we increased the range of $u$ being unmatched, that is $(0,\theta_0]$. Because we assumed being unmatched is the worst-case, increasing this range will result in a worse lower bound.

Now assuming $\theta_0$ always attains the right-hand side of the interval, i.e. $\theta_0=\min\{\frac{i_{\theta_0}}{n},\theta_1\}$, by the same argument as before, we can obtain the worst-case scenario by setting $\theta_1$ on either the left-hand side of the interval ($\theta_1=(\frac{i_{\theta_1}-1}{n})^+$) or on the right-hand side of the interval ($\theta_1=\frac{i_{\theta_1}}{n}$)\footnote{Notice that the case $\theta_0$ and $\theta_1$ both attaining $(\frac{i_{\theta_1}-1}{n})^+$ is not included in the inequality system, it is easy to see that this case is bounded below by $\theta_0=\frac{i_{\theta_1}-1}{n}$ and $\theta_1=(\frac{i_{\theta_1}-1}{n})^+$, as the intervals are not changed but we lowered the compensation $h(\theta_0)$ from $h(i_{\theta_1})$ to $h(i_{\theta_1}-1)$.}. Inequalities $(1),(2)$ correspond to taking both $\theta_0$ and $\theta_1$ on the right-hand side of their respective intervals; whereas inequalities $(3),(4)$ are the cases where $\theta_0$ lies on the right-hand side while $\theta_1$ lies on the left-hand side. For a fixed choice of $\theta_0,\theta_1$, we have two inequalities. One is taking $g(i_u)$ as the worse-case and one is taking $1-g(i_v)-h(i_v)$ as the worse-case value for the term $\min\{g(x_u),1-g(x_v)-h(x_v)\}$.

\paragraph{Case $(x_u,x_v^A,x_b^A)$ (\cref{franking_lower_bound_AA}):} $\forall i_u, i_v,i_b, i_{\theta_0}, i_{\theta_1}\; s.t.\; i_{\theta_0}< i_{\theta_1}\wedge i_{v} \leq i_{\theta_1}$, the bound translates to:
\begin{align*}
    GF(i_u,i_v^A,i_b^A)\leq\; &\frac{1}{n}\left[ \left(\sum_{k=1}^{i_{\theta_1}} g(k)\right) +\max\{i_b-i_{\theta_1}-1,0\}\cdot h(i_{\theta_0})\right.\\
    &\quad\quad\;+ i_{\theta_0}\cdot (1-g(i_b)-h(i_b)) + (i_{\theta_1}-i_{\theta_0})\cdot (1-g(i_v)-h(i_v))& (1)\\
    &\left.\vphantom{\int}\quad\quad+(n-i_{\theta_1})\cdot (1-g(i_v)-h(i_v))\right]\\
        GF(i_u,i_v^A,i_b^A)\leq\; &\frac{1}{n}\left[ \left(\sum_{k=1}^{i_{\theta_1}} g(k)\right) +\max\{i_b-i_{\theta_1}-1,0\}\cdot h(i_{\theta_0})\right.\\
    &\quad\quad\;+ i_{\theta_0}\cdot (1-g(i_b)-h(i_b)) + (i_{\theta_1}-i_{\theta_0})\cdot g(i_u)& (2)\\
    &\left.\vphantom{\int}\quad\quad+(n-i_{\theta_1})\cdot (1-g(i_v)-h(i_v))\right]\\
        GF(i_u,i_v^A,i_b^A)\leq\; &\frac{1}{n}\left[ \left(\sum_{k=1}^{i_{\theta_1}} g(k)\right) +\max\{i_b-i_{\theta_1}-1,0\}\cdot h(i_{\theta_0})\right.\\
    &\quad\quad\;+ i_{\theta_0}\cdot g(i_u) + (i_{\theta_1}-i_{\theta_0})\cdot g(i_u)& (3)\\
    &\left.\vphantom{\int}\quad\quad+(n-i_{\theta_1})\cdot (1-g(i_v)-h(i_v))\right]\\
        GF(i_u,i_v^A,i_b^A)\leq\; &\frac{1}{n}\left[ \left(\sum_{k=1}^{i_{\theta_1}-1} g(k)\right) +\max\{i_b-i_{\theta_1},0\}\cdot h(i_{\theta_0})\right.\\
    &\quad\quad\;+ i_{\theta_0}\cdot (1-g(i_b)-h(i_b)) + (i_{\theta_1}-i_{\theta_0}-1)\cdot (1-g(i_v)-h(i_v))& (4)\\
    &\left.\vphantom{\int}\quad\quad+(n-i_{\theta_1}+1)\cdot (1-g(i_v)-h(i_v))\right]\\
        GF(i_u,i_v^A,i_b^A)\leq\; &\frac{1}{n}\left[ \left(\sum_{k=1}^{i_{\theta_1}-1} g(k)\right) +\max\{i_b-i_{\theta_1},0\}\cdot h(i_{\theta_0})\right.\\
    &\quad\quad\;+ i_{\theta_0}\cdot (1-g(i_b)-h(i_b)) + (i_{\theta_1}-i_{\theta_0}-1)\cdot g(i_u)& (5)\\
    &\left.\vphantom{\int}\quad\quad+(n-i_{\theta_1}+1)\cdot (1-g(i_v)-h(i_v))\right]\\
        GF(i_u,i_v^A,i_b^A)\leq\; &\frac{1}{n}\left[ \left(\sum_{k=1}^{i_{\theta_1}-1} g(k)\right) +\max\{i_b-i_{\theta_1},0\}\cdot h(i_{\theta_0})\right.\\
    &\quad\quad\;+ i_{\theta_0}\cdot g(i_u) + (i_{\theta_1}-i_{\theta_0}-1)\cdot g(i_u)& (6)\\
    &\left.\vphantom{\int}\quad\quad+(n-i_{\theta_1}+1)\cdot (1-g(i_v)-h(i_v))\right]\\
\end{align*}
Whereas $\forall i_u, i_v,i_b, i_{\theta_0}, i_{\theta_1}\; s.t.\; i_{\theta_0}= i_{\theta_1}\wedge i_{v} \leq i_{\theta_1}$, we have 
\begin{align*}
    GF(i_u,i_v^A,i_b^A)\leq\; &\frac{1}{n}\left[ \left(\sum_{k=1}^{i_{\theta_1}} g(k)\right) +\max\{i_b-i_{\theta_1}-1,0\}\cdot h(i_{\theta_0})\right.\\
    &\quad\quad\;+ i_{\theta_0}\cdot (1-g(i_b)-h(i_b)) + (i_{\theta_1}-i_{\theta_0})\cdot (1-g(i_v)-h(i_v))& (1)\\
    &\left.\vphantom{\int}\quad\quad+(n-i_{\theta_1})\cdot (1-g(i_v)-h(i_v))\right]\\
        GF(i_u,i_v^A,i_b^A)\leq\; &\frac{1}{n}\left[ \left(\sum_{k=1}^{i_{\theta_1}} g(k)\right) +\max\{i_b-i_{\theta_1}-1,0\}\cdot h(i_{\theta_0})\right.\\
    &\quad\quad\;+ i_{\theta_0}\cdot (1-g(i_b)-h(i_b)) + (i_{\theta_1}-i_{\theta_0})\cdot g(i_u)& (2)\\
    &\left.\vphantom{\int}\quad\quad+(n-i_{\theta_1})\cdot (1-g(i_v)-h(i_v))\right]\\
        GF(i_u,i_v^A,i_b^A)\leq\; &\frac{1}{n}\left[ \left(\sum_{k=1}^{i_{\theta_1}} g(k)\right) +\max\{i_b-i_{\theta_1}-1,0\}\cdot h(i_{\theta_0})\right.\\
    &\quad\quad\;+ i_{\theta_0}\cdot g(i_u) + (i_{\theta_1}-i_{\theta_0})\cdot g(i_u)& (3)\\
    &\left.\vphantom{\int}\quad\quad+(n-i_{\theta_1})\cdot (1-g(i_v)-h(i_v))\right]
\end{align*}
The same reasoning for the choices of $\theta_0$ and $\theta_1$ applies here as well. The worst-case gain for $u$ is in the range $(0,\theta_0)$ where it gets $\min\{g(i_u),1-g(i_b)-h(i_b)\}$ amount of gain, so the minimum is always achieved at $\theta_0$ taking a value as large as possible. Hence we assume $\theta_0=\min\{\frac{i_{\theta_0}}{n},\theta_1\}$ always. We are not sure about $\theta_1$ so we take both ends as possible choices. When $i_{\theta_0}=i_{\theta_1}$, maximizing $\theta_0$ forces $\theta_1$ to take the right-hand side value as well. Inequalities $(1),(2),(3)$ correspond to the cases where $\theta_1$ lies on the right side of the interval $(\frac{i_{\theta_1}-1}{n},\frac{i_{\theta_1}}{n}]$. Inequalities $(4),(5),(6)$ correspond to the cases where $\theta_1$ lies on the left side. 

We also have to consider the worst choices among $\min\{g(i_u),1-g(i_b)-h(i_b)\}$ and $\min\{g(i_u),1-g(i_v)-h(i_v)\}$ for the intervals $(0,\theta_0)$ and $(\theta_0,\theta_1)$, respectively. 
\begin{itemize}
    \item Case $(1),(4)$ corresponds to choosing the pairs $1-g(i_b)-h(i_b)$ and $1-g(i_v)-h(i_v)$.
    \item  Case $(2),(5)$ corresponds to choosing the pairs  $1-g(i_b)-h(i_b)$ and $ g(i_u)$.
    \item  Case $(3),(6)$ corresponds to choosing the pairs $g(i_u)$ and $g(i_u)$.
\end{itemize}
 We do not need to consider the combination $g(i_u)$ and $1-g(i_v)-h(i_v)$ because $g(i_u)<1-g(i_b)-h(i_b)$ implies $g(i_u)<1-g(i_v)-h(i_v)$ so such a combination is not possible.

With constraints for all $6$ cases defined, we classify them into two groups as in (\cref{claim:Franking_lower_bound_uniform_profiles_A,claim:Franking_lower_bound_uniform_profiles_P}):
\begin{align*}
    &\forall i_u, &GF^P(i_u)\leq&\; GF(i_u,i_v^P,\bot)\\
    &\forall i_u, &GF^P(i_u)\leq&\; GF(i_u,i_v^P,i_b^P)\\
    &\forall i_u,i_b^A &GF^P(i_u)\leq&\; GF(i_u,i_v^P,i_b^A)\\
    &\forall i_u, &GF^A(i_u)\leq&\; GF(i_u,\bot,\bot)\\
     &\forall i_u,i_{\text{start of $v$}} &GF^A(i_u)\leq&\;\frac{1}{n+1-i_{\text{start of $v$}}} \sum_{i_v=i_{\text{start of $v$}}}^n GF(i_u,i_v^A,\bot)\\
    &\forall i_u,i_b,i_{\text{start of $v$}} &GF^A(i_u)\leq&\;\frac{1}{i_b+1-i_{\text{start of $v$}}} \sum_{i_v=i_{\text{start of $v$}}}^{i_b} GF(i_u,i_v^A,i_b^A)\\
    &\forall i_u,i_b,i_{\text{start of $v$}} &GF^A(i_u)\leq&\;\frac{1}{i_b+1-i_{\text{start of $v$}}} \sum_{i_v=i_{\text{start of $v$}}}^{i_b} GF(i_u,i_v^A,\min\{i_b+1,n\}^A)\\
\end{align*}
Using the same reasoning as in the \ranking{} case, when considering $x_v^A$ ranging over intervals, we can consider intervals only of the type $(\frac{i}{n},\frac{j}{n}]$. The choices of $i_b^A$ and $\min\{i_b+1,n\}^A$ correspond to the cases $x_b=\frac{i_b}{n}$ and $x_b=(\frac{i_b}{n})^+$.

Finally, we sum over all possible cases of $i_{\theta_1^u}$ and obtain the approximation ratio of \Franking{} by
\begin{claim}
    \phantomsection
    \label{clm:Franking_discretization}
    The approximation ratio of \Franking{} on a general graph is lower bounded by
    $$\min_{0\leq i_{\theta_1^u} \leq n} \frac{1}{n}\left[\sum_{i_u=1}^{i_{\theta_1^u}} GF^P(i_u)+\sum_{i_u=i_{\theta_1^u}+1}^n GF^A(i_u)\right]$$
\end{claim}
\section{Numerical Results}
\phantomsection
\label{sec:numerical_results}
We present the results of \nameref{Discretized_Tightened_Ranking_LP} for $n\leq 80$ and \nameref{subsec:discre-Franking_LP} for $n\leq 45$ in \Cref{tab:fn_rectangular1,tab:fn_rectangular2}. We do not individually present the simple Ranking LP results, but note that running the simple Ranking LP as in \cref{subsec:discre-ranking_LP} gives $0.559$ at $n=80$. 

All code was written in Python (version 3.12.3) and are available through GitHub\footnote{\href{https://github.com/yutaoyt7/STOC_2026_LP_Codes}{https://github.com/yutaoyt7/STOC\_2026\_LP\_Codes}}.  To solve the LP instances, we used Gurobi Optimization package (version 12.0.3). The experiments were conducted on a computer cluster with 64 physical cores, with CPU model Intel(R) Xeon(R) Platinum 8375C CPU @ 2.90GHz, instruction set [SSE2|AVX|AVX2|AVX512], and 1TB of RAM. The operating system was Ubuntu 24.04.3 LTS. Running the tightened \ranking{} LP with $n=80$ took $1.5$ hours and running \Franking{} with $n=45$ took about $6$ hours.

\begin{table}[htbp]
\centering
\caption{Tightened \ranking{} LP Results}
\begin{tabular}{|c|c|c|c|c|c|c|c|}
\hline
$n$ & $\text{ LP Results}$ & $n$ & $\text{ LP Results}$ & $n$ & $\text{ LP Results}$ & $n$ & $\text{ LP Results}$ \\
\hline
1  & 0.39999        & 8  & 0.54429 & 15 & 0.55297 & 30 & 0.55741 \\
\hline
2  & 0.48263        & 9  & 0.54639 & 16 & 0.55356 & 35 & 0.55801 \\
\hline
3  & 0.51391    & 10  & 0.54804 & 17 &  0.55406  & 40 & 0.55846 \\
\hline
4  & 0.52480    & 11 & 0.54947 & 18 &  0.55450 & 50 & 0.55909 \\
\hline
5  & 0.53247    & 12 & 0.55060 & 19 & 0.55490 & 60 & 0.55950 \\
\hline
6  & 0.53783    & 13 & 0.55152 & 20 & 0.55526 & 70 & 0.55979 \\
\hline
7  & 0.54140    & 14 & 0.55229 & 25 & 0.55657 & 80 & 0.56001 \\
\hline
\end{tabular}
\label{tab:fn_rectangular1}
\end{table}

\begin{table}[htbp]
\centering
\caption{\Franking{} LP Results}
\begin{tabular}{|c|c|c|c|c|c|c|c|}
\hline
$n$ & $\text{ LP Results}$ & $n$ & $\text{ LP Results}$ & $n$ & $\text{ LP Results}$ & $n$ & $\text{ LP Results}$ \\
\hline
1  & 0.5        & 6  & 0.52125 & 12 & 0.53102 & 25 & 0.53654 \\
\hline
2  & 0.5        & 7  & 0.52338 & 14 & 0.53248 & 30 & 0.53745 \\
\hline
3  & 0.50555    & 8  & 0.52600 & 16 &  0.53372  & 35 & 0.53813 \\
\hline
4  & 0.51153    & 9 & 0.52767 & 18 &  0.53448 & 40 & 0.53861 \\
\hline
5  & 0.51793    & 10 & 0.52880 & 20 & 0.53524 & 45 & 0.53900 \\
\hline
\end{tabular}
\label{tab:fn_rectangular2}
\end{table}

\section{Main Theorems}
\RankingGeneral*
\begin{proof}
    By \cref{G_u_Ranking_discretized_tight}, running \nameref{Discretized_Tightened_Ranking_LP} with $n$ set to any number gives a lower bound for the approximation ratio of \ranking{}. We solved the $n=80$ case and obtained $0.56001$ as a lower bound.
\end{proof}

\FRankingGeneral*
\begin{proof}
    By \cref{clm:Franking_discretization}, running \nameref{subsec:discre-Franking_LP} with $n$ set to any number gives a lower bound for the approximation ratio of \Franking{}. We solved the $n=45$ case and obtained $0.53900$ as a lower bound.
\end{proof}

\section{Ranking on Large Odd Girth Graphs}
\phantomsection
\label{sec:odd_girth_graphs}
In \cref{lem:victim_querycommit_case_u} and \cref{lem:victim_querycommit_case_u_star}, we proved that the number of blockers $u$ and $u^*$ have when they are victimized is positively correlated with the length of an even-length alternating path $P$ from $u^*$ to $u$. Since we assumed that $(u,u^*)\in E$, $P\cup\{u,u^*\}$ constitutes an odd-length cycle. The length of this cycle is bounded below by the \emph{odd girth} of the graph. For graphs with larger odd girths, we can hence collect more compensations for victim vertices, and prove better lower bounds for the approximation ratio of \ranking{}.

\begin{definition}
    The odd girth of a simple general graph $G$ is the length of the shortest odd-length cycle of the graph. Denote $\mathcal{G}_{2k+1}$ as the class of simple general graphs of odd girth $\geq 2k+1$. 
\end{definition}
Notice that the general graphs we are considering without restriction are of class $\mathcal{G}_3$. For graphs in $\mathcal{G}_{2k+1}$, we are able to collect at least $k$ copies of compensations for $u$ and $u^*$ when they become victims, as each of them will have at least $k$ blockers. We will modify the \ranking{} LP for each $\mathcal{G}_{2k+1}$ by taking the extra compensations into account.
\subsection{Continuous \ranking{} LP for $\mathcal{G}_{2k+1}$}
We first modify the function constraints for $h$ in \cref{def:function_constraint_ranking} to account for increasing copies of compensation collected. 
\begin{claim}[Assume Unmatched when Possible for $\mathcal{G}_{2k+1}$]
\phantomsection
\label{function_constraints_odd_girth_graph}
Let $g,h:[0,1]\times [0,1]\to [0,1]$ be gain and compensation functions that satisfy \cref{def:function_constraint_ranking} constraints 1,2, and 3. For graph $G$ of class $\mathcal{G}_{2k+1}$, where $k>1$, we instead enforce the following as constraints $4$ and $5$:
\begin{enumerate}[start=4]
    \item $\forall x,y,\; 1-g(x,y)-h(x,y)\geq k\cdot h(0^+,1).$
    \item $\forall x,y,\; g(x,y)-h(y,x)\geq k\cdot h(0^+,1).$
\end{enumerate}
As a result, when a vertex $u$ could be either matched or unmatched, we can always calculate a lower bound for $gain(u)$ assuming $u$ is unmatched if the collected compensation is $\leq k$ copies.
\end{claim}
The claim holds by the same reasoning as \cref{fact:assume_unmatched_worst_case}. Now we are ready to modify the \ranking{} LP. We will add the extra compensation to the system of simple \ranking{} LP deduced in \cref{sec:ranking} instead of the tightened \ranking{} LP from \cref{sec:tigher_bounds}, as integration with the former is more intuitive and straightforward. 

\begin{lemma}[Extra Compensation for Class $(x_u,x_v,\bot)$]
\phantomsection
\label{G_no_backup_ranking_odd_girth}
 Fix $G\in \mathcal{G}_{2k+1}$ where $k>1$. Assume $g,h$ satisfy the above function constraints \cref{function_constraints_odd_girth_graph}. Let $\vecx_{\shortminus u^*}$ be a rank vector with profile $(x_u,x_v,\bot)$. Let $G(x_u,x_v,\bot,\theta_0)$ be the corresponding lower bound for this rank vector with fixed $\theta_0<x_u$ as in \cref{G_ranking_nobackup}. Then the expected total gain of $u$ and $u^*$ is lower bounded by
   \begin{align*}
       G_k(x_u,x_v,\bot,\theta_0)=\;&G(x_u,x_v,\bot,\theta_0)\\
       &+ \theta_0\cdot \max\{k-2,0\}\cdot h(x_u,0^+)\\
       &+ (1-\max\{\theta_0,x_v\})\cdot \max\{k-2,0\}\cdot h(x_v,\theta_0).
   \end{align*}
   And we have that
   $$\ev[\;\text{gain}(u)+\text{gain}(u^*)\;|\; \text{profile of $\vecx_{\shortminus u^*}=(x_u,x_v,\bot)$}]$$
   is lower bounded by:
    $$\inf_{\theta_0<x_u}\{G_k(x_u,x_v,\bot,\theta_0)\}.$$
\end{lemma}
\begin{proof}
    Fix $G\in \mathcal{G}_{2k+1},\;x_u,\;x_v,\;\theta_0$ as stated. Let $x\in[0,1]$ be the random rank of $u^*$. 
    
    \noindent\textbf{Extra gain for $u$:} When $x\in[0,\theta_0)$, assume $u$ is unmatched by $u^*$ and has alternating paths $\M(\vecx_{\shortminus u^*}(x))\oplus \M(\vecx_{\shortminus u^*})$ of the form $P=(u^*,w,u_2,...,u_{l-2},v,u)$. For $G\in \mathcal{G}_{2k+1}$, the alternating path will have length $\geq 2k$. Otherwise, we will have an odd length cycle $P\cup\{(u,u^*)\}$ of length $<2k+1$, which is a contradiction. By \cref{lem:victim_querycommit_case_u}, $u^*,u_2,...,u_{l-2}$ are all blockers of $u$. In $G(x_u,x_v,\bot,\theta_0)$, we have taken into account the compensation from $u^*$ and $u_{l-2}$, and we are free to add the compensations from $u_2,...,u_{l-4}$, which is at least a set of $k-2$ vertices. Each of them, being an even-indexed vertex in the path, will have rank $\leq x_u$ by the monotonicity property along the alternating path \cref{fact:monotonicity_for_alt_path_ranking}. W.O.L.G., we could assume their match has rank $>0$, again as the event of any vertex having rank $=0$ has $0$ probability measure. Therefore, each of the $k-2$ vertices pays at least $h(x_u,0^+)$ amount of compensation to $u$. This constitutes the extra gain part: 
    $$ \theta_0\cdot \max\{k-2,0\}\cdot h(x_u,0^+).$$
    Now, we are assuming $u$ receives $ k-2+2=k$ copies of compensations, by \cref{function_constraints_odd_girth_graph}, this is a valid lower bound even for $gain(u)$ when it is matched in the range $x\in[0,\theta_0)$.

     \noindent\textbf{Extra gain for $u^*$:} When $x\in(\max\{x_v,\theta_0\},1]$, we assume $u^*$ is unmatched in $\M(\vecx_{\shortminus u^*}(x))$. Assume $\theta_0>0$, by \cref{fact:theta_0_II}, there exists a left neighborhood $(\theta_0^-,\theta_0)$ of $\theta_0$ such that for any $x_{\theta_0}\in (\theta_0^-,\theta_0)$, $u$ is unmatched in $\M(\vecx_{\shortminus u^*}(x_{\theta_0}))$ and the third vertex $u_2$ in the alternating path $\M(\vecx_{\shortminus u^*}(x_{\theta_0}))\oplus \M(\vecx_{\shortminus u^*})$ has rank $\theta_0$.
    Fix such $x_{\theta_0}$ and let $P=(u^*,w,u_2,...,u_{l-2},v,u)$ be the alternating path $\M(\vecx_{\shortminus u^*}(x_{\theta_0}))\oplus \M(\vecx_{\shortminus u^*})$. By \cref{lem:victim_querycommit_case_u_star}, $u^*$ is the victim of $w,u_3,...u_{l-3},v$ when it is unmatched in $\M(\vecx_{\shortminus u^*}(x))$. In $G(x_u,x_v,\bot,\theta_0)$, we have already taken the compensations from $w$ and $v$ into account. By similar reasoning as before, there are at least $k-2$ extra vertices $u_3,...,u_{l-3}$ we could collect compensations from. Each of them, being an odd-indexed vertex in the alternating path, has rank $\leq x_v$ by the monotonicity property along the alternating path \cref{fact:monotonicity_for_alt_path_ranking}. Their matches in $\M(\vecx_{\shortminus u^*}(x))$, which are $u_4,...,u_{l-2}$, has rank $\geq x_{u_2}=\theta_0$. So each of these vertices will pay at least $h(x_v,\theta_0)$ amount of compensation to $u^*$. This constitutes the extra gain part 
    $$ (1-\max\{\theta_0,x_v\})\cdot \max\{k-2,0\}\cdot h(x_v,\theta_0).$$
    When $\theta_0=0$, this part is $0$ as $h(x_v,0)=0$ by the function constraints, so we are not over-approximating for the degenerate case $\theta_0=0$.
    
    Note that here we do not consider $\theta_0=x_u$. In the system of bounds from \cref{G_ranking_nobackup}, $\theta_0=x_u$ is reserved for the case where the alternating path $\M(\vecx_{\shortminus u^*}(\theta_0^-))\oplus \M(\vecx_{\shortminus u^*})$ has length exactly $2$. Since we assumed $G\in \mathcal{G}_{2k+1}$ with $k> 1$, the graph will be triangle-free and hence the bound from \cref{G_ranking_nobackup} for case $\theta_0=x_u$ does not apply here. It is still possible that $\theta_0=x_u$ with the alternating path $\M(\vecx_{\shortminus u^*}(\theta_0^-))\oplus \M(\vecx_{\shortminus u^*})$ having length $\geq 2k$. Such a case can be approximated by
    $$\lim_{\theta_0\to x_u^-} G_k(x_u,x_v,\bot,\theta_0)$$
    and hence we are not underestimating the bounds.

\end{proof}
\begin{lemma}[Extra Compensation for Class $(x_u,x_v,x_b)$]
\phantomsection
\label{G_backup_ranking_odd_girth}
 Fix $G\in \mathcal{G}_{2k+1}$ where $k>1$. Assume $g,h$ satisfy the above constraints \cref{function_constraints_odd_girth_graph}. Let $\vecx_{\shortminus u^*}$ be a rank vector with profile $(x_u,x_v,x_b)$. Let $G(x_u,x_v,x_b,\theta_0)$ be the corresponding lower bound for this rank vector with fixed $\theta_0<x_u$ as in \cref{G_ranking_nobackup}. Then the expected total gain of $u$ and $u^*$ is lower bounded by
   \begin{align*}
       G_k(x_u,x_v,x_b,\theta_0)=\;&G(x_u,x_v,x_b,\theta_0)+ \max\{x_b-\max\{\theta_0,x_v\},0\}\cdot \max\{k-2,0\}\cdot h(x_v,\theta_0).
   \end{align*}
   Hence, 
   $$\ev[\;\text{gain}(u)+\text{gain}(u^*)\;|\; \text{profile of $\vecx_{\shortminus u^*}=(x_u,x_v,x_b)$}]$$
   is lower bounded by:
    $$\inf_{\theta_0<x_u}\{G_k(x_u,x_v,x_b,\theta_0)\}.$$
\end{lemma}
\begin{proof}
    We are only collecting extra gain for $u^*$. This is because when $u$ has a backup, $u^*$ can never make $u$ unmatched. The only difference in the extra gain we collected for $u^*$ between this case and \cref{G_no_backup_ranking_odd_girth} is changing $1-\max\{\theta_0,x_v\}$ to $\max\{x_b-\max\{\theta_0,x_v\},0\}$, which corresponds to the range of $x\in (\max\{\theta_0,x_v\},1]$ where we can prove $u^*$ is a victim. When $u$ does not have a backup, $u^*$ will be a victim throughout this interval. When $u$ has a backup at $x_b$, we only know that $u^*$ is a victim for $x\in (\max\{\theta_0,x_v\},x_b)$ (\cref{lem:victim_querycommit_case_u_star}). The rest follows from the same line of reasoning as the no-backup case, \cref{G_no_backup_ranking_odd_girth}.
\end{proof}
\subsection{Discretized \ranking{} LP for $\mathcal{G}_{2k+1}$}
\phantomsection
\label{}
For odd girth graphs $\mathcal{G}_{2k+1}$ with $k>1$, we can formulate a system of discretized \ranking{} LP by incorporating the described modifications in the previous subsection.

Fix $n\geq 1$ (piece of discretization) and $2k+1>3$ (odd girth).  Define $g,h,g_P,g_B$ as piecewise constant functions similar to \nameref{subsec:discre-ranking_LP}. We incorporate the function constraints as usual except that now we set 
\begin{align*}
    &4.\quad\forall i,j\in[n] && 1 - g(i,j) - h(i,j) \ge k\cdot h(1,n).\\
    &5.\quad\forall i,j\in[n] && g(i,j)-h(j,i)  \ge k\cdot h(1,n).
\end{align*}
We keep the bound for profiles $(x_u,\bot,\bot)$ as it is.

For profiles $(x_u,x_v,\bot)$, with $x_u\in(\frac{i_u-1}{n},\frac{i_u}{n}]$ and $x_v\in (\frac{i_v-1}{n},\frac{i_v}{n}]$, we introduce the following constraints for $G_k(i_u,i_v,\bot)$ where $RHS\text{ of } (1),(2)$ expressions are the corresponding $RHS$ expressions in \nameref{subsec:discre-ranking_LP} for profiles $(x_u,x_v,\bot)$:
\begin{align*}
\shortintertext{$\forall i_{\theta_0}\leq i_u\leq i_v$:}
\qquad G_k(i_u,i_v,\bot)\leq  RHS\text{ of } (1)+\frac{1}{n}\Big[ &(n-i_v)\cdot \max\{k-2,0\}\cdot h(i_v,i_{\theta_0})\\
&+ i_{\theta_0}\cdot \max\{k-2,0\}\cdot h(i_u,1) \Big], & (1) \\[0.6em]
\shortintertext{$\forall  i_{\theta_0},i_v\leq i_u$:}
\qquad G_k(i_u,i_v,\bot)\leq  RHS\text{ of } (2)+\frac{1}{n}\Big[ &(n-\max\{i_{\theta_0},i_v-1\})\cdot \max\{k-2,0\}\cdot h(i_v,i_{\theta_0})\\
&+ i_{\theta_0}\cdot \max\{k-2,0\}\cdot h(i_u,1) \Big]. & (2) 
\end{align*}
Similarly, for profiles $(x_u,x_v,x_b)$, with $x_u\in(\frac{i_u-1}{n},\frac{i_u}{n}]$ and $x_v\in (\frac{i_v-1}{n},\frac{i_v}{n}]$,$x_b\in (\frac{i_b-1}{n},\frac{i_b}{n}]$, we introduce the following constraints for $G_k(i_u,i_v,i_b)$ where $RHS\text{ of } (1.1),(1.2),(2)$ expressions are the corresponding $RHS$ expressions in \nameref{subsec:discre-ranking_LP} for profiles $(x_u,x_v,x_b)$:
{\small
\begin{align*}
\shortintertext{$\forall i_{\theta_0}\leq i_u\leq i_v<i_b$:}
\qquad G_k(i_u,i_v,i_b)\leq{}& RHS\text{ of } (1.1)+\frac{1}{n}\Big[ (i_b-i_v-1)\cdot \max\{k-2,0\}\cdot h(i_v,i_{\theta_0})\Big],& (1.1)\\[0.6em]
\shortintertext{$\forall i_{\theta_0}\leq i_u\leq i_v=i_b$:}
\qquad G_k(i_u,i_v,i_b)\leq{}& RHS\text{ of } (1.2), & (1.2)\\[0.6em]
\shortintertext{$\forall i_{\theta_0}, i_v\leq i_u$:}
\qquad G_k(i_u,i_v,i_b)\leq{}& RHS\text{ of } (2)+\frac{1}{n}\Big[ \max\{i_b-1-\max\{i_{\theta_0},i_v-1\},0\}\cdot \max\{k-2,0\}\cdot h(i_v,i_{\theta_0})\Big].& (2)
\end{align*}
}
The remaining aggregation parts are the same as \cref{subsec:discre-ranking_LP}, with $G_k$ replacing $G$.
\begin{claim}[Discretized Odd Girth \ranking{} LP]
\phantomsection
\label{G_u_Ranking_Odd_Girth}
   Define $G^k_u$ for the graph class $\mathcal{G}_{2k+1}$ for $k>1$ as follows:
\begin{align*}
    &\forall i_u ,&G^k_u(i_u)\leq\; & G(i_u,\bot,\bot),\\
    &\forall i_u, i_{\text{start of $v$}} ,&G^k_u(i_u)\leq\; &\frac{1}{n+1-i_{\text{start of $v$}}} \sum_{j=i_{\text{start of $v$}}}^n G_k(i_u,j,\bot),\\
    &\forall i_u, \forall i_{\text{start of $v$}}\leq i_b ,&G^k_u(i_u)\leq\; &\frac{1}{i_b+1-i_{\text{start of $v$}}} \sum_{j=i_{\text{start of $v$}}}^{i_b}G_k(i_u,j,\min\{i_b+1,n\}),\\
    &\forall i_u, \forall i_{\text{start of $v$}}\leq i_b ,&G^k_u(i_u)\leq\; &\frac{1}{i_b+1-i_{\text{start of $v$}}} \sum_{j=i_{\text{start of $v$}}}^{i_b}G_k(i_u,j,i_b).
\end{align*}
Then the approximation ratio for \ranking{} on $\mathcal{G}_{2k+1}$ is lower bounded by
    $$\frac{1}{n}\sum_{i_u=1}^n G^k_u(i_u).$$
We denote the system of discretized \ranking{} LP for $\mathcal{G}_{2k+1}$ as \ranking{} LP$(\mathcal{G}_{2k+1})$
\end{claim}
\subsection{Numerical Results for \ranking{} LP$(\mathcal{G}_{2k+1})$}
We run \ranking{} LP$(\mathcal{G}_{2k+1})$\footnote{Code is available at \href{https://github.com/yutaoyt7/STOC_2026_LP_Codes.}{https://github.com/yutaoyt7/STOC\_2026\_LP\_Codes}} as given in \cref{G_u_Ranking_Odd_Girth} for various $k$ values with setting $n=80$ uniformly and obtained the following results\footnote{$\mathcal{G}_3$ is the class of all general graphs, the result for $\mathcal{G}_3$ is obtained from running the tightened discretized \ranking{} LP in \cref{G_u_Ranking_discretized_tight}.}:
\begin{table}[htbp]
\centering
\caption{Numerical Results of \ranking{} LP $(\mathcal{G}_{2k+1})$}
\begin{tabular}{|c|c|c|c|}
\hline
\text{Graph Class} & $\text{ LP Results}$ & \text{Graph Class} & $\text{ LP Results}$\\
\hline
$\mathcal{G}_{3}$  & 0.56001  &$\mathcal{G}_{13}$  & 0.59071    \\
\hline
$\mathcal{G}_{5}$  & 0.56288  &$\mathcal{G}_{17}$  & 0.59697  \\
\hline
$\mathcal{G}_{7}$ & 0.57023  &$\mathcal{G}_{33}$  & 0.60693  \\
\hline
$\mathcal{G}_{9}$  & 0.57911   &$\mathcal{G}_{65}$  & 0.61231 \\
\hline
$\mathcal{G}_{11}$  & 0.58587   &$\mathcal{G}_{129}$  & 0.61514  \\
\hline
\end{tabular}
\label{tab:fn_rectangular4}
\end{table}

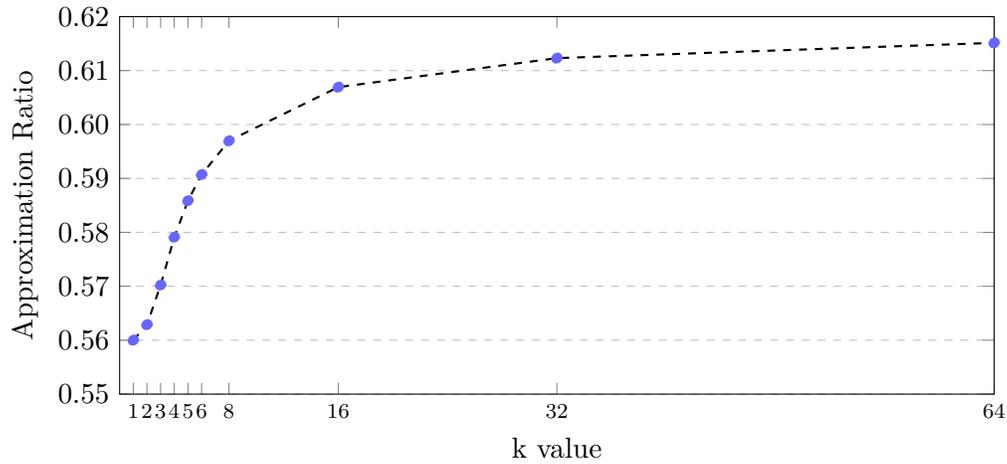
\begin{figure}[htbp]
\centering
\begin{tikzpicture}
\begin{axis}[
    width=0.8\textwidth,
    height=0.4\textwidth,
    xlabel={k value},
    ylabel={Approximation Ratio},
    xmin=0, xmax=64,
    ymin=0.55, ymax=0.62,
    xtick={1,2,3,4,5,6,8,16,32,64},
    ytick distance=0.01,
    ymajorgrids=true,
    grid style=dashed,
    x tick label style={font=\scriptsize},
    yticklabel style={
    /pgf/number format/fixed,
    /pgf/number format/precision=2,
    /pgf/number format/zerofill
},
]
\addplot[
    thick,
    dashed,
    mark=*,
    mark options={fill=blue!60, draw=blue!60,line width=0pt},
] coordinates {
    (1,0.56001)
    (2,0.56288)
    (3,0.57023)
    (4,0.57911)
    (5,0.58587)
    (6,0.59071)
    (8,0.59697)
    (16,0.60693)
    (32,0.61231)
    (64,0.61514)
};
\end{axis}
\end{tikzpicture}
\caption{LP results of \ranking{} on $\mathcal{G}_{2k+1}$.}
\end{figure}
\begin{theorem}
    For each graph class $\mathcal{G}_{2k+1}$ listed in \Cref{tab:fn_rectangular4}, the approximation ratio of \ranking{} on $\mathcal{G}_{2k+1}$ is at least the corresponding value reported in the table.
\end{theorem}
\begin{proof}
    By \cref{G_u_Ranking_Odd_Girth}, running the \ranking{} LP$(\mathcal{G}_{2k+1})$ with arbitrary $n$ gives a lower bound for the approximation ratio of \ranking{} on $\mathcal{G}_{2k+1}$.
\end{proof}

The LP values in \cref{tab:fn_rectangular4} exhibit a clear monotone trend: the lower bound improves steadily as the odd girth increases. This supports the view that short odd cycles are a major source of difficulty in analyzing and bounding the performance of \ranking{} on general graphs. Interestingly, the improvement from $\mathcal{G}_3$ to $\mathcal{G}_5$ is less pronounced than the subsequent increases from $\mathcal{G}_5$ to $\mathcal{G}_7$ and from $\mathcal{G}_7$ to $\mathcal{G}_9$. This is because in the formulation of the \ranking{} LP, we already proved that most alternating paths have length at least $4$. It would therefore be interesting to prove analogous statements for the occurrence of longer alternating paths. If one could show that alternating paths of length up to $2k$ are unlikely, then the LP value for $\mathcal{G}_{2k+3}$ would suggest a similar, though optimistic, approximation ratio of \ranking{} on general graphs.

This behavior is also consistent with the classical fact that \ranking{} admits a $1-\frac{1}{e}$ guarantee in the one-sided online bipartite matching model. Since bipartite graphs contain no odd cycle at all, graphs of large odd girth may be viewed as a relaxation of the bipartite setting. Our results therefore suggest that, as odd cycles become less likely to appear at a local scale, the behavior of \ranking{} on general graphs becomes increasingly bipartite-like.
\section{Tighter Lower Bounds for $(x_u,x_v,\bot)$ in \ranking{}}
\phantomsection
\label{sec:tigher_bounds}
In this section, we revisit the case $(x_u,x_v,\bot)$ and explain why the lower bound from \cref{fact:lowerbound_u_nobackup} can be strengthened. The key point is that when inserting $u^*$ creates a longer alternating path, the path structure may force $u$ to receive compensation from more vertices than those explicitly accounted for in the previous bound. In particular, the cases of alternating path length $4$ and length greater than $4$ lead to different compensation patterns, and this distinction yields a better lower bound. This observation motivates the refined analysis in this section and the introduction of a new impacting rank $\theta_3$.

Consider the following case: $u$ does not have a backup with respect to $\vecx_{\shortminus u^*}$, and inserting $u^*$ at rank $x\in [0,\theta_0)$ causes an alternating path of length $\geq 4$, resulting in $u$ becoming unmatched. Let the alternating path be the following: (where $u_2$ might be equal to $u_{k-2}$ when the path has length $4$)
$$\M(\vecx_{\shortminus u^*}(x))\oplus \M(\vecx_{\shortminus u^*})=(u^*,w, u_2,...,u_{k-2},v,u)$$
 By \cref{lem:victim_querycommit_case_u}, $u$ is the victim of $u^*$ and $u_{k-2}$, hence receives compensation from them. In \cref{fact:lowerbound_u_nobackup}, the term 
$$\int_0^{\theta_0}h(x,x_{\text{match of $u^*$}})\; dx\;+\;\theta_0\cdot  h(x_u, x_v)$$
  corresponds to the expected gain from these two vertices, where $h(x,x_{\text{match of $u^*$}})$ is the compensation from $u^*$, and $h(x_u, x_v)$ is the compensation from $u_{k-2}$. 

  Notice that another vertex that could pay compensation to $u$ is the vertex $u_2$. In particular, this vertex, being the backup of $w$, will have rank $\leq \theta_0$ by \cref{fact:theta_0_II}. Hence we know that there exists a copy of compensation for $u$ from $u_2$ of the form $h(y,z)$ where $y\leq \theta_0$. We also know that there exists a copy of compensation for $u$ from $u_{k-2}$ of the form $h(y,z)$ where $z=x_v$. This observation leads to the following two cases:

  \begin{itemize}
    \item If the alternating path has length $>4$, then $u_2$ and $u_{k-2}$ are distinct. In this case, $u$ may receive one copy of compensation from $u_2$ that is at least $h(\theta_0,0^+)$, and another copy from $u_{k-2}$ that is at least $h(x_u,x_v)$.
    \item If the alternating path has length $4$, then $u_2=u_{k-2}$. In this case, $u$ receives a single copy of compensation from this vertex that is at least $h(\theta_0,x_v)$.
  \end{itemize}
Both of these cases are better than receiving a single copy of compensation of amount $h(x_u,x_v)$, since $\theta_0$ is guaranteed to satisfy $\theta_0 \le x_u$ by \cref{fact:theta_0}.
 
\newcommand{\Llabel}[1]{%
  \ifcase#1
    $u^*$
  \or
    $w$
  \or
    $u_2$
  \or
    $u_3$
  \or
    $u_{k-2}$
  \or
    $v$
  \or
    $u$
  \fi
}

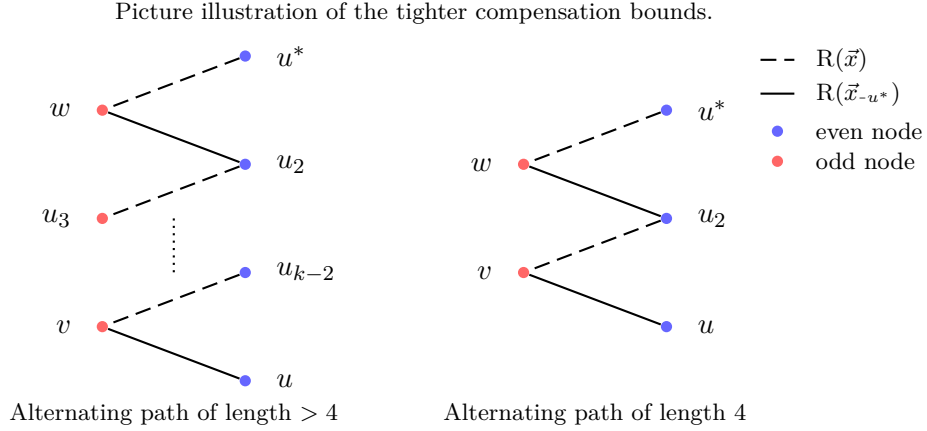
\newcommand{\Rlabel}[1]{%
  \ifcase#1
    $u^*$
  \or
    $w$
  \or
    $u_2$
  \or
    $v$
  \or
    $u$
  \fi
}

\def\FigShift{0cm} 

\makebox[\linewidth][c]{%
\hspace*{\FigShift}%
\begin{tikzpicture}[
  even/.style={circle, fill=blue!60!white, inner sep=1.5pt},
  odd/.style={circle, fill=red!60!white, inner sep=1.5pt},
  legendtext/.style={font=\footnotesize, align=left},
]

\def\TotalW{14}      
\def\PicW{4/5 * 14}  
\def\LegW{0.5/5 * 14} 

\pgfmathsetmacro{\SysL}{-\TotalW/2}
\pgfmathsetmacro{\SysR}{ \TotalW/2}
\pgfmathsetmacro{\PicL}{\SysL+1}
\pgfmathsetmacro{\PicR}{\SysL+1+\PicW}
\pgfmathsetmacro{\LegL}{\PicR}
\pgfmathsetmacro{\LegR}{\SysR}

\pgfmathsetmacro{\xA}{\PicL + \PicW/4}
\pgfmathsetmacro{\xB}{\PicL + 3*\PicW/4}
\pgfmathsetmacro{\xL}{(\LegL+\LegR)/2}

\def\xLR{0.95}
\def\yStep{0.72}
\def\labGap{0.28}

\foreach \i in {0,...,6} {
    \pgfmathsetmacro{\y}{(3-\i)*\yStep}
    \ifodd\i
        \node[odd] (L\i) at ({\xA-\xLR}, {\y}) {};
        \node[anchor=east] at ({\xA-\xLR-\labGap}, {\y}) {\Llabel{\i}};
    \else
        \node[even] (L\i) at ({\xA+\xLR}, {\y}) {};
        \node[anchor=west] at ({\xA+\xLR+\labGap}, {\y}) {\Llabel{\i}};
    \fi
}

\draw[dash pattern=on 6pt off 3pt, thick] (L0) -- (L1);
\draw[thick] (L1) -- (L2);
\draw[dash pattern=on 6pt off 3pt, thick] (L2) -- (L3);
\draw[dotted, thick] ({\xA}, {0}) -- ({\xA}, {-\yStep});
\draw[dash pattern=on 6pt off 3pt, thick] (L4) -- (L5);
\draw[thick] (L5) -- (L6);

\foreach \i in {0,...,4} {
    \pgfmathsetmacro{\y}{(2-\i)*\yStep}
    \ifodd\i
        \node[odd] (R\i) at ({\xB-\xLR}, {\y}) {};
        \node[anchor=east] at ({\xB-\xLR-\labGap}, {\y}) {\Rlabel{\i}};
    \else
        \node[even] (R\i) at ({\xB+\xLR}, {\y}) {};
        \node[anchor=west] at ({\xB+\xLR+\labGap}, {\y}) {\Rlabel{\i}};
    \fi
}

\draw[dash pattern=on 6pt off 3pt, thick] (R0) -- (R1);
\draw[thick] (R1) -- (R2);
\draw[dash pattern=on 6pt off 3pt, thick] (R2) -- (R3);
\draw[thick] (R3) -- (R4);

\pgfmathsetmacro{\yTitle}{-3*\yStep - 0.45}
\node[font=\footnotesize] at (\xA,\yTitle) {Alternating path of length $> 4$};
\node[font=\footnotesize] at (\xB,\yTitle) {Alternating path of length $4$};

\pgfmathsetmacro{\yTop}{3*\yStep}

\begin{scope}[shift={({\xL-1.5},{\yTop-0.05})}]
    \draw[dash pattern=on 6pt off 3pt, thick] (0,0) -- (0.45,0);
    \node[legendtext, anchor=west] at (0.60,0) {$\M(\vecx)$};

    \draw[thick] (0,-0.45) -- (0.45,-0.45);
    \node[legendtext, anchor=west] at (0.60,-0.45) {$\M(\vecx_{\shortminus u^*})$};

    \node[even] at (0.225,-0.95) {};
    \node[legendtext, anchor=west] at (0.60,-0.95) {even node};

    \node[odd] at (0.225,-1.35) {};
    \node[legendtext, anchor=west] at (0.60,-1.35) {odd node};
\end{scope}

\pgfmathsetmacro{\yCaption}{\yTop + 0.3}
\node[anchor=south, font=\footnotesize] at ({(\SysL+\SysR)/2}, \yCaption)
{Picture illustration of the tighter compensation bounds.};

\end{tikzpicture}%
}

  To formally establish this intuition, we define the following new impacting rank $\theta_3$, which is similar to $\theta_0$ but captures the collection of ranks where the alternating path causing $u$ to be unmatched has length $\geq 6$ and hence $u$ is not guaranteed to receive $h(\theta_0,x_v)$ amount of compensation from $u_{k-2}$.
  \begin{definition}[Length-$6$ Impacting Rank $\theta_3$]
\phantomsection
\label{def:theta_3_ranking}
    Assume $u$ is matched in $\M(\vecx_{\shortminus u^*})$. The length-$6$ impacting rank $\theta_3$ with respect to $\vecx_{\shortminus u^\ast}$ is
$$
\theta_3=\sup_{x\in[0,1]}\left\{
x \;\middle|\;
\begin{gathered}
u \text{ is made worse off in } \M(\vecx_{\shortminus u^\ast}(x)) \text{ and the third to last vertex $u_{k-2}$}\\
\text{ in the alternating path }
\M(\vecx_{\shortminus u^\ast}(x))\oplus \M(\vecx_{\shortminus u^\ast})
\text{ has rank $>\theta_0$}
\end{gathered}
\right\}.
$$
\end{definition}
It is easy to see that $0\leq \theta_3\leq \theta_0$, as the set of ranks that defines $\theta_3$ is a subset of the set of ranks that define $\theta_0$ in \cref{def:theta_0_ranking}. Similar to $\theta_0$, we also have two bounds regarding the rank of the third vertex $u_2$. Further, even without defining it explicitly, the alternating paths caused by inserting $u^*$ at these ranks have length $\geq 6$.
\begin{fact}[Properties of $\theta_3$]
\phantomsection
\label{fact:theta_3_II}
Fixing $\vecx_{\shortminus u^*}$, let $\theta_3$ be the length-$6$ impacting rank as defined in \cref{def:theta_3_ranking} with respect to $\vecx_{\shortminus u^*}$. Assuming $\theta_3>0$ is non-trivial, let $S_{\theta_3}$ be the collection of ranks $x$ that define $\theta_3$. I.e.,
    $$
S_{\theta_3}=\left\{
x \;\middle|\;
\begin{gathered}
u \text{ is made worse off in } \M(\vecx_{\shortminus u^\ast}(x)) \text{ and the third to last vertex $u_{k-2}$}\\
\text{ in the alternating path }
\M(\vecx_{\shortminus u^\ast}(x))\oplus \M(\vecx_{\shortminus u^\ast})
\text{ has rank $>\theta_0$}
\end{gathered}
\right\}.
$$
    Then
\begin{enumerate}
    \item $\theta_3\leq \theta_0$.
    \item For any $x\in S_{\theta_3}$, the third vertex $u_2$ in the alternating path $\M(\vecx_{\shortminus u^*}(x))\oplus \M(\vecx_{\shortminus u^*})$ has rank $x_{u_2}\leq \theta_3$.
    \item There exists a small enough punctured left neighborhood of $\theta_3$, denoted as $(\theta_3^-,\theta_3)$, such that $(\theta_3^-,\theta_3)\subseteq S_{\theta_3}$. Further, for any $x\in(\theta_3^-,\theta_3)$, the third vertex $u_2$ in the alternating path $\M(\vecx_{\shortminus u^*}(x))\oplus \M(\vecx_{\shortminus u^*})$ has rank $x_{u_2}= \theta_3$.
    \item  For any $x\in S_{\theta_3}$, the alternating path $\M(\vecx_{\shortminus u^*}(x))\oplus \M(\vecx_{\shortminus u^*})$ has length $\geq 6$.
\end{enumerate}
\end{fact}
\begin{proof}
    Property $1$ holds because $S_{\theta_3}$ is a subset of the set of ranks that defines $\theta_0$ in \cref{def:theta_0_ranking}.
    
    Properties $2$ and $3$ can be proved by the same proof as \cref{fact:theta_0_II} with $S$ replaced by $S_{\theta_3}$.

    Property 4: Fixing $x\in S_{\theta_3}$, the alternating path has even length since $u$ is made worse off. It suffices to show that the alternating path $\M(\vecx_{\shortminus u^\ast}(x))\oplus \M(\vecx_{\shortminus u^\ast})$ is not of length $2$ or $4$.

    If the length of $\M(\vecx_{\shortminus u^\ast}(x))\oplus \M(\vecx_{\shortminus u^\ast})$ is $2$, then the third to last vertex $u_{k-2}$ is $u^*$, hence of rank $x$. As $x\in S_{\theta_3}$, we have $x\leq \theta_3\leq \theta_0$. Thus the rank of $u_{k-2}$ cannot be $>\theta_0$.

    If the length of $\M(\vecx_{\shortminus u^\ast}(x))\oplus \M(\vecx_{\shortminus u^\ast})$ is $4$, then the third to last vertex $u_{k-2}$ is the third vertex $u_2$, which, by property $2$, is of rank $\leq \theta_3\leq \theta_0$. Thus cannot be $>\theta_0$.
\end{proof}
With the definition of $\theta_3$ and the above properties, we can tighten our bounds for profiles $(x_u,x_v,\bot)$. We first highlight which parts of the gains in \cref{G_ranking_nobackup} we will be reconsidering.
\begin{definition}[Compensation Gains and Match Gains]
    Fix $\vecx_{\shortminus u^*}$ with profile $(x_u,x_v,\bot)$. Let the compensation gains be the expected gains ($CG$) for $u$ and $u^*$ over the randomization of $u^*$ collected when they are unmatched (i.e. from compensation):
    \[
CG(x_u,x_v,\theta_0)=
\left\{
\begin{array}{l r}
(1-x_v)\cdot(h(x_u,\theta_0)+h(x_v,x_u))\\
+\int_0^{\theta_0} h(x,x_u)\;dx + \theta_0\cdot h(x_u,x_v),
& \text{if }x_u<x_v,\\[1em]
(1-\max\{x_v,\theta_0\})\cdot(h(x_v,\theta_0)+h(x_v,x_u))\\
+\int_0^{\theta_0} h(x,x_v)\;dx + \theta_0\cdot h(x_u,x_v),
& \text{if }x_v<x_u\wedge \theta_0<x_u,\\[1em]
(1-\theta_0)\cdot h(x_v,x_u) +\int_0^{\theta_0} h(x,x_v)\;dx,
& \text{if } x_v<x_u \wedge \theta_0=x_u.\\[1em]
\end{array}
\right.
\]
Let the match gains $(MG)$ be the remaining parts of the gains for $u$ and $u^*$ over the randomization of $u^*$ collected when they are matched:
 \[
MG(x_u,x_v,\theta_0)=\left\{
\begin{array}{l r}
\int_0^{x_v} g_P(x_u,x)\;dx+(1-\theta_0)\cdot g_B(x_u,x_v),
& \text{if } x_u<x_v ,\\[1em]
\int_0^{\theta_0} g_P(x_v,x)\;dx+\int_{\theta_0}^{\max\{x_v,\theta_0\}} g_P(x_u,x)\;dx\\
+(1-\theta_0)\cdot g_B(x_u,x_v),
& \text{else} .
\end{array}
\right.\]
\end{definition}
Then we can rewrite \cref{G_ranking_nobackup} as the summation of $CG$ and $MG$:
\begin{fact}
    Fix $\vecx_{\shortminus u^*}$ with profile $(x_u,x_v,\bot)$. We have:
    $$\ev_{x}[\text{ gain}(u) + \text{gain}(u^\ast)]\geq \inf_{\theta_0\leq x_u}\{CG(x_u,x_v,\theta_0)+MG(x_u,x_v,\theta_0)\}.$$
\end{fact}
\begin{proof}
    This is just rewriting \cref{G_ranking_nobackup}.
\end{proof}
We will tighten the compensation gain function $CG$ by considering $\theta_3\in[0,\theta_0]$ for the cases where the alternating path has length $\geq 4$. We divide each of the cases into two subcases depending on whether $\theta_3=\theta_0$ or not, giving rise to the following tightened compensation gain function $CG_T$. 
\begin{claim}
\phantomsection
\label{func:tight_comp_gain}
    Fix $\vecx_{\shortminus u^*}$ with profile $(x_u,x_v,\bot)$. Let the tightened compensation gain be the following function $CG_T$, for $\theta_3\leq \theta_0\leq x_u$:
    
\[
CG_T(x_u,x_v,\theta_0,\theta_3)=
\left\{
\begin{array}{l r}
(1-x_v)\cdot(h(x_u,\theta_0)+h(x_v,x_u))\\
+(1-x_v)\cdot(h(x_u,\theta_3)+h(x_v, \theta_3))\\
+\int_0^{\theta_0} h(x,x_u)\;dx + \theta_3\cdot (h(x_u,x_v)+h(\theta_3,0^+))\\
+(\theta_0-\theta_3)\cdot h(\theta_0,x_v),
& \text{if }x_u<x_v\wedge \theta_3<\theta_0,\\[1em]
(1-x_v)\cdot h(x_v,x_u)\\
+(1-x_v)\cdot(h(x_u,\theta_3)+h(x_v, \theta_3))\\
+\int_0^{\theta_0} h(x,x_u)\;dx + \theta_3\cdot (h(x_u,x_v)+h(\theta_3,0^+)),
& \text{if }x_u<x_v\wedge \theta_3=\theta_0,\\[1em]
(1-\max\{x_v,\theta_0\})\cdot(h(x_v,\theta_0)+h(x_v,x_u))\\
+(1-\max\{x_v,\theta_0\})\cdot(h(x_v,\theta_3)+h(x_v, \theta_3))\\
+\int_0^{\theta_0} h(x,x_v)\;dx + \theta_3\cdot (h(x_u,x_v)+h(\theta_3,0^+))\\
+(\theta_0-\theta_3)\cdot h(\theta_0,x_v),
& \text{if }x_v<x_u\wedge \theta_3<\theta_0,\\[1em]
(1-\max\{x_v,\theta_0\})\cdot h(x_v,x_u)\\
+(1-\max\{x_v,\theta_0\})\cdot(h(x_v,\theta_3)+h(x_v, \theta_3))\\
+\int_0^{\theta_0} h(x,x_v)\;dx + \theta_3\cdot (h(x_u,x_v)+h(\theta_3,0^+)),
& \text{if }x_v<x_u\wedge \theta_3=\theta_0,\\[1em]
(1-\theta_0)\cdot h(x_v,x_u) +\int_0^{\theta_0} h(x,x_v)\;dx,
& \text{if } x_v<x_u \wedge \theta_0=x_u.
\end{array}
\right.
\]
We have that the total expected gain of $u$ and $u^*$ over the randomization of $u^*$ collected when they are unmatched is lower bounded by $CG_T(x_u,x_v,\theta_0,\theta_3)$.
\end{claim}
\begin{proof}
    We first break the expected compensation gains into pieces for a better understanding. Assume for now that both $\theta_3$ and $\theta_0$ are non-zero. By \cref{fact:theta_3_II} and \cref{fact:theta_0_II}, there exist small enough neighborhoods $(\theta_3^-,\theta_3)$ and $(\theta_0^-,\theta_0)$ such that, when $u^*$ is inserted in these neighborhoods, it causes alternating paths of lengths $\geq 6$ and $\geq 4$, respectively, resulting in $u$ being unmatched. When $\theta_3\neq\theta_0$, we could assume that these two sets are disjoint, by choosing $\theta_0^-$ close enough to $\theta_0$. Let the respective alternating paths be the following (when $\theta_3=\theta_0$, we only consider $x_{\theta_3}$)
    \begin{align*}
                 \M(\vecx_{\shortminus u^*}(x_{\theta_3}))\oplus\M(\vecx_{\shortminus u^*})&=(u^*,w^{\theta_3},u_2^{\theta_3},...u_{k-3}^{\theta_3},u_{k-2}^{\theta_3},v,u),\\
         \M(\vecx_{\shortminus u^*}(x_{\theta_0}))\oplus\M(\vecx_{\shortminus u^*})&=(u^*,w^{\theta_0},...,u_{k-2}^{\theta_0},v,u).
    \end{align*}
    We will be collecting compensations for $u^*$ from a few of the above designated vertices.
    
\noindent \textbf{Case $\bm{x_u<x_v\wedge \theta_3<\theta_0}$.} For the first case, the function can be decomposed into the following parts:
    \[\begin{array}{l}
\underbrace{(1-x_v)\cdot(h(x_u,\theta_0)+h(x_v,x_u))}_{T_1}
+\underbrace{(1-x_v)\cdot(h(x_u,\theta_3)+h(x_v,\theta_3))}_{T_2}\\
+\underbrace{\int_0^{\theta_0} h(x,x_u)\;dx}_{T_3}
+\underbrace{\theta_3\cdot (h(x_u,x_v)+h(\theta_3,0^+))}_{T_4}+\underbrace{(\theta_0-\theta_3)\cdot h(\theta_0,x_v)}_{T_5}.
\end{array}\]
Let $x$ be the random rank of $u^*$ when inserted into $\vecx_{\shortminus u^*}$. Each $T_i$ term corresponds to the following compensations received:
\begin{description}
    \item[$\bm{T_1}$:] Compensation received by $u^*$ from $w^{\theta_0}$ and $v$ for $x\in (x_v,1]$.
    \item[$\bm{T_2}$:] Compensation received by $u^*$ from $w^{\theta_3}$ and $u_{k-3}^{\theta_3}$ for $x\in (x_v,1]$.
    \item[$\bm{T_3}$:] Compensation received by $u$ from $u^*$ for $x\in (0,\theta_0]$.
    \item[$\bm{T_4}$:] Compensation received by $u$ for $x\in [0,\theta_3)$ from vertices other than $u^*$.
    \item[$\bm{T_5}$:] Compensation received by $u$ for  $x\in (\theta_3,\theta_0)$ from vertices other than $u^*$.
\end{description}
We analyze the function term by term. Recall that by our function constraints \cref{def:function_constraint_ranking}, the amount of gain a matched vertex receives outweighs at least $4$ copies of compensation. In the tightened compensation function, for any vertex $u$ or $u^*$, at any $x\in [0,1]$, we will gather at most $4$ copies of compensation for it, hence we could still lower bound the gain for any vertex assuming it is unmatched when it could either be matched or unmatched. 

\noindent\textbf{$\bm{T_1}$} When $x_u<x_v$, for $x\in(x_v,1]$, $u^*$ could be matched or unmatched. Assume $u^*$ is unmatched. $u^*$ will receive at least $h(x_u,\theta_0)$ amount of compensation from $w^{\theta_0}$ and $h(x_v,x_u)$ amount of compensation from $v$ by the same line of reasoning as in \cref{fact:lowerbound_u_star_nobackup}. This constitutes the $T_1$ term. 

\noindent\textbf{$\bm{T_2}$}
For $x\in(x_v,1]$, assume $u^*$ is unmatched. By \cref{lem:victim_querycommit_case_u_star}, $u^*$ is the victim of $w^{\theta_3}$ and $u_{k-3}^{\theta_3}$ and hence will receive compensation from them. First we show that we could assume these two vertices are not $v$ or $w^{\theta_0}$ to avoid double counting with the $T_1$ term. 

These two vertices are not $v$ since $v$ is already the second to last vertex $u_{k-1}^{\theta_3}$ in the alternating path $\M(\vecx_{\shortminus u^*}(\theta_3^-))\oplus\M(\vecx_{\shortminus u^*})$, and the path has length $\geq 6$.

By \cref{fact:theta_3_II}, property 2, and our choice of $x_{\theta_3}$, $w^{\theta_3}$ is matched to a vertex of rank $\theta_3$. By \cref{fact:theta_0_II}, property 2, and our choice of $x_{\theta_0}$, $w^{\theta_0}$ is matched to a vertex of rank $\theta_0$. These two facts together show that $w^{\theta_3}$ and $w^{\theta_0}$ are matched to vertices of different ranks ($\theta_3<\theta_0$), hence cannot be the same vertex.

To show that $u_{k-3}^{\theta_3}\neq w^{\theta_0}$, notice that the match of $u_{k-3}^{\theta_3}$ in $\M(\vecx_{\shortminus u^*}(x))$, is $u_{k-2}^{\theta_3}$. This vertex is the third to last vertex in the alternating path $\M(\vecx_{\shortminus u^*}(x_{\theta_3}))\oplus\M(\vecx_{\shortminus u^*})$. By our choice of $x_{\theta_3}$, the rank of $u_{k-2}^{\theta_3}$ is $>\theta_0$ as $x_{\theta_3}\in S_{\theta_3}$. Hence, $u_{k-3}^{\theta_3}\neq w^{\theta_0}$ as the match of $w^{\theta_0}$ in $\M(\vecx_{\shortminus u^*}(x))$ has rank $\theta_0$.

Now we lower bound the value for compensations from $w^{\theta_3}$ and $u_{k-3}^{\theta_3}$. As $w^{\theta_3}$ is the match of $u^*$ when $u^*$ has rank $x_{\theta_3}<x_v$, by \cref{fact:matching_guarantee_ranking}, property 1, $w^{\theta_3}$ has rank at most $x_u$. The match of $w^{\theta_3}$, as argued before, has rank $\theta_3$. So compensation from $w^{\theta_3}$ is at least $h(x_u,\theta_3)$. The rank of $u_{k-3}^{\theta_3}$ is at most $x_v$, as it is an odd indexed vertex in the alternating path before $x_v$, and hence by \cref{fact:monotonicity_for_alt_path_ranking}, has rank $\leq x_v$. The match of $u_{k-3}^{\theta_3}$ in $\M(\vecx_{\shortminus u^*}(x))$, which is $u_{k-2}^{\theta_3}$, has rank $>\theta_0\geq \theta_3$. So compensation from $u_{k-3}^{\theta_3}$ is at least $h(x_v,\theta_3)$\footnote{We could as well choose the better rank $\theta_0$ instead of $\theta_3$. However, notice that when $\theta_3=0$, our choice of $x_{\theta_3}$ is an invalid rank. By bounding the compensation with $h(x_v,\theta_3)$, this degenerate case is also covered as $h(x_v,0)=0$, and we are effectively introducing $0$ compensation.}.

\noindent\textbf{{$\bm{T_3}$}.} For $x\in [0,\theta_0)$, $u$ could be unmatched due to $u^*$. We assume it is  unmatched. Then $u$ is the victim of $u^*$ hence will receive compensation from $u^*$ with amount $h(x,x_{\text{match of }u^*})$. By the same line of reasoning as in \cref{G_ranking_nobackup}, we could assume this compensation is at least $h(x,x_u)$.

\noindent\textbf{{$\bm{T_5}$}.} For $x\in (\theta_3,\theta_0)$, if $u$ is unmatched, by definition of $\theta_3$, we know that the third to last vertex $u_{k-2}$ in the alternating path $\M(\vecx_{\shortminus u^*}(x))\oplus\M(\vecx_{\shortminus u^*})=(u^*,u_1,...u_{k-2},v,u)$ has rank $\leq \theta_0$. In this case $u$ receives compensation from $u^*$ and the vertex $u_{k-2}$, which is matched to $v$. Since $u_{k-2}$ has rank $\leq \theta_0$, it pays at least $h(\theta_0,x_v)$ amount of compensation to $u$. Since $x_u<x_v$, by \cref{fact:theta_0}, we have $u_{k-2}\neq u^*$, so we are not double counting.

\noindent\textbf{{$\bm{T_4}$}.} Let $S_{\theta_3}$ be the set of ranks as defined in \cref{fact:theta_3_II}. Fix $x\in [0,\theta_3)$, assume $u$ is unmatched in $\M(\vecx_{\shortminus u^*}(x))$ as usual. We separately calculate the compensation for $u$ for cases $x\notin S_{\theta_3}$ and $x\in S_{\theta_3}$.

If $x\notin S_{\theta_3}$, then the third to last vertex $u_{k-2}$ in its respective alternating path has rank $\leq \theta_0$ by definition. Similar to the $T_5$ term case, this implies that $u$ will receive at least $h(\theta_0, x_v)$ amount of compensation from $u_{k-2}$.

If $x\in S_{\theta_3}$, then by \cref{fact:theta_3_II}, property 4, the alternating path has length $\geq 6$, let the path be $u^*,u_1,u_2,...,u_{k-2},v,u$. $u$ will receive compensation from at least two other vertices $u_2$ and $u_{k-2}$ besides $u^*$. By \cref{fact:theta_3_II}, property 2, the rank of $u_2$ is $\leq \theta_3$. Hence the compensation from $u_2$ is at least $h(\theta_3, 0^+)$\footnote{we could assume W.L.O.G. no vertex has rank exactly $0$ as such event has $0$ probability measure, hence the rank of the match of $u_2$ is at least $0^+$.}. The compensation from $u_{k-2}$ is at least $h(x_u,x_v)$, as $u_{k-2}$ is matched to $v$ and  $(u_{k-2},v)$ is queried before $(u,v)$, which implies $x_{u_{k-2}}\leq x_u$.

\noindent\textbf{Assuming $\bm{h(\theta_3,0^+)+h(x_u,x_v)}$ is the Smaller Compensation.} Combining the two cases, the compensation of $u$ received at $x\in [0,\theta_3)$ is at least the minimum of $h(\theta_0,x_v)$ and $h(\theta_3,0^+)+h(x_u,x_v)$. If $h(\theta_0,x_v)$ is the smaller one, then we could uniformly lower bound the compensation for $u$ for all $x\in[0,\theta_0)$ with $h(\theta_0,x_v)$ by combining terms $T_4$ and $T_5$, in which case the function is further lower bounded by setting $\theta_3=0$, as the term $T_2$ decreases with $\theta_3$.

To elaborate on this, if $\min\{h(\theta_0,x_v),h(\theta_3,0^+)+h(x_u,x_v)\}=h(\theta_0,x_v)$. Then the whole system of bounds reduces to
\begin{align*}
    (1-x_v)\cdot(h(x_u,\theta_0)+h(x_v,x_u))+(1-x_v)\cdot(h(x_u,\theta_3)+h(x_v, \theta_3))+\int_0^{\theta_0} h(x,x_u)\;dx  +\theta_0\cdot h(\theta_0,x_v),
\end{align*}
which is further lower bounded by $CG_T(x_u,x_v,\theta_0,0)$ having value
\begin{align*}
    &(1-x_v)\cdot(h(x_u,\theta_0)+h(x_v,x_u))+\int_0^{\theta_0} h(x,x_u)\;dx  +\theta_0\cdot h(\theta_0,x_v)
\end{align*}
Hence we only need to consider the case $h(\theta_3,0^+)+h(x_u,x_v)<h(\theta_0,x_v)$. This argument holds for all the subsequent cases as well.

\noindent \textbf{Case $\bm{x_u<x_v\wedge \theta_3=\theta_0}$.}  This case is almost identical to the previous case. The only change happens at $T_1$ and $T_5$. $T_5$ diminishes to $0$ because $\theta_0-\theta_3=0$. For $T_1,$ we removed the compensation $h(x_u,\theta_0)$ as we can no longer assume $w^{\theta_3}\neq w^{\theta_0}$.

\noindent \textbf{Case $\bm{x_v<x_u\wedge \theta_3<\theta_0}$.}
This case differs from the first case at $T_1$ and $T_2$ terms. All of the first coordinate of the $h$ function is set to $x_v$. We always have this because any compensation $h(y,z)$ paid to $u^*$ is from an odd-indexed vertex in alternating paths $\M(\vecx_{\shortminus u^*}(x_{\theta_3}))\oplus\M(\vecx_{\shortminus u^*})$ and $\M(\vecx_{\shortminus u^*}(x_{\theta_0}))\oplus\M(\vecx_{\shortminus u^*})$. These two alternating paths always end at $v,u$, with $v$ being the last odd-indexed vertex. Hence any odd-indexed vertex in these paths has rank $\leq x_v$. As a result, we can safely upper-bound the first coordinate for all of the $h$ terms in $T_1,T_2$ by $x_v$. Further, the range of $x$ where $u^*$ is assumed to be unmatched is $(\max(x_v,\theta_0),1]$. So the compensation is multiplied by the probability $(1-\max(x_v,\theta_0))$.

\noindent \textbf{Case $\bm{x_v<x_u\wedge \theta_3=\theta_0}$.}
The transition from $\theta_3<\theta_0$ to $\theta_3=\theta_0$ for the $x_v<x_u$ case is the same as the transition for the $x_u<x_v$ case.

\noindent \textbf{Case $\bm{x_v<x_u\wedge \theta_0=x_u}$.}
This is the same bound as \cref{G_ranking_nobackup} for case $x_v<x_u\wedge \theta_0=x_u$.

\noindent \textbf{Case $\bm{\theta_3=0}$ or $\bm{\theta_0=0}$.}
When either of the parameters $\theta_3$ and $\theta_0$ is $0$, the corresponding compensations $h(y,\theta_3)$ and $h(y,\theta_0)$ become $0$ by the function constraints \cref{def:function_constraint_ranking}. Hence we are not over-approximating for the degenerate cases.
\end{proof}
With the tightened compensation gain bounds, we could derive a tightened system of lower bounds for the expected total gain for profiles $(x_u,x_v,\bot)$.
\begin{fact}
\phantomsection
\label{G_tighter_no_backup}
    Fixing $\vecx_{\shortminus u^*}$ with profile $(x_u,x_v,\bot)$. We have:
    $$\ev_{x}[\text{ gain}(u) + \text{gain}(u^\ast)]\geq \inf_{\theta_3\leq\theta_0\leq x_u}\{CG_T(x_u,x_v,\theta_0,\theta_3)+MG(x_u,x_v,\theta_0)\}.$$
\end{fact}
\begin{proof}
    By \cref{func:tight_comp_gain}, $CG_T$ is a valid lower bound for the expected compensation gains for any $0\leq\theta_3\leq \theta_0\leq x_u$.  Further, we assumed that all the match gains collected in the $MG$ function are independent of $\theta_3$, as $\theta_3$ only gives a finer analysis of the unmatched $u,u^*$, and has no influence on the gains for matched $u,u^*$.
\end{proof}
\subsection{Discretized Tightened \ranking{} LP}
\phantomsection
\label{Discretized_Tightened_Ranking_LP}
The system of bounds in \cref{G_tighter_no_backup} can be written as the following function $G_T(x_u,x_v,\bot)$:
\begin{align*}
&\min_{\theta_0\leq x_u}\left\{
    \;\;\inf_{\theta_3<\theta_0}\;\;\Biggl[\;\int_0^{x_v} g_P(x_u,x)\,dx  +(1-\theta_0)\cdot\,g_B(x_u,x_v)\right.\\
&\qquad\qquad\qquad\qquad
        +(1-x_v)\cdot(h(x_u,\theta_0)+h(x_v,x_u))\\
&\qquad\qquad\qquad\qquad
+(1-x_v)\cdot(h(x_u,\theta_3)+h(x_v, \theta_3))\\
&\qquad\qquad\qquad\qquad
+\int_0^{\theta_0} h(x,x_u)\;dx + \theta_3\cdot (h(x_u,x_v)+h(\theta_3,0^+))\\
&\qquad\qquad\qquad\qquad
+(\theta_0-\theta_3)\cdot h(\theta_0,x_v)\bigr) \; \Biggr],
&(1)\phantom{\Biggr\}}& \\[0.75em]
&\qquad\quad
\;\;\min_{\theta_3=\theta_0}\;\;\Biggl[\;\int_0^{x_v} g_P(x_u,x)\,dx  +(1-\theta_0)\cdot\,g_B(x_u,x_v)\\
&\qquad\qquad\qquad\qquad
+(1-x_v)\cdot h(x_v,x_u)\\
&\qquad\qquad\qquad\qquad
+(1-x_v)\cdot(h(x_u,\theta_3)+h(x_v, \theta_3))\\
&\qquad\qquad\qquad\qquad
+\int_0^{\theta_0} h(x,x_u)\;dx + \theta_3\cdot (h(x_u,x_v)+h(\theta_3,0^+)) \; \Biggr],
&(2)\Biggr\}
&\text{, (if } x_u<x_v),\\
&\min_{\theta_0\leq x_u}\Biggl\{
    \inf_{\theta_3<\theta_0}\;\;\Biggl[\;
        \int_0^{\theta_0} g_P(x_v,x)\,dx
        +\int_{\theta_0}^{\max\{\theta_0,x_v\}} g_P(x_u,x)\,dx +(1-\theta_0)\cdot\,g_B(x_u,x_v)\\
&\qquad\qquad\qquad\qquad
+(1-\max\{\theta_0,x_v\})\cdot(h(x_v,\theta_0)+h(x_v,x_u))\\
&\qquad\qquad\qquad\qquad
+(1-\max\{\theta_0,x_v\})\cdot(h(x_v,\theta_3)+h(x_v, \theta_3))\\
&\qquad\qquad\qquad\qquad
+\int_0^{\theta_0} h(x,x_v)\;dx + \theta_3\cdot (h(x_u,x_v)+h(\theta_3,0^+))\\
&\qquad\qquad\qquad\qquad
+(\theta_0-\theta_3)\cdot h(\theta_0,x_v)\bigr) \; \Biggr],
&(3)\phantom{\Biggr\}}& \\[0.75em]
&   \qquad\quad\inf_{\theta_3=\theta_0}\;\;\Biggl[\;
        \int_0^{\theta_0} g_P(x_v,x)\,dx
        +\int_{\theta_0}^{\max\{\theta_0,x_v\}} g_P(x_u,x)\,dx +(1-\theta_0)\cdot\,g_B(x_u,x_v)\\
&\qquad\qquad\qquad\qquad
+(1-\max\{\theta_0,x_v\})\cdot h(x_v,x_u)\\
&\qquad\qquad\qquad\qquad
+(1-\max\{\theta_0,x_v\})\cdot(h(x_v,\theta_3)+h(x_v, \theta_3))\\
&\qquad\qquad\qquad\qquad
+\int_0^{\theta_0} h(x,x_v)\;dx + \theta_3\cdot (h(x_u,x_v)+h(\theta_3,0^+)) \; \Biggr],
&(4)\phantom{\Biggr\}}& \\[0.75em]
&\qquad\qquad
    \min_{\theta_0=x_u}\;\;\Biggl[\;
        \int_0^{\theta_0} g_P(x_v,x)\,dx
        +(1-\theta_0)\cdot\,h(x_v,x_u) \\
&\qquad\qquad\qquad\qquad
        +\int_0^{\theta_0} h(x,x_v)\,dx
        +(1-\theta_0)\cdot\,g_B(x_u,x_v)
    \;\Biggr]
&(5)\Biggr\}
&\text{, (if } x_v<x_u).
\end{align*}

Following the same discretization method as \cref{subsec:discre-ranking_LP}, for $x_u\in (\frac{i_u-1}{n},\frac{i_u}{n}]$ and $x_v\in (\frac{i_v-1}{n},\frac{i_v}{n}]$, the bound  of \cref{G_tighter_no_backup} can be translated into the following system of constraints $G_T(i_u,i_v,\bot)$, where $i_{\theta_3}$ and $i_{\theta_0}$ take integers in $\{0,...,n\}$:

\begin{align*}
\shortintertext{$\forall i_{\theta_3}< i_{\theta_0}\leq i_u\leq i_v$}
\qquad G_T(i_u,i_v,\bot)\leq{}&
\frac{1}{n}\left[\left(\sum_{j=1}^{i_v-1} g_P(i_u,j)\right)+\frac{1}{2}g_P(i_u,i_v)+(n-i_{\theta_0})\cdot g_B(i_u,i_v)\right. & \\
&\quad +(n-i_v)\cdot(h(i_u,i_{\theta_0})+h(i_v,i_u)+h(i_u,i_{\theta_3})+h(i_v,i_{\theta_3})) & \\
&\quad + \left(\sum_{j=1}^{i_{\theta_0}} h(j,i_u)\right)+i_{\theta_3}\cdot(h(\min\{i_{\theta_3}+1,i_u\},1)+h(i_u,i_v)) & \\
&\quad +(i_{\theta_0}-i_{\theta_3})\cdot h(\min\{i_{\theta_0}+1,i_u\},i_v)\left. \vphantom{\int} \right], & (1) \\[0.6em]
\shortintertext{$\forall i_{\theta_3}=i_{\theta_0}\leq i_u\leq i_v$}
\qquad G_T(i_u,i_v,\bot)\leq{}&
\frac{1}{n}\left[\left(\sum_{j=1}^{i_v-1} g_P(i_u,j)\right)+\frac{1}{2}g_P(i_u,i_v)+(n-i_{\theta_0})\cdot g_B(i_u,i_v)\right. & \\
&\quad +(n-i_v)\cdot(h(i_v,i_u)+h(i_u,i_{\theta_3})+h(i_v,i_{\theta_3})) & \\
&\quad + \left(\sum_{j=1}^{i_{\theta_0}} h(j,i_u)\right)+i_{\theta_3}\cdot(h(\min\{i_{\theta_3}+1,i_u\},1)+h(i_u,i_v))\left. \vphantom{\int} \right], & (2) \\[0.6em]
\shortintertext{$\forall i_{\theta_0}, i_v\leq i_u\wedge i_{\theta_3}<i_{\theta_0}$}
\qquad G_T(i_u,i_v,\bot)\leq{}&
\frac{1}{n}\left[\vphantom{\sum_{j=1}^{i_{\theta_0}}}\right.\left(\sum_{j=1}^{i_{\theta_0}} g(i_v,j)\right)+\left(\sum_{j=i_{\theta_0}+1}^{i_{v}-1} g_P(i_u,j)\right)+(n-i_{\theta_0})\cdot g_B(i_u,i_v) & \\
&\quad +(n-\max\{i_{\theta_0},i_v-1\})\cdot(h(i_v,i_{\theta_0})+h(i_v,i_u)) & \\
&\quad +(n-\max\{i_{\theta_0},i_v-1\})\cdot(h(i_v,i_{\theta_3})+h(i_v,i_{\theta_3})) & \\
&\quad + i_{\theta_3}\cdot(h(\min\{i_{\theta_3}+1,i_u\},1)+h(i_u,i_v)) & \\
&\quad +(i_{\theta_0}-i_{\theta_3})\cdot h(\min\{i_{\theta_0}+1,i_u\},i_v)\left. \vphantom{\int} \right], & (3) \\[0.6em]
\shortintertext{$\forall i_{\theta_0}, i_v\leq i_u\wedge i_{\theta_3}=i_{\theta_0}$}
\qquad G_T(i_u,i_v,\bot)\leq{}&
\frac{1}{n}\left[\vphantom{\sum_{j=1}^{i_{\theta_0}}}\right.\left(\sum_{j=1}^{i_{\theta_0}} g(i_v,j)\right)+\left(\sum_{j=i_{\theta_0}+1}^{i_{v}-1} g_P(i_u,j)\right)+(n-i_{\theta_0})\cdot g_B(i_u,i_v) & \\
&\quad +(n-\max\{i_{\theta_0},i_v-1\})\cdot h(i_v,i_u) & \\
&\quad +(n-\max\{i_{\theta_0},i_v-1\})\cdot(h(i_v,i_{\theta_3})+h(i_v,i_{\theta_3})) & \\
&\quad +i_{\theta_3}\cdot(h(\min\{i_{\theta_3}+1,i_u\},1)+h(i_u,i_v))\left. \vphantom{\int} \right], & (4) \\[0.6em]
\shortintertext{$\forall  i_v\leq i_u \wedge i_u-1\leq i_{\theta_0}\leq i_u$}
\qquad G_T(i_u,i_v,\bot)\leq{}&
\frac{1}{n}\left[\left(\sum_{j=1}^{i_{\theta_0}} g(i_v,j)\right)+ (n-i_{\theta_0})\cdot h(i_v,i_u) \right. & \\
&\quad + (n-i_{\theta_0})\cdot g_B(i_u,i_v)\left. \vphantom{\int}\right], & (5) 
\end{align*}

We point out a few key assumptions we are making:
\begin{itemize}
    \item We assume the system of continuous bounds are bounded by setting the variables $x_v,\theta_0,\theta_3$ as integral multiples of $\frac{1}{n}$. This is true as for each of the cases, fixing the sub-intervals for $x_v,\theta_0,\theta_3$ equivalently fixes the values of the $h$ and $g$ functions (as the two functions are piece-wise constant). And the expressions are equivalently linear combinations of the variables.

    \item In case $(1),...,(4)$, we consider grid point values $i_{\theta_0}\in\{0,...,i_u\}$. For $h$ function of form $h(y,\theta_0)$, we plug in the value $h(y, \frac{i_{\theta_0}}{n})$ and for $h$ function of form $h(\theta_0, y)$, we plug in the value $h(\min\{(\frac{i_{\theta_0}}{n})^+,\frac{i_u}{n}\}, y)$, as these two choices result in the smaller $h$ values.

    \item In case $(1),...,(4)$ we also consider the grid point values for $i_{\theta_3}\in\{0,...,i_{\theta_0}\}$. The choice of $h$ values follows the same line of reasoning as $\theta_0$. Formally, for case $(1)$ and $(3)$, we also need to consider the grid point $i_{\theta_3}=i_{\theta_0}$. But notice $(1),(3)$ are strictly dominated by $(2),(4)$ respectively when $i_{\theta_3}=i_{\theta_0}$, so we can omit this case for $(1)$ and $(3)$.
\end{itemize}
The remaining discretization assumptions follow the same discretization idea as for $G(x_u,x_v,\bot)$ in \cref{subsec:discre-ranking_LP}, including assuming $x_v$ is uniformly distributed in $(\frac{i_v-1}{n},\frac{i_v}{n}]$ for cases $(1)$ and $(2)$, assuming $x_v=(\frac{i_v-1}{n})^+$ for cases $(3)$ and $(4)$, and assuming $\theta_0=(\frac{i_u-1}{n})^+$ and $\frac{i_u}{n}$ for case $(5)$.

Replacing the discretized lower bound $G(i_u,i_v,\bot)$ in \cref{G_u_Ranking_discretized} by $G_T(i_u,i_v,\bot)$, we get a tighter system of lower bounds.
\begin{claim}
\phantomsection
\label{G_u_Ranking_discretized_tight}
    For $x_u\in(\frac{i_u-1}{n},\frac{i_u}{n}]$, the bound for uniformly distributed profiles (\cref{claim:ranking_lower_bound_uniform_profiles}) can also be translated to $G^T_u$, where G is the lower bound function defined in \cref{subsec:discre-ranking_LP} and $G_T$ is the above tightened lower bound system: 
\begin{align*}
    &\forall i_u &G^T_u(i_u)\leq\; & G(i_u,\bot,\bot)\\
    &\forall i_u, i_{\text{start of $v$}} &G^T_u(i_u)\leq\; &\frac{1}{n+1-i_{\text{start of $v$}}} \sum_{j=i_{\text{start of $v$}}}^n G_T(i_u,j,\bot)\\
    &\forall i_u, \forall i_{\text{start of $v$}}\leq i_b &G^T_u(i_u)\leq\; &\frac{1}{i_b+1-i_{\text{start of $v$}}} \sum_{j=i_{\text{start of $v$}}}^{i_b}G(i_u,j,\min\{i_b+1,n\})\\
    &\forall i_u, \forall i_{\text{start of $v$}}\leq i_b &G^T_u(i_u)\leq\; &\frac{1}{i_b+1-i_{\text{start of $v$}}} \sum_{j=i_{\text{start of $v$}}}^{i_b}G(i_u,j,i_b)
\end{align*}
And  the approximation ratio for \ranking{} on general graphs is lower bounded by
    $$\frac{1}{n}\sum_{i_u=1}^n G^T_u(i_u).$$
\end{claim}
\begin{proof}
    This is because $G_T(x_u,x_v,\bot)$ is also a system of valid lower bounds for the total expected gain for profiles $(x_u,x_v,\bot)$ and that $G_T(i_u,i_v,\bot)$ is a valid discretization.
\end{proof}

\section{Lower Bound by Uniform Distributions on Profiles}
\phantomsection
\label{uniform-profiles}
To prove that we can bound an increasing distribution on profiles by the worst-case uniform distribution, we invoke the Fubini–Tonelli Theorem, which states:
\begin{theorem}[Folland, Real Analysis, Thm.2.37 \cite{Folland1999}]
\phantomsection
\label{thm:fubini}
    Let $(X, \mathcal{A}, \mu)$ and $(Y, \mathcal{B}, \nu)$ be $\sigma$-finite measure spaces. If $f\geq 0$ is measurable on the product measure $X\times Y$, then
    $$
\int_{X \times Y} f(x,y)\, d(\mu \times \nu)(x,y)
= \int_X \left( \int_Y f(x,y)\, d\nu(y) \right) d\mu(x)
= \int_Y \left( \int_X f(x,y)\, d\mu(x) \right) d\nu(y).
$$
\end{theorem}
Imagine that $g:[a,b]\to[0,\infty)$ is a pdf that is monotonically increasing. The following lemma states that the expected value of $f(x)$ with respect to the distribution $g(x)$ over $[a, b]$ is at least as large as the expected value of $f(x)$ under the worst-case uniform distribution over $[c, b]$, where $a \le c < b$.
\begin{lemma}[Lower Bound by Uniform Distribution]
\phantomsection
\label{lem:inf-bound}
    Let $f:[a,b]\to [0,1]$ be a bounded measurable function (we assume measurable in this paper with respect to the Lebesgue $\sigma-$algebra on $\mathbb{R}^n$, which is complete). Let $g:[a,b]\to [0,\infty)$ be integrable and monotonically increasing. Denote $G(b)-G(a)=\int_a^b g(x)dx$, we have the following inequality:
$$\int_a^b f(x)g(x)dx\geq (G(b)-G(a))\cdot\inf\{\frac{1}{b-c}\int_c^bf(x)dx\;|\; a\leq c<b\}$$
\end{lemma}
\begin{proof}
Define $h:[a,b]\times[0,\infty)\to [0,1]$ by
$$
h(x,y)=
\begin{cases}
f(x), & y\le g(x),\\
0, & \text{otherwise}.
\end{cases}
$$
We write $h=h_1h_2$, where
\begin{align*}
h_1(x,y)&=
\begin{cases}
1, & 0\le y\le g(x),\\
0, & \text{otherwise},
\end{cases}\\
h_2(x,y)&=f(x).
\end{align*}
We show that $h_1$ and $h_2$ are measurable.

\begin{itemize}
    \item $h_1$ is measurable since
    \[
    h_1^{-1}(\{1\})
    =
    \{(x,y)\in [a,b]\times[0,\infty): y\le g(x)\}
    =
    H^{-1}([0,\infty)),
    \]
    where $H(x,y)=g(x)-y$. Since $g$ is measurable, $(x,y)\mapsto g(x)$ is measurable, and $(x,y)\mapsto y$ is the coordinate projection, which is also measurable. Hence $H$ is measurable. Since $[0,\infty)$ is Borel, its preimage is measurable. Therefore $h_1=\mathbf 1_{H^{-1}([0,\infty))}$ is measurable.

    \item $h_2$ is measurable because for any measurable set $A\in\mathcal{L}(\mathbb{R})$,
    \[
    h_2^{-1}(A)=f^{-1}(A)\times [0,\infty),
    \]
    which is measurable in the product $\sigma$-algebra.
\end{itemize}

Therefore $h$ is measurable.  By Fubini-Tonelli's theorem \cref{thm:fubini}, we have 
\begin{align*}
\int_a^b f(x)g(x)\,dx
&=\int_a^b\int_0^{g(x)} f(x)\,dy\,dx\\
&=\int_a^b\int_0^\infty h(x,y)\,dy\,dx\\
&=\int_0^\infty\int_a^b h(x,y)\,dx\,dy.
\end{align*}

For each $y\ge 0$, define
\[
E_y:=\{x\in[a,b]: y\le g(x)\},
\qquad
c(y):=\inf(E_y\cup\{b\}).
\]
Since $g$ is increasing, $E_y$ is an upper interval in $[a,b]$, so it differs from $[c(y),b]$ by at most one point. Hence
\[
\int_a^b h(x,y)\,dx
=
\int_{E_y} f(x)\,dx
=
\int_{c(y)}^b f(x)\,dx.
\]

Now let
\[
m:=\inf\left\{\frac{1}{b-c}\int_c^b f(x)\,dx:\ a\le c<b\right\}.
\]
Then for every $y\ge 0$,
\[
\int_{c(y)}^b f(x)\,dx \ge (b-c(y))\,m.
\]
Therefore,
\begin{align*}
\int_a^b f(x)g(x)\,dx
&=\int_0^\infty \int_{c(y)}^b f(x)\,dx\,dy\\
&\ge \int_0^\infty (b-c(y))\,m\,dy\\
&=m\int_0^\infty (b-c(y))\,dy\\
&=m\int_0^\infty\int_a^b h_1(x,y)\,dx\,dy\\
&=m\int_a^b\int_0^\infty h_1(x,y)\,dy\,dx\\
&=m\int_a^b g(x)\,dx\\
&=(G(b)-G(a))\,m.
\end{align*}
This proves the claim.
\end{proof}

\paragraph{Bounding by uniform distribution in \ranking{} (\cref{claim:ranking_lower_bound_uniform_profiles}).}

In the \ranking{} model, for each $u$, its profile with respect to $\vecx_{\shortminus u^\ast}$ conditioning on a fixed $x_u$ value is either $(x_u,\bot,\bot),(x_u,x_v,\bot)$ or $(x_u,x_v,x_b)$.
In the case where $u$ is unmatched in $\M(\vecx_{\shortminus u^\ast})$. $G(x_u,\bot,\bot)$ is a direct lower bound. If $u$ is matched but does not have a backup, then the distribution of profiles $(x_u,x_v,\bot)$ has an increasing pdf with respect to increasing $x_v\in[0,1]$, then by \cref{lem:inf-bound}, we can lower bound the expected gain by
$$\inf_{v_0<1} \left[\frac{1}{1-v_0}\int_{v_0}^{1} G(x_u,x_v,\bot)dx_v\right].$$
Similarly, if $u$ has a backup $b$ at $b_0$, the distribution of profiles $(x_u,x_v,b_0)$ has an increasing pdf for $x_v\in[0,b_0]$, and we can lower bound the expected gain by
$$\inf_{v_0<b_0} \left[\frac{1}{b_0-v_0}\int_{v_0}^{b_0} G(x_u,x_v,b_0)dx_v\right].$$
Taking a union bound for the three cases yields \cref{claim:ranking_lower_bound_uniform_profiles}.

\paragraph{Bounding by uniform distribution in \Franking{} (\cref{claim:Franking_lower_bound_uniform_profiles_P}).}

In the \Franking{} model, for each fixed $x_u>\theta_1^u$, its profile is either $(x_u,\bot,\bot),\;(x_u,x_v^A,\bot)$, or $(x_u,x_v^A,x_b^A)$.
In the case where $u$ has a profile $(x_u,\bot,\bot)$, $GF(x_u,\bot,\bot)$ is a direct lower bound. If $u$ does not have a backup, the distribution of profiles $(x_u,x_v^A,\bot)$ has an increasing pdf for $x_v\in[0,1]$ by \cref{lemma:monotonicity_Franking}, then by \cref{lem:inf-bound}, we can lower bound the expected gain by
$$\inf_{v_0<1} \left[\frac{1}{1-v_0}\int_{v_0}^{1} GF(x_u,x_v^A,\bot)dx_v\right].$$
Similarly, if $u$ has a backup $b$ at $x_b=b_0$, the distribution of profiles $(x_u,x_v^A,x_b^A)$ has an increasing pdf for $x_v\in[0,b_0]$, and we can lower bound the expected gain by
$$\inf_{v_0<b_0} \left[\frac{1}{b_0-v_0}\int_{v_0}^{b_0} G(x_u,x_v^A,b_0^A)dx_v\right].$$
Taking a union bound for the three cases yields \cref{claim:Franking_lower_bound_uniform_profiles_A}.

\section*{Appendix A}
\phantomsection
\addcontentsline{toc}{section}{Appendix A}
\subsection*{Various Randomized Greedy Matching Algorithms}
\phantomsection
\addcontentsline{toc}{subsection}{Various Randomized Greedy Matching Algorithms}
\label{Sec:various_VI_algorithms}

We briefly summarize several randomized greedy matching algorithms and their known bounds. 

\vspace{1pt}
\noindent\textbf{Greedy.}
Processes edges in an arbitrary deterministic order.
Achieves a tight $0.5$ approximation ratio
\cite{Randomized_greedy_matching,karp1990}.

\vspace{1pt}
\noindent\textbf{\Franking{}.}
Vertices are processed in a fixed deterministic order,
and neighbors according to a common random permutation $\sigma$.
It achieves a $0.521$ lower bound and $0.5671$ upper bound
\cite{0.521Franking}.

\vspace{1pt}
\noindent\textbf{IRP (Independent Random Permutation).}
Vertices are processed in a fixed deterministic order,
while neighbors are considered independently and uniformly at random.
Achieves a tight $0.5$ approximation ratio
\cite{Randomized_greedy_matching}.

\vspace{1pt}
\noindent\textbf{RDO (Random Decision Order).}
Vertices are processed in a uniform random order,
and neighbors in an arbitrary deterministic order.
The approximation ratio lies between $0.531$ and $0.625$
\cite{0.531RDO}.

\vspace{1pt}
\noindent\textbf{\ranking{}.}
Vertices are processed in a uniform random permutation $\pi$,
and neighbors in the same order as $\pi$.
The lower bound is at least $0.5469$
\cite{0.523ranking,0.526ranking,0.505_simplified_Ranking,0.546ranking},
and the upper bound is at most $0.727$
\cite{0.727Ranking}.

\vspace{1pt}
\noindent\textbf{UUR (Uniform-Uniform-Ranking).}
Vertices are processed in a random permutation $\pi$,
and neighbors according to a common random permutation $\sigma$, independent from $\pi$.
Since the bounds for RDO and \Franking{} translate directly to UUR, the current best known lower bound for UUR is $0.531$, extending from RDO~\cite{0.531RDO}, and the upper bound is $0.75$, which is an upper bound that holds for arbitrary VI randomized greedy algorithm~\cite{flawedranking}.

\vspace{1pt}
\noindent\textbf{MRG (Modified Randomized Greedy).}
Vertices are processed uniformly at random,
and neighbors independently and uniformly at random.
The lower bound is at least $0.531$
\cite{randomized_greedy_matching_II,2/3MRG,0.531RDO},
and the upper bound is at most $\frac{2}{3}$
\cite{2/3MRG}.

\subsection*{Fully Online Model and Vertex Arrival Models}
\phantomsection
\addcontentsline{toc}{subsection}{Fully Online Model and Vertex Arrival Models}
\label{Fully-Online_Section}
In \cite{0.521Franking}, the following fully online matching problem model is proposed:
\begin{definition}[Fully Online Model]
\phantomsection
\label{def:fully-online}
Graph $G=(V,E)$ is hidden from the algorithm in the beginning. Each vertex $v\in V$ has an arrival time and a deadline. Upon arrival of vertex $v$, edges between $v$ and previously arrived vertices are revealed. The algorithm needs to match each vertex $v$ before its deadline; otherwise, it cannot be matched afterward. Further, the model assumes that all neighbors of a vertex $v$ will arrive before the deadline of $v$.
\end{definition}
The authors proposed the following algorithm \Fullyranking{}:

\begin{algorithm}[H]
  \caption{\Fullyranking{} algorithm for general graphs.}
  \phantomsection
  \makeatletter\def\@currentlabelname{\textsc{Fully-Ranking}}\makeatother
  \label{alg:Fullyranking}
  \KwIn{Graph $G=(V,E)$}
  \SetAlgoLined

 For each $v$ that arrives, sample an independent uniform random value $x_v \in [0,1]$\;

 At the deadline of vertex $v$, match $v$ to the neighbor $u\in N(v)$ with the smallest $x_u$ value that is available.
\end{algorithm}

The natural intuition for applying \nameref{alg:Franking} to the fully online model is to view the deadlines of the vertices as the adversarial order $\pi$ in which we visit them. However, under this view, we cannot directly adapt \Franking{} to the fully online model, since we do not have all vertices at start to generate the full random preference order.

Nonetheless, observe that \Franking{} does not require the entire permutation explicitly. It only requires the relative order of neighbors when a vertex $v$ makes its decision at its deadline. This relative order is generated by sampling $x_u \in [0,1]$ independently and uniformly at random, and is independent of the arrival order. Therefore, we can generate the ranks $x_u$ upon the arrival of each vertex, as long as all neighbor ranks $x_u$ for $u \in N(v)$ are sampled when \Franking{} makes its matching decision for $v$ at its deadline. This condition is guaranteed by the model assumption that all neighbors are revealed before the deadline.

Therefore, our lower bound for \Franking{} also applies to the fully online model. Indeed, the natural translation of \Franking{} in this setting is precisely \Fullyranking{}. Formally, we show the following two claims:

\begin{claim}
    For any instance of the fully online model, \nameref{alg:Franking} is directly applicable, assuming all vertices arrive at the beginning and viewing the deadlines as the adversarial decision order $\pi$. Further, in this setting, \nameref{alg:Franking} is the same algorithm as \nameref{alg:Fullyranking}.
\end{claim}

\begin{proof}
    Since all vertices arrive at the beginning, \Fullyranking{} will sample all the $x_v$ values altogether. If we view the deadlines as the decision order $\pi$ in \Franking{}, it is easy to see that the two algorithms are the same.
\end{proof}

\begin{claim}
    Let $I_1$ be an instance of the fully online matching problem. Let $I_2$ be the same instance as $I_1$ except that we assume all vertices arrive at the beginning for $I_2$. Then \nameref{alg:Fullyranking} has the same expected performance on $I_1$ and $I_2$.
\end{claim}

\begin{proof}
    Fix an arbitrary realization $\vecx$ of the random ranks of vertices. Since the value of each $x_v$ is independent of the arrival order, $I_1$ and $I_2$ have the same probability of realizing $\vecx$. We prove by induction that \Fullyranking{} has the same output matching for $I_1$ and $I_2$, conditioning on the event that both instances realize $\vecx$.

    Let $\pi$ be the increasing order of deadlines of the vertices. Let $\M^t_i$, $\A^t_i$, and $D^t_i$ for $i\in\{1,2\}$ be the partial matchings, the sets of available vertices\footnote{Vertices that have arrived, being unmatched, and have not reached their deadlines.}, and the sets of vertices that have already passed their deadlines, respectively, right before the deadline of $v$, where $\pi(v)=t$.
    
    \noindent\textbf{Base case.} When $t=1$, no decisions are made, hence $\M^1_1=\M^1_2$.

    \noindent\textbf{Inductive Step.} Assume $\M^t_1=\M^t_2$. Since $I_1$ and $I_2$ have the same deadlines, $D_1^t=D_2^t$. Because we assumed that all neighbors of $v$ arrive before its deadline, we have $N(v)\subseteq A^t_i\cup D^t_i\cup \M^t_i$\footnote{We abuse notation $\M^t_i$ to also denote the set of matched vertices at time $t$.} for both $i\in \{1,2\}$. Since $A^t_i$ is disjoint from $D^t_i\cup\M^t_i$ by definition, we have $N(v)\cap A^t_i=N(v)-(D^t_i\cup\M^t_i)$. Further, since $D^t_1=D^t_2$ and $\M^t_1=\M^t_2$, we have $N(v)\cap A^t_1=N(v)\cap A^t_2$, which means $v$ has the same set of available neighbors to choose from for both $I_1$ and $I_2$ at time $t$. Since we conditioned on both instances realizing $\vecx$, $v$ will pick the same neighbor to match to, which implies $\M^{t+1}_1=\M^{t+1}_2$.
\end{proof}

Hence, for any instance of the fully online matching problem, running \Fullyranking{} is equivalent to running \Franking{} on the same graph instance, assuming all vertices arrive at the beginning. Thus, this extends our bound for \Franking{} to the fully online matching problem.
\paragraph{Vertex Arrival Models}
One common assumption that is used for vertex arrival models is that upon arrival of a vertex, only edges incident with previously arrived neighbors are revealed. In this paper, we assume all incident edges are revealed, regardless of whether the neighbor has arrived or not. 

By the setting of the fully online (adversarial vertex arrival) model, it makes sense to assume all neighbors are revealed at the deadline (decision) time of each vertex. If an edge $(u,v)$ has $u$ arriving after the deadline of $v$, then the algorithm cannot match them; effectively, $(u,v)$ does not exist. For the uniform random vertex arrival model, \cite{0.546ranking} proved that \ranking{} has the same expected performance under both assumptions, so W.L.O.G., we can assume all neighbors are revealed.
\section*{Appendix B}
\phantomsection
\addcontentsline{toc}{section}{Appendix B}
\label{sec:appendix-a}

\begin{proof}[\textbf{Proof for \cref{assumeperfectmatching}}]
    For any list $L$ of query orders and for any vertex $v\in G$, by the alternating path lemma (\cref{lem:alt-path}), we always have $|\M(L)|\geq |\M(L_{\shortminus v})|$. Hence, for any randomized greedy matching algorithm $A$, let the randomization be by drawing a list $L\sim D_{A}$ for some distribution $D_A$ of query lists. We have, for any $v\in V$,
    $$\ev_{L\sim D_A }\left[|\M(L)|\right]\geq \ev_{L\sim D_A }[|\M(L_{\shortminus v})|].$$
    This means that removing a vertex will not result in an increase in the size of output matching by $A$. For any $G$ that does not admit a perfect matching, we can fix a maximum matching $M^\ast$ and remove any $v$ that does not belong to the maximum matching $M^\ast$ one by one. The removal does not decrease the size of the maximum matching and will not increase the expected size of the output matching by $A$. Thus sequentially removing $v$ does not increase the approximation ratio of $A$ when we prune $G$ until $G$ only contains vertices in $M^\ast$.
\end{proof}

\begin{proof}[\textbf{Proof for \cref{lem:alt-path}}]
    We prove by induction on time $t$. Initially, nothing is queried, and the path $\M^t(L)\oplus\M^t(L_{\shortminus v})$ can be viewed as a degenerate path of a single node $u_0=v$. We also have $\A^t(L)\oplus \A^t(L_{\shortminus v})\cup\{v\}$ by definition.

    Assume the alternating path lemma holds at time $t=(u,w)$. Denote $u_0,...,u_k$ as the alternating path  $\M^{t}(L)\oplus\M^{t}(L_{\shortminus v})$. Let $t^+$ be the time when the next pair of nodes is ready to be queried. If both $u,w$ are not $u_k$, then the availability of $\{u,w\}$ is the same in $\A^t(L)$ and $A^t(L_{\shortminus v})$ by the alternating path lemma. This implies either both $\M^{t^+}(L),\M^{t^+}(L_{\shortminus v})$ include $(u,w)$ or both do not include it. Hence, the alternating path does not change from $t$.
    
    If one of $u,w$ is $u_k$, assuming W.L.O.G. $u=u_k$, in the case where $k$ is even, by inductive hypothesis $\A^t(L)= \A^t(L_{\shortminus v})\cup\{u\}$. If $w\not\in \A^t(L)$ then $w\not\in \A^t(L_{\shortminus v})$ in which case both matchings do not include $(u,w)$. The only interesting case is when $w$ is available for both lists at time $t$ but $u$ is available only in $\A^t(L)$, in which case $(u,w)$ gets included in $\M^{t^+}(L)$ but not in $\M^{t^+}(L_{\shortminus v})$. And the alternating path extends by $(u,w)$ where $u=u_k$, $w=u_{k+1}$, and we have $(u_k,u_{k+1})\in \M(L)$. Further, the set of available vertices has 
    $$\A^{t^+}(L)=\A^t(L)-\{u,w\}=\A^t(L_{\shortminus v})\cup\{u\}-\{u,w\}=\A^{t^+}(L_{\shortminus v})-\{w\}$$ 
    Thus the inductive hypothesis holds at $t^+$ when $k$ is even. For the case where $k$ is odd, the proof is symmetric with $L$ and $L_{\shortminus v}$ swapped.
\end{proof}
\begin{proof}[\textbf{Proof for \cref{fact:backup_uniquely_describe_worse_off}}]
    We will prove that, roughly speaking, removing or introducing $v$ is effectively removing $u_1$ from the perspective of $u$. If this is true, then following the definition of backup, $u$'s matching condition is as described by its backup.

    Assume $u$ is matched to $u_1$ in $L$, and assume $u$ is made worse off by removing $v$. Let $t=(u,u_1)$ be the time edge $(u,u_1)$ gets queried. By assumption $u$ is made worse off by removing $v$ implies that $u$ is available at time $t$ but $u_1$ is not in $\A^t(L_{\shortminus v})$. By the alternating path lemma, we know that in this case $\A^t(L)=\A^t(L_{\shortminus v})\cup\{u_1\}$. Also notice that $(u,u_1)\in \M(L)$, which implies both $u,u_1$ are available in $\A^t(L)$. Hence we also have $\A^t(L)=\A^t(L_{\shortminus u_1})\cup\{u_1\}$. These two facts imply $\A^t(L_{\shortminus u_1})=\A^t(L_{\shortminus v})$. Since both lists have the same set of available vertices, $\M(L_{\shortminus v})$ and $\M(L_{\shortminus u_1})$ will include the same set of edges from $t$. This implies that $u$ has the same matching condition in $\M(L_{\shortminus v})$ as in $\M(L_{\shortminus u_1})$.

    Assume $u$ is matched to $u_1$ in $L_{\shortminus v}$, and assume $u$ is made worse off by adding $v$. Let $t=(u,u_1)$ be the time edge $(u,u_1)$ gets queried. By assumption $u$ is made worse off by introducing $v$ implies that $u$ is available at time $t$ but $u_1$ is not in $\A^t(L)$. By the alternating path lemma, we know that in this case $\A^t(L_{\shortminus v})=\A^t(L)\cup\{u_1\}$. Also notice that $(u,u_1)\in \M(L_{\shortminus v})$, which implies both $u,u_1$ are available in $\A^t(L_{\shortminus v})$. Hence we also have $\A^t(L_{\shortminus v})=\A^t(L_{\shortminus vu_1})\cup\{u_1\}$. These two facts imply $\A^t(L)=\A^t(L_{\shortminus vu_1})$. Since both lists have the same set of available vertices, $\M(L)$ and $\M(L_{\shortminus vu_1})$ will include the same set of edges from $t$. This implies that $u$ has the same matching condition in $\M(L)$ as in $\M(L_{\shortminus vu_1})$.
\end{proof}

\begin{proof}[\textbf{Proof for \cref{lem:victim_querycommit_case_u}}]
    Assume $u$ has a match $w$ in $\M(L_{\shortminus v})$ and $u$ is unmatched in $\M(L)$. Let $u_0,...u_k$ be the alternating path $\M(L_{\shortminus v})\oplus\M(L)$. Let $u_{i}\neq u$ be an even-indexed node in the path. We want to show that $u$ is matched in $\M(L_{\shortminus u_i})$. Let $t$ be the time when $(u_i,u_{i+1})$ gets queried, by the alternating path lemma, we know that $\A^t(L)=\A^t(L_{\shortminus v})\cup\{u_i\}$ (at this time, $(u_i,u_{i+1})$ is yet to be included in $\M(L)$, so the alternating path terminates at $u_i$ at time $t$). Note that we also have $\A^t(L)=\A^t(L_{\shortminus u_i})\cup\{u_i\}$. This implies $\A^t(L_{\shortminus v})=\A^t(L_{\shortminus u_i})$, and hence both matchings will include the same set of edges in the future. Since $u$ is matched in $\M(L_{\shortminus v})$, it will be matched in $\M(L_{\shortminus u_i})$.
\end{proof}
\begin{proof} [\textbf{Proof for \cref{lem:victim_querycommit_case_u_star}}]
     Let $v,L,L'$ be as stated in the lemma. Denote the match of $u$ in $\M(L_{\shortminus v})$ as $w$ and the backup as $b$ (assuming $u$ has a backup in $L_{\shortminus v}$). Let $u_0,...u_{j},...,u_k$ be the alternating path $\M(L)\oplus \M(L_{\shortminus v})$. Let $u_i$ be an odd-indexed vertex in the path where $i<j$. We want to show that $v$ is matched in $L'_{\shortminus u_i}$. Let $t$ be the time when $(u_i,u_{i+1})$ gets queried in $L$, by the alternating path lemma, $\A^t(L_{\shortminus v})=\A^t(L)\cup\{u_i\}$ (because $i$ is odd-indexed). As $u_i$ is not matched before $t$ in $\M^t(L_{\shortminus v})$, it is effectively non-existing until $t$ for $L_{\shortminus v}$. This implies $\A^t(L_{\shortminus u_iv})=\A^t(L_{\shortminus v})-\{u_i\}=\A^t(L)$. This implies that from time $t$, $\M(L)$ and $\M(L_{\shortminus u_iv})$ will include the same set of vertices, one of which is $(u,b)$ (which is edge $(u_j,u_{j+1})$ in the alternating path, included at time later than $t$).  Notice that $L'$ differs from $L$ only in edges incident to $v$ by definition. So we have $L_{\shortminus u_iv}$ equivalent to $L'_{\shortminus u_iv}$ for greedy matching, i.e., $\M(L_{\shortminus u_iv})=\M(L'_{\shortminus u_iv})$. Since $u$ is matched to $b$ in $\M(L_{\shortminus u_iv})$, it is also matched to $b$ in $\M(L'_{\shortminus u_iv})$. Since we assumed that $(u,v)$ has a higher priority than $(u,b)$ in $L'$, we know that if $v$ is still unmatched in $\M(L'_{\shortminus u_i})$ before $(u,v)$ gets queried, $u$ will be available (as it will continue to be available until $(u,b)$ gets queried, which is at a later time). Hence $(u,v)$ will be matched in $\M(L'_{\shortminus u_i})$ in this case.

    For the case $u$ not having a backup $b$, we can argue by the same logic and show that $u$ is unmatched in $\M(L'_{\shortminus u_iv})$ hence no matter when $(u,v)$ gets queried in $\M(L'_{\shortminus u_i})$, $u$ will be available provided that $v$ is unmatched. In such a case $(u,v)$ will be matched.
\end{proof}

\begin{proof}[\textbf{Proof for \cref{claim:victim_when_no_alt_path}}]
    We want to show that $v$ is matched in $\M(L_{\shortminus w}')$. Let $t$ be the time $(u,v)$ gets queried, if $v$ is already matched, then we are done. If not, then the alternating path $\A^t(L_{\shortminus w}')=\A^t(L'_{\shortminus vw})\cup\{v\}$ by the alternating path lemma. Notice that $u$ is still available at time $t$ in $\A^t(L'_{\shortminus w})$, as in the case of $u$ not having a backup, $u$ is available until the end; in the case of $u$ having a backup $b$, $u$ is available until $(u,b)$ gets queried, which is at a time later than $t$. Hence $(u,v)$ will be matched.
\end{proof}

\begin{proof}[\textbf{Proof for \cref{fact:matching_guarantee_ranking}}]
    \textbf{1.} Assume $x<x_v$. Let $t$ be the time when $(u,u^\ast)$ gets queried. If $u^\ast$ is still available, then the alternating path lemma implies $\A^t(\vecx_{\shortminus u^\ast})=\A^t(\vecx_{\shortminus u^\ast}(x))\cup\{u^\ast\}$. Since $t$ is earlier than $(u,v)$, $u$ is still available in both sets and hence $u,u^\ast$ will be matched in $\M^t(\vecx_{\shortminus u^\ast}(x))$.
    
    \textbf{2.} Assume $u$ becomes worse off by introducing $u^\ast$ at rank $x$, then $u$ must be an even-indexed vertex $u_{2i}$ in the alternating path $\M(\vecx_{\shortminus u^\ast})\oplus \M(\vecx_{\shortminus u^\ast}(x))$. By the special monotonicity property for \ranking{} (\cref{fact:monotonicity_for_alt_path_ranking}), we have $x=x_{u_0}<x_{u_{2i}}=x_{u^\ast}$. So when $x>x_u$, $u$ cannot be an even-indexed vertex in the alternating path and hence cannot become worse off, i.e., matching to a vertex of rank larger than its current match $x_v$.
    
    \textbf{3.} Similarly, if $u$ is made worse off by introducing $u^\ast$, then $u$ is an even-indexed vertex and $v$ is an odd-indexed vertex. The match of $u^\ast$ is $u_1$ in the alternating path, the smallest-ranked odd vertex, hence $x_{u_1}\leq x_v$ by the special monotonicity property (\cref{fact:monotonicity_for_alt_path_ranking}).
\end{proof}

\begin{proof}[\textbf{Proof for \cref{fact:randomized_bipartition}}]
For each fixed $\vecx$, since $\M(\vecx)$ and $M^\ast$ are both matchings, their union does not contain odd-length cycles, which implies $\M(\vecx)\cup M^\ast$ is a bipartite graph. Pick an arbitrary bi-partition $U_1(\vecx),\;U_2(\vecx)$ of $\M(\vecx)\cup M^\ast$, where $U_1(\vecx),\;U_2(\vecx)$ are the two sides. Let $\chi[\vecx]:V\to\{P,B\}$ be the bi-partition that assigns the role $P$ to all vertices in $U_1(\vecx)$ and assigns the role $B$ to all vertices in $U_2(\vecx)$. Let $\chi[\vecx]^r$ be the reversed assignment, i.e. assigning $B$ to $U_1(\vecx)$ vertices and assigning $P$ to $U_2(\vecx)$ vertices. We define $\chi:V\to \{P,B\}$ as the random bi-partition such that the joint distribution of $\chi,\vecx$ satisfies 
\begin{align*}
\text{pdf}(\chi, \vecx) =
\begin{cases}
\frac{1}{2} p(\vecx), & \text{if } \chi = \chi[\vecx], \\[4pt]
\frac{1}{2} p(\vecx), & \text{if } \chi = \chi[\vecx]^r, \\[4pt]
0, & \text{otherwise,}
\end{cases}
\end{align*}
where $p(\vecx)$ is the pdf of $\vecx$. It's not hard to see that, by the definition, gain-sharing schemes 1 and 2 are well-defined with probability $1$. Let $E$ be the event 
$$PD(u^\ast)\;\wedge\; BU(u) \;\wedge\; PD(\text{match of $u$})\; \wedge \; BU(\text{match of $u^\ast$}).$$ 
The conditional pdf of $\vecx$ on $E$ is the same as that of the uniform random rank vector $\vecx$, as conditioning on $E$ is equivalent to conditioning on the event $PD(u^\ast)$ alone, which is independent of the rank vector $\vecx$. Let $\alpha$ be a lower bound for the expected value of gain$(u)$ + gain$(u^\ast)$ under gain-sharing rules $1,\;2$ for any pairs of vertices $(u,u^*)\in M^*$, conditioning on event $E$. We have
\begin{align*}
    \alpha\leq&\min_{(u,u^*)\in M^*}\{\ev_{\vecx}[\text{ gain$(u)$ + gain$(u^\ast)$}\;|\; E\;]\}\\
    \leq &\frac{\sum_{(u,u^\ast)\in M^\ast}\Big[\frac{1}{2}\ev_{\vecx}[\text{ gain$(u)$ + gain$(u^\ast)$ }|PD(u^*)]+\frac{1}{2}\ev_{\vecx}[\text{ gain$(u)$ + gain$(u^\ast)$ }|PD(u)]\Big]}{|M^*|}\\
    =&\frac{\ev_{\vecx}[\sum_{v\in V} \text{ gain$(v)$} ]}{|M^\ast|}= \frac{\ev_{\vecx}[|\M(\vecx)|]}{|M^\ast|}.\qedhere
\end{align*}
\end{proof}

\begin{proof}[\textbf{Proof for \cref{fact:matching_guarantee_franking}}]
     Property 1: Let $t$ be the time $(u,u^\ast)$ gets queried. Assume $u^\ast$ is not matched at time $t$. Then the partial matching $\M^t(\vecx_{\shortminus u^\ast})$ equals the partial matching $\M^t(\vecx_{\shortminus u^\ast}(x))$.  Hence $\A^t(\vecx_{\shortminus u^\ast}(x))=\A^t(\vecx_{\shortminus u^\ast})\cup\{u^\ast\}$ by the alternating path lemma (\cref{lem:alt-path}). Since $x<x_v$ and $u$ actively matches $v$ in $\M(\vecx_{\shortminus u^\ast})$, we know that $(u,u^\ast)$ is queried before $(u,v)$ and hence $u$ is available at time $t$. In this case, $u,u^\ast$ will be matched. Since $u$ has an earlier deadline, $u^\ast$ is the passive vertex.
    
     Property 2: If $u^\ast$ is matched actively or unmatched, then by the time $(u,v)$ gets queried, $u^\ast$ is still not matched and has no influence on the partial matchings $\M^t(\vecx_{\shortminus u^\ast}),\M^t(\vecx_{\shortminus u^\ast}(x))$. This implies $(u,v)$ will be included in both final matchings.

     Property 3: Assume $u^\ast$ is matched passively to $w$ in $\M(\vecx_{\shortminus u^\ast}(x_0))$. Fix $x\in [0,x_0]$. If $u^\ast$ is not matched in the partial matching $\M^t(\vecx_{\shortminus u^\ast}(x))$ until $w$ arrives at time $t$, then $w$ will have the same list of available neighbors as in $\vecx_{\shortminus u^\ast}(x_0)$ at time $t$. Since $u^\ast$ with rank $x_0$ was the best available neighbor of $w$ in the list before, $u^\ast$ with rank $x\leq x_0$ will still be the best available neighbor, hence $u^\ast$ will be passively matched to $w$ in this case.
\end{proof}

\begin{proof}[\textbf{Proof for \cref{lemma:monotonicity_Franking}}]
    Similar to the proof of \cref{monoto_before_backup_ranking}, it suffices to show that $(u,v)$ are still matched together after demoting $v$ to rank $x$.

    Assume the profile of $u$ with respect to $\vecx_{\shortminus u^\ast}$ is $(x_u,x_v^A,\bot)$. Suppose we demote $v$ to rank $x>x_v$ and denote the permuted rank vector as $\vecx_{\shortminus u^\ast}^{\;\prime}$. Since the rank of $v$ is demoted, it will not become more favorable to any active vertex, which means none of the active vertices before $u$ will change their choices. Since we also know $u$ does not have a backup, which means at the time $u$ actively picks, it has no other available neighbor besides $v$. The two facts combined show that $u$ will still actively match $v$ in $\M(\vecx_{\shortminus u^\ast}^{\;\prime})$.

    A similar proof holds for the case $(x_u,x_v^A,x_b^A)$, now for any $x$ such that $x_v<x<x_b$, if we demote the rank of $v$ to $x$, $v$ will not become more favorable to any active vertex before $u$, hence the matching until $u$ is not changed. When it comes to $u$’s turn to pick, the second-best choice for $u$ is $b$, which has a larger rank $x_b$ than $x$. Hence $u$ will still be matched to $v$ in $\M(\vecx_{\shortminus u^\ast}^{\;\prime})$.
\end{proof}
\bibliographystyle{plain}
\bibliography{references}
\end{document}